\documentclass[poms,final,nonblindrev]{poms1_V1_arxiv} % Use this for listing out authors' information submissions

\OneAndAHalfSpacedXI % % default line spacing !!!!!!!!!!!!!!!!!!!!!!!
\usepackage[hypertexnames=false]{hyperref}
\usepackage{lmodern}
\usepackage{booktabs} % For formal tables
\usepackage[ruled,algo2e]{algorithm2e} % For algorithms

\SetAlFnt{\small}
\SetAlCapFnt{\small}
\SetAlCapNameFnt{\small}
\SetAlCapHSkip{0pt}
\IncMargin{-\parindent}

\usepackage{xurl}
\usepackage{graphicx}
\usepackage{algorithmicx}
\usepackage{algorithm}
\usepackage{algpseudocode} 
\usepackage{subfigure}
\usepackage[showdeletions]{color-edits}
\usepackage{caption}
\usepackage{paralist}
\usepackage{tikz-network}
\usetikzlibrary{decorations.pathreplacing}
\usepackage[noabbrev, capitalise]{cleveref}
\usepackage{enumitem}
\usepackage{paralist}
\usepackage{booktabs}
\usepackage{tcolorbox}
\usepackage{xcolor}

\definecolor{rose}{HTML}{C0392B}

\addauthor{mp}{red}
\addauthor{ar}{blue}
\addauthor{r}{black}

\newcommand{\one}{\mathbf 1}
\newcommand{\zero}{\mathbf 0}

\newcommand{\linktoproof}[1]{\begin{center} \hyperref[#1]{\texttt{[Link to Proof]}}\end{center}}
\newcommand{\linktostatement}[1]{\begin{center} \hyperref[#1]{\texttt{[Link to Statement]}}\end{center}}
\newcommand{\ev}[2]{\mathbb E_{#1}\left [ #2 \right ]}

\newcommand{\shrink}[0]{\vspace{-1em}}
\renewcommand{\Pr}{\mathbb P}
\newcommand{\xpar}[1]{\noindent \textbf{#1.}}

\newtheorem{mdefinition}{Definition}

\newcommand{\cC}{\mathcal{C}}
\newcommand{\cD}{\mathcal{D}}
\newcommand{\cE}{\mathcal{E}}

\newcommand{\cG}{\mathcal{G}}

\newcommand{\cK}{\mathcal{K}}

\newcommand{\cN}{\mathcal{N}}

\newcommand{\cR}{\mathcal{R}}
\newcommand{\cS}{\mathcal{S}}

\newcommand{\cU}{\mathcal{U}}

\newcommand{\Rbb}{\mathbb{R}}

\newcommand{\Ravgsamp}[1]{\hat R^{\mathsf{avg}}_{#1}}

\newcommand{\katz}{\mathsf {Katz}}
\newcommand{\rei}{\mathrm {REI}}
\newcommand{\reiavg}[1]{\rei^{\mathsf{avg}}_{#1}}
\newcommand{\potentialimpact}{\mathrm{PI}}

\newcommand{\Oeps}[1]{O_{\varepsilon} \left ( #1 \right )}
\newcommand{\Omegaeps}[1]{\Omega_{\varepsilon} \left ( #1 \right )}

\renewcommand{\bi}[1]{\textbf{\textit{#1}}}

\usepackage{graphicx}
\usepackage{subfigure, epsfig}
\usepackage{natbib}

 \bibpunct[, ]{(}{)}{,}{a}{}{,}%
 \def\bibfont{\small}%
\TheoremsNumberedThrough     % Preferred (Theorem 1, Lemma 1, Theorem 2)
\ECRepeatTheorems

\EquationsNumberedThrough    % Default: (1), (2), ...
\crefname{appendix}{Appendix}{Appendices} % Defines "Appendix" for singular/plural
\begin{document}
%%%%%%%%%%%%%%%%

%\setcounter{page}{1}
%\DOI{10.1287/xxxx.0000.0000}%

\RUNAUTHOR{Papachristou, Rahimian, and Azadegan}

\RUNTITLE{Structural Measures of Resilience for Supply Chains}

\TITLE{Structural Measures of Resilience for Supply Chains}

\ARTICLEAUTHORS{%
\AUTHOR{Marios Papachristou}
%\AFF{Department of Computer Science, Cornell University, \EMAIL{papachristoumarios@gmail.com}, \URL{}}
\AFF{W.P. Carey School of Business, Arizona State University, \EMAIL{mpapachr@asu.edu}}
\AUTHOR{M. Amin Rahimian}
\AFF{University of Pittsburgh, \EMAIL{rahimian@pitt.edu}}
\AUTHOR{Arash Azadegan}
\AFF{Rutgers Business School, \EMAIL{aazadegan@business.rutgers.edu}}

% \AUTHOR{John Doe}
% \AFF{Department of Operations Management, Operations University, \EMAIL{jdoe@operations.edu}} %, \URL{}}
% \AUTHOR{Jane Doe}
% \AFF{Institute for Supply Chain Management, University of Management \EMAIL{jdoe@iscm.edu}}
} % end of the block

\ABSTRACT{%
Modern production systems are increasingly defined by dense networks of multi-tier sourcing dependencies, where localized upstream disruptions can cascade into system-wide collapses. While supply chain resilience has garnered significant managerial attention, we still lack theoretically-grounded, reliable, analytical metrics that can distinguish inherently resilient architectures from fragile ones. This paper addresses this gap by developing a structural resilience framework and a novel metric, defined as the maximum supplier failure rate that a network can sustain while maintaining an aggregate production level. Using node percolation theory and branching processes, we identify four critical structural determinants of resilience: the number of raw materials, the number of finished goods, sourcing requirements, and sourcing influence. Our analysis reveals two distinct regimes: ``top hat'' architectures, which are characterized by excessive raw materials and high centralization, making them inherently fragile; and ``rolling pin'' structures, which maintain controlled input/output widths and sparsity, allowing them to absorb non-trivial shocks. To operationalize these insights, we formulate resilience computation as a scalable linear program that approximates cascading failure sizes in large-scale networks with cycles, heterogeneous suppliers, and structural decoupling. Furthermore, we extend our framework to account for exogenous failure correlations, such as those arising from geographic or geopolitical factors that can undermine traditional supplier and input diversification strategies. We validate our theoretical results using multi-echelon supply chain data. These tools can inform network design, supplier diversification, and inventory planning to proactively reduce systemic risk.
% Enter your abstract
}%

%Use this for final submission
%\HISTORY{This paper was first submitted on January 1, 2021 and has been with the authors for 3 months for 2 revisions.}

\maketitle

\vspace{-20pt}

\section{Introduction} \label{sec:introduction}

Modern production systems rely on dense networks of interdependent suppliers and inputs that collectively determine whether downstream firms can continue to operate. These networks span raw-material extraction, chemical and component manufacturing, and finished goods assembly; potentially linking thousands of firms and products across multiple tiers. Although these networks allow firms to specialize and scale efficiently, they also create systemic exposure, where small upstream failures can spread through many tiers and trigger disproportionate downstream production loss \citep{melnyk2014understanding}. The 2012 Evonik disruption \citep{evonik} illustrates this vulnerability concretely (\cref{fig:evonik}).

% An explosion at the firm's facility in Marl, Germany, halted production of cyclododecatriene (CDT), the essential precursor to Nylon 12, a resin used globally in automotive fuel and brake lines. Because Evonik and Arkema together accounted for roughly half of global Nylon~12 capacity, and because CDT was effectively single-sourced, the incident created an immediate system-wide risk. Automakers and Tier-1 suppliers warned that component production could stop within weeks, and industry groups convened emergency meetings to evaluate inventories, substitution options, and potential production shutdowns \citep{evonik}. In effect, the loss of a single chemical input threatened to cascade throughout an entire global manufacturing ecosystem, raising the question of how much disruption such networks can absorb before aggregate production collapses. In practical terms, this gives managers a way to look at their network structure and ask a simple question: \textit{is this design fundamentally stable, or does it become fragile under even modest upstream disruption?}

Additionally, when Hurricane Maria struck Puerto Rico in 2017, it devastated a hub of U.S. medical manufacturing. Baxter International, the largest supplier of IV saline bags to U.S. hospitals, had three major plants on the island -- one was on backup generators for weeks, and others faced intermittent power supply. This single geographic disaster led to a nationwide shortage of IV fluids used for basic hospital care. Hospitals across the mainland resorted to rationing IV bags and using substitute methods (oral rehydration, syringes) until supply slowly recovered. This case highlights geographic and supplier concentration risk in supply networks. A localized shock cascaded into a system-wide disruption because a critical input (saline bags) was sourced predominantly from one vulnerable region \citep{aguero_weathering_2024}.

% \begin{figure}[t]
%     \centering
%     \includegraphics[width=0.55\linewidth]{figures/example_slide.pdf}
%     \caption{Production network corresponding to the Evonik incident report extracted from \citet{evonik}. There are $K = 11$ products in total. The present network has $r = 2$ raw materials (Butadiene, CDT), $c = 3$ finished goods (Automobiles, Electronics, Pharmaceuticals), a sourcing requirement $m = 4$, which is determined by the product that has the maximum number of inputs (here,  Automobiles have the maximum number of inputs which equals to 3), and a sourcing influence $\mu = 3$ which is determined by the product that has the maximum number of outputs. Here, Nylon-12 sources the maximum number of products in the network, corresponding to three tier-1 components. Plastic parts are also sourced to three original equipment manufacturers (OEMs) in tier 0. The source failure of Evonik happens at the CDT node (colored red), and spreads to all components that depend on it (colored pink) for production. Due to the Evonik failure, eight of the $11$ products are affected, and three remain unaffected (colored light blue). Our resilience measure requires at least 90.9\% of the products to be operational, which is not satisfied in this case, even for infinitesimally small shocks, deeming the network fragile. Our analysis reveals how factors such as the sourcing requirement and influence affect structural resilience in a quantifiable manner.}
%     \label{fig:evonik}
% \end{figure}

\begin{figure}[t]
    \centering
    \includegraphics[width=0.55\linewidth]{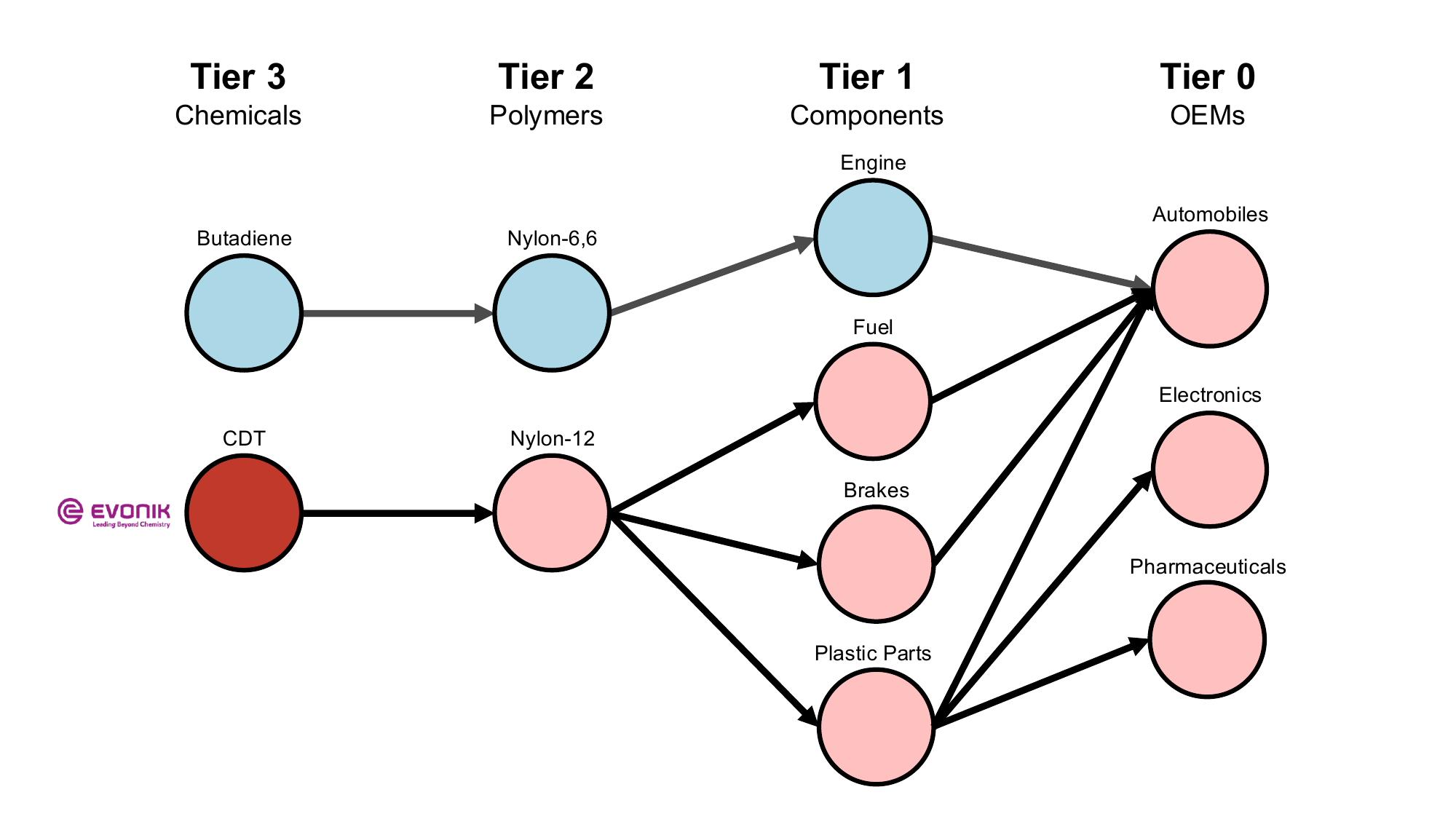}
    \caption{Production network corresponding to the Evonik incident \citep{evonik}. An explosion at Evonik's facility in Marl, Germany halted production of cyclododecatriene (CDT), the essential precursor to Nylon 12, a resin used globally in automotive fuel and brake lines. Because Evonik and Arkema together accounted for roughly half of global Nylon 12 capacity and CDT was effectively single-sourced, the incident created an immediate system-wide risk. Automakers and Tier-1 suppliers warned that component production could stop within weeks, and industry groups convened emergency meetings to evaluate inventories, substitution options, and potential production shutdowns. There are $K = 11$ products in total. The network has $r = 2$ raw materials (Butadiene, CDT), $c = 3$ finished goods (Automobiles, Electronics, Pharmaceuticals), a sourcing requirement $m = 4$, determined by the product with the maximum number of inputs (Automobiles, with 3 inputs), and a sourcing influence $\mu = 3$, determined by the product sourcing the maximum number of outputs (Nylon-12, supplying three tier-1 components; Plastic Parts also sources three OEMs in tier 0). The failure of Evonik occurs at the CDT node (red), spreading to all components depending on it (pink). Eight of the 11 products are affected; three remain unaffected (light blue). Our resilience measure requires at least 90.9\% of products to be operational; a condition not satisfied here even for infinitesimally small shocks, deeming the network fragile. This illustrates how factors such as sourcing requirements and sourcing influence affect structural resilience in a quantifiable manner.}
    \label{fig:evonik}
\end{figure}

Events such as Evonik and the 2017 U.S. IV shortage incidents underscore a broader challenge: understanding when a supply network is \textit{structurally} predisposed to cascading failures and how large a supplier shock it can absorb while preserving aggregate production~\citep{ergun2022structured}. 

Although resilience has received substantial managerial attention; see, for instance,  \citet{chen2023when,melnyk2014understanding,ivanov2020viability}, there remains a lack of theoretically-grounded and analytical metrics that characterize when production networks are inherently resilient or fragile. As the complexity of modern supply chains increases, analytically quantifying the rate at which they become fragile (or how resilient they are) and ranking the various supply chain architectures to derive insights is a complex and ever-challenging problem.

The existing contagion literature offers powerful tools for analyzing network resilience, but they often pose several limitations. Firstly, they do it from a static point \citep{kim2015supply,kleindorfer2005managing,craighead2007severity}, thus omitting collateral contagion effects, or assume single points of failure \citep{levi2016identifying,simchi2014superstorms} that preclude dynamic interactions and broader network effects. Secondly, many existing approaches rely on simulation to study specific disruption scenarios, offering limited systemic guidance on structural limits across network architectures \citep{nair2011supply,thadakamaila2004survivability,perera2017network,levi2016identifying} or give analytical results for simplistic network structures such as tree-like networks \citep{elliott2022supply}, preferential attachment networks \citep{albert_error_2000}, and have viewed the concept of resilience at the supplier level; see, for example, \citet{chen2023when,snyder2016or,kouvelis2013robust}---and not from the perspective of supply. 

This paper addresses these gaps by developing a structural, theoretically-grounded framework that quantifies \textit{how much supplier failure} a production network can tolerate before aggregate supply collapses. Our central contribution is a resilience metric that analytically captures the maximum supplier failure rate that a network can sustain while keeping almost all production operational. Our metric concerns the topology of the underlying network and provides a well-defined, practical diagnostic for supply chains. The resilience metric allows us to compare different networks to each other based on the rate at which the metric goes to zero (indicating fragility) or stays bounded away from it. Finally, the resilience metric allows us to account for correlations that are exogenous to the network structure, i.e, going beyond correlations induced by the supply network.

Building on this metric, we theoretically identify four structural determinants of resilience that govern whether networks absorb or amplify disruptions: the number of raw materials ($r$), the number of finished goods ($c$), the sourcing requirement ($m$), and the sourcing influence ($\mu$). These characteristics determine whether a network is fundamentally \emph{resilient} or \emph{fragile}. Networks with an excessive number of raw materials relative to their size resemble a ``top hat'', and are fragile because any upstream failure can cascade widely. By contrast, networks with controlled input and output widths and sparsity, i.e., networks that resemble ``rolling pins'', can absorb non-trivial shocks and are, thus, resilient (\cref{fig:qualitative_representation}).

To provide an operationalizable and more fine-grained result, we formulate resilience computation as a linear program that approximates the size of cascading failures in general networks, including those with cycles, heterogeneous suppliers, or structural decoupling. This tool enables rapid \textit{``what-if''} analysis; firms and managers can estimate the resilience of these networks, quantify their respective risk, and take appropriate safeguarding measures. 

This paper proceeds by first formalizing production networks and introducing a structural resilience metric, then characterizing resilient and fragile architectures, and finally extending the framework to more general settings with added correlations and heterogeneity. 

\subsection{Managerial Examples}

Managers often assess supply risk by focusing on how inputs are sourced at the firm level, even though system-wide behavior is shaped by network structures that are not visible locally; see \citet{choi2001a}.

In a common situation, a firm sources a critical component through multiple tier-one sourcing channels located in different regions. Each channel appears independent and reliable when evaluated individually. During a disruption, all channels fail simultaneously because they depend on the same specialized input in the upstream. From the firm’s perspective, there was no visible single failure point, yet production halted. This example illustrates how exposure can persist even when risk mitigation appears sensible at the firm level \citep{tang_perspectives_2006}. 

A similar pattern emerges as firms expand product portfolios. To spread demand risk and improve utilization, manufacturers add new products, each supported by its own sourcing configuration. Over time, products may become dependent on overlapping inputs, shared processes, or common upstream facilities. A disruption that once would have affected a narrow product line can then affect a much larger share of the output. This example aligns with evidence that upstream structural complexity is associated with a higher frequency of disruption \citep{bode_structural_2015}.

Hidden coupling can also arise across levels. A downstream assembler maintains redundancy and inventory buffers in the immediate sourcing layer. However, several of those sourcing channels may depend on the same upstream processing stage for a critical operation. When that processing stage experiences a short outage, the impact spreads across multiple inputs and downstream products, quickly overwhelming the safeguards that were designed using firm-level risk assessments. More broadly, previous work emphasizes that resilience depends not only on local buffers but also on how firms develop relational capabilities across their sourcing networks \citep{toyli_harri_lorentz_and_lauri_ojala_influence_2013}.

These situations illustrate how vulnerability can arise from the structure of production and sourcing networks rather than from isolated decisions or specific disruption events.

% \subsection{Paper Organization}

% The remainder of this paper is organized as follows. \cref{sec:related_work} provides an overview of the related work on  resilience in the operations and supply chain management literature. In \cref{sec:preliminaries}, we introduce our formal structural model and the novel resilience metric. \cref{sec:main_results} provides a detailed look at the structural determinants of resilience, identifying the architectural regimes and their respective impacts on resilience. \cref{sec:extensions} examines the robustness and extensions of the resilience regime, such as operationalizing resilience with linear programming,  a connection to the Risk Exposure Index (REI), which is a widely-used supply chain resilience index, and a general model to measure resilience in the presence of exogenous correlations.  \cref{sec:experiments} presents an empirical validation of our theoretical framework using real-world supply chain data, demonstrating the practical utility of the proposed metric. Finally, \cref{sec:discussion} concludes with managerial implications and directions for future research.

\section{Literature Review} \label{sec:related_work}

A growing body of work reframes supply and production systems away from collections of independent firms and toward networks whose structure shapes system behavior, particularly through sourcing relationships that define interfirm dependencies. Early research in the operations management (OM) and supply chain management (SCM) literature shows that coordination, performance, and risk emerge from patterns of interaction rather than centralized control or bilateral decisions \citep{choi2001a,pathak_complexity_2007}, a view later formalized through explicit theories of supply networks as interconnected architectures \citep{choi2009a}. As this perspective developed, scholars emphasized that risk in supply systems is often systemic, rooted in sourcing interdependencies rather than isolated operational choices or managerial interventions \citep{kleindorfer2005managing,sodhi_supplychain_2014}. Empirical evidence further links sourcing concentration, overlap between suppliers, and coupling across sourcing relationships to the variation in the exposure to disruption and the ability of the systems to continue to meet downstream demand \citep{craighead2007severity,blackhurst_empirically_2011,lee1997a}. Research adopting a complex network perspective adds that increasing interconnections among sourcing relationships can generate nonlinear execution challenges that are not readily observable at the dyadic level \citep{kim_structural_2011,hearnshaw2013complex}. Together, these studies suggest that many operational vulnerabilities arise from the sourcing structure itself, motivating analytical approaches that abstract from firm-specific details to isolate structural exposure \citep{kim2015supply,luo2019a}. 

Viewing supply systems through the lens of sourcing relationships shifts the attention from firm attributes to structural exposure embedded in shared suppliers, concentration, and limited substitutability. Within OM and SCM, multiple streams suggest that highly interconnected sourcing architectures may respond nonlinearly to shocks, with disruptions propagating through shared suppliers even when initial failures are limited \citep{craighead2007severity,kim2015supply}. Analytical and computational studies further indicate that sourcing architecture governs how disruptions are transmitted across interconnected entities, shaping system-level exposure rather than isolated losses in individual firms \citep{gopal_design_2016,zhao2018supply}. Empirical research suggests that positional asymmetries within sourcing networks and embeddedness in shared supplier relationships influence whether disruptions remain contained or escalate into broader supply loss, although such effects are typically documented within specific contexts or network realizations  \citep{carvalho2019production,ding_social_2023}. Analytical models reinforce the possibility that certain sourcing structures exhibit abrupt transitions from stability to widespread failure once the severity of the disruption crosses implicit thresholds, even when individual firms appear robust in isolation \citep{bimpikis2019supply,elliott2023supply,elliott2022supply}. Although recent work formalizes the existence of such thresholds, it stops short of offering general diagnostics applicable across alternative sourcing architectures  \citep{acemoglu2016networks,elliott2023supply,elliott2022supply}.

The implications of structural vulnerability become tangible once disruptions begin to move through sourcing dependencies over time. Empirical studies document that localized disruptions can spread across tiers through shared suppliers, capacity constraints, and temporal coupling, though the magnitude, timing, and persistence of the resulting supply loss vary widely between contexts and sourcing configurations \citep{blackhurst_empirically_2011,snyder2016or}. To capture these dynamics, subsequent work introduces the notion of ripple effects, describing how disruption impacts evolve as firms adjust their sourcing, production, and structural decoupling decisions, typically within specific network configurations or disruption scenarios \citep{ivanov_simulation-based_2017,ivanov_what_2025,dolgui2021a}. Analytical research further highlights that risk can propagate dynamically through supplier networks even in the absence of additional shocks, driven by endogenous adjustments within sourcing relationships \citep{kim2015supply,ivanov_simulation-based_2017}. Simulation-based studies explore how lead times, buffering policies, and substitution behavior shape propagation, but often rely on detailed parameterization that hides the inherent limits to resilience across sourcing structures \citep{nair2011supply,hosseini2020a}. Empirical evidence also indicates that system-level ability to preserve aggregate supply may continue to deteriorate even after individual firms recover, underscoring the importance of indirect and lagged effects transmitted through shared sourcing dependencies \citep{lee1997a,ivanov_simulation-based_2017}.

Faced with these dynamics, researchers have pursued a range of methodological strategies to evaluate resilience in sourcing systems. A prominent stream relies on scenario-based analysis and computational simulation to examine how particular sourcing configurations respond to disruptions and recover over time, yielding detailed insights into recovery dynamics and mitigation policies under assumed disruption realizations \citep{nair2011supply,ivanov_simulation-based_2017}. Other studies embed endogenous recovery actions and control decisions within optimization frameworks, allowing resilience to be assessed through restoration and reconfiguration strategies subject to time and capacity constraints, though typically for specific sourcing architectures and policy choices \citep{hosseini2020a,anokhin_mobilityasservice_2021}. In parallel, scholars have proposed resilience and robustness metrics that summarize exposure, service loss, or recovery performance, often aggregating outcomes observed under stress testing of sourcing structures rather than isolating limits inherent in the structure itself \citep{zobel2014a,simchi2014superstorms,zhao2019modelling}. While these approaches provide valuable guidance on recovery and mitigation, they often presuppose a sourcing structure, leaving open how much resilience is afforded by the structure itself prior to intervention \citep{aldrighetti2021a}.

First, our paper models the production network as a directed graph that undergoes a node percolation process, representing product failures. Each product has some suppliers, and if all of them fail, the product cannot be produced. Subsequently, a cascade is caused by such a failure. Our analysis builds on and extends \citet{elliott2022supply}. \citet{elliott2022supply} consider a \emph{hierarchical production network} that undergoes a \emph{link percolation process}. In their model, firms are operational if each of their inputs has a functioning link to at least one supplier. The reliability of the network is defined as the probability that the root product is successfully produced. They identify three structural drivers of reliability: \emph{(i)} the network depth and size, \emph{(ii)} the number of inputs that each product requires, and \emph{(iii)} the degree of multi-sourcing or the number of suppliers (which corresponds to the ability of firms to multisource their required inputs). We investigate the effect of similar structural factors---i.e., topology, number of inputs, and number of suppliers---but under a fundamentally expanded setup by: \emph{(i)} considering production shocks at the level of individual suppliers rather than links; \emph{(ii)} considering more general architectures than the hierarchical production networks of \citet{elliott2022supply}, and \emph{(iii)} studying a novel theory-informed resilience metric that guarantees that almost all products can be produced in the event of a shock. One of our main results shows that production networks are either {operationally resilient} or {fragile}. \citet{elliott2022supply} point out that if the magnitude of the systemic shock is greater than some critical value, the reliability (the probability that the root product is produced) goes to zero, corresponding to \emph{fragility}. When the magnitude of the shocks is below this critical value, the reliability attains a strictly positive value corresponding to \emph{resilience}. They observe that reliability increases with multisourcing (having many suppliers)  and decreases with interdependency (requiring a multitude of inputs). Unlike \citet{elliott2022supply}, we focus on determining the largest shock level a network can withstand while still ensuring that a significant fraction of products survive, especially as the network size grows. Our parameterization, based on size, multisourcing, and interdependency, yields insights consistent with theirs: resilience improves with multisourcing and deteriorates with increasing interdependencies. In addition, our model and resilience metric complement \citet{elliott2022supply} by enabling the analysis of arbitrary production networks, not just tree structures. We generalize two key premises of their work to broader settings. First, we show that having a large number of raw materials can cause fragility, extending their main insight beyond tree-like topologies. Second, we show that networks with few raw materials and finished goods, as well as those with sparsity, are operationally resilient, thus complementing their theory. In addition to their model, we also model structural decoupling and operationalize resilience through linear programming. Conceptually, this extension complements that of \citet{elliott2022supply}, as structural decoupling is related to their notion of the strength of the relationship between suppliers. In contrast, our results focus on characterizing the maximum tolerable shock probability (resilience) using the LP approximation of the cascade size. 

Second, our paper is closely related to \citet{chen2023when} and offers highly complementary perspectives that, together, provide a comprehensive toolkit for supply chain resilience. While \citet{chen2023when} contributes an essential optimization framework for allocating resources to maximize flow within a specific facility network, our work enriches this tactical approach by providing the foundational architectural theory that explains why certain networks are more responsive to such investments. By identifying network archetypes, our work helps managers understand the underlying geometric limits of their networks. These insights complement and inform the investment strategies proposed by \citet{chen2023when} as networks which are deemed more vulnerable through our framework may require higher intervention priority.

Third, our work complements and advances the insights of the Risk Exposure Index (REI; \citet{levi2016identifying}) by transitioning the methodology from iterative stress-testing of isolated node failures to a comprehensive model that accounts for distributed failures and treats architectural collapse as a stochastic process. The work of \citet{levi2016identifying} primarily uses the REI as a diagnostic tool to assess the performance impact of removing single, isolated nodes from a graph. In contrast, our work treats failure as a stochastic process (node percolation), moving beyond the impact of a single node to determine the system-wide failure rate that an entire architecture can sustain before aggregate production collapses.  In addition, our work leverages linear programming and formally links our notion of resilience to the REI measure through linear programming duality.  This connection enables managers to view REI not just as a simulation result of an isolated failure, but as a structural property of the position of a node within the network hierarchy, providing a unified perspective. 

The above literature paints a detailed but incomplete picture of how the sourcing structure shapes vulnerability, disruption propagation, and resilience. Much of the existing work evaluates resilience within particular sourcing configurations or under exogenously specified disruption scenarios, making it difficult to distinguish structural fragility from scenario-dependent outcomes or to identify inherent limits to tolerance \citep{nair2011supply,ivanov_simulation-based_2017}. Analytical results are often derived for stylized dependency architectures or narrow policy settings, limiting their ability to characterize how much supplier failure different sourcing structures can withstand while preserving aggregate supply \citep{acemoglu2016networks,bimpikis2019supply}, a concern echoed in calls for more unified structural perspectives \citep{hearnshaw2013complex}. From an operational perspective, this limits managers’ ability to compare alternative sourcing designs or assess whether resilience improvements stem from structure or contingent recovery actions \citep{snyder2016or,chen2023when}. We aim to help bridge this gap by examining how the sourcing topology, interdependence among shared suppliers, cascading failure mechanisms, and correlated exogenous shocks jointly shape the limits of system-level resilience.

\section{Model Setup} \label{sec:preliminaries}

\begin{figure}[t]
    \centering
    \subfigure[Multi-tier network. $\cR$ corresponds to the set of raw materials and $\cC$ corresponds to the set of finished goods.\label{subfig:product_graph}]{
    \begin{tikzpicture}[transform shape,scale=0.8]
        \Vertex[x=-1, y=-1]{u1}
        \Vertex[x=-1, y=0]{u2}
        \Vertex[x=-1, y=1]{u3}
        \Text[x=-1, y=-2, fontsize=\normalsize]{$\cR$}

        \Vertex[x=1, y=-0.5]{v1}
        \Vertex[x=1, y=0.5]{v2}

        \Edge[Direct, color=black](u1)(v1)
        \Edge[Direct, color=black](u3)(v1)
        
        \Edge[Direct, color=black](u2)(v2)
        \Edge[Direct, color=black](u1)(v2)
        
        \Text[x=2, y=0]{$\dots$}
        
        \Vertex[x=5, y=-1]{w1}
        \Vertex[x=5, y=0]{w2}
        \Vertex[x=5, y=1]{w3}
        \Text[x=5, y=-2, fontsize=\normalsize]{$\cC$}
        
        \Vertex[x=3, y=0.5]{z1}
        \Vertex[x=3, y=-0.5]{z2}
        
        \Edge[Direct, color=black](z1)(w1)
        \Edge[Direct, color=black](z1)(w2)
        \Edge[Direct, color=black](z1)(w3)
        
        \Edge[Direct, color=black](z2)(w1)
        \Edge[Direct, color=black](z2)(w2)
        \Edge[Direct, color=black](z2)(w3)

    \end{tikzpicture}} \hfill
    \subfigure[Underlying supply relationships between products $i, j$ with supplier sets $\cS(i)$ and $\cS(j)$. \label{subfig:supply_chain}]{
    \begin{tikzpicture}[transform shape,scale=0.5]
    \Text[x=0, y=-2, fontsize=\LARGE]{$i$}
    \Text[x=6, y=-2, fontsize=\LARGE]{$j$}

    \Vertex[x=0, y=1, color=teal, opacity=0.1, size=4]{i}
    \Vertex[x=0, y=1, color=teal, opacity=0.7, size=0.75]{Si1}
    \Vertex[x=0, y=2, color=teal, opacity=0.7, size=0.75]{Si2}
    \Vertex[x=0, y=0, color=teal, opacity=0.7, size=0.75]{Si3}

    \Text[x=-1.1, y=1,fontsize=\Large]{$\cS(i)$}

    \Text[x=7.1, y=1,fontsize=\Large]{$\cS(j)$}

    \Vertex[x=6, y=1, color=teal, opacity=0.1, size=4]{j}
    \Vertex[x=6, y=0.33, color=teal, opacity=0.7, size=0.75]{Sj1}
    \Vertex[x=6, y=1.66, color=teal, opacity=0.7, size=0.75]{Sj2}

    \Edge[Direct,style=dashed](Si1)(Sj1)
    \Edge[Direct,style=dashed](Si1)(Sj2)

    \Edge[Direct,style=dashed](Si2)(Sj1)
    \Edge[Direct,style=dashed](Si2)(Sj2)

    \Edge[Direct,style=dashed](Si3)(Sj1)
    \Edge[Direct,style=dashed](Si3)(Sj2)
    
    \end{tikzpicture}}\hfill     %   \par\medskip
\subfigure[Main structural parameters. Raw materials $r$, final goods $c$, and sourcing dependencies $m, \mu$.\label{subfig:network_params}]{
%\begin{minipage}[t]{0.32\linewidth}\vspace{0pt}\centering
%\resizebox{\linewidth}{!}{%
\begin{tikzpicture}[transform shape,scale=0.72,>=Stealth]
    % --- Raw materials (left) ---
    \Vertex[x=-2.2, y=-0.7]{r1}
    \Vertex[x=-2.2, y=0]{r2}
    \Vertex[x=-2.2, y=0.7]{r3}
     % --- Braces for r ---
    \draw[decorate,decoration={brace,amplitude=4pt},thick]
        (-2.55,-0.85) -- (-2.55,0.85)
        node[midway,xshift=-12pt,align=center]{$r$};
     % --- Omitted layers/products in the middle (like 2a) ---
    \Text[x=-1.45, y=0]{$\cdots$}

    % --- Two representative (possibly distinct) products attaining mu and m ---
    \Vertex[x=-0.65, y=1.3, size=0.6]{vmu}
    \Text[x=-0.15, y=1.6]{$\mu$}

    \Vertex[x=0.65, y=-1.3, size=0.6]{vm}
    \Text[x=0.15, y=-1.6]{$m$}

    % --- Endpoints for the fans (completely invisible) ---
    % % targets of mu fan-out
    % \Vertex[x=0.34, y=0.15, size=0, style={draw=none,fill=none}]{mu_t1}
    % \Vertex[x=0.45, y=0.55, size=0, style={draw=none,fill=none}]{mu_t2}
    % \Vertex[x=0.35, y=0.95, size=0, style={draw=none,fill=none}]{mu_t3}
        % --- longer fan-out targets for mu ---
    \Vertex[x=1.10, y=0.25, size=0, style={draw=none,fill=none}]{mu_t1}
    \Vertex[x=1.30, y=0.70, size=0, style={draw=none,fill=none}]{mu_t2}
    \Vertex[x=1.10, y=1.15, size=0, style={draw=none,fill=none}]{mu_t3}
    
    % % sources of m fan-in
    % \Vertex[x=-0.35, y=-0.95, size=0, style={draw=none,fill=none}]{m_s1}
    % \Vertex[x=-0.45, y=-0.55, size=0, style={draw=none,fill=none}]{m_s2}
    % \Vertex[x=-0.35, y=-0.15, size=0, style={draw=none,fill=none}]{m_s3}
        % --- longer fan-in sources for m ---
    \Vertex[x=-1.10, y=-1.15, size=0, style={draw=none,fill=none}]{m_s1}
    \Vertex[x=-1.30, y=-0.70, size=0, style={draw=none,fill=none}]{m_s2}
    \Vertex[x=-1.10, y=-0.25, size=0, style={draw=none,fill=none}]{m_s3}

    % --- Fan-out edges (use TikZ-network Edge so styling matches Fig 2a/2b) ---
    \Edge[Direct, color=black](vmu)(mu_t1)
    \Edge[Direct, color=black](vmu)(mu_t2)
    \Edge[Direct, color=black](vmu)(mu_t3)

    % --- Fan-in edges (use TikZ-network Edge) ---
    \Edge[Direct, color=black](m_s1)(vm)
    \Edge[Direct, color=black](m_s2)(vm)
    \Edge[Direct, color=black](m_s3)(vm)

   % --- Omitted layers/products in the middle (like 2a) ---
    % \Text[x=-.8, y=0.35]{$\vdots$}
    \Text[x=1.45, y=0]{$\cdots$}

     % --- Final goods (right) ---
    \Vertex[x=2.2, y=-0.7]{c1}
    \Vertex[x=2.2, y=0]{c2}
    \Vertex[x=2.2, y=0.7]{c3}
    % --- Braces for r and c ---
    % \draw[decorate,decoration={brace,amplitude=4pt},thick]
    %     (-3.35,-0.85) -- (-3.35,0.85)
    %     node[midway,xshift=-12pt,align=center]{$r$\\\scriptsize raw};

    \draw[decorate,decoration={brace,amplitude=4pt,mirror},thick]
        (2.55,-0.85) -- (2.55,0.85)
        node[midway,xshift=12pt,align=center]{$c$};
\vspace{20pt}
\end{tikzpicture}}%
%\end{minipage}}

    \caption{\cref{subfig:product_graph} shows the production network. Each node $i$ in the network of \cref{subfig:product_graph} has a supplier set $\cS(i)$. The underlying network of supply relationships between two products $i$ and $j$ with a sourcing dependency is shown in \cref{subfig:supply_chain}. \cref{subfig:network_params} highlights the key structural parameters that appear in our theoretical results: number of raw materials  $r$, finished goods $c$, sourcing requirement $m$, and sourcing influence $\mu$.
}
    \label{fig:supply_chain}
\end{figure}

This section formalizes the production network and the failure process and introduces the resilience metric that serves as the organizing device for the analysis. The goal is to use this metric to assess how the network structure governs the ability of production systems to withstand cascading failures.

We model the production system as a network $\mathcal{G}(\mathcal{K}, \mathcal{E})$, where $\mathcal{K}$ represents a set of $K$ products. Each product $i \in \mathcal{K}$ requires a specific set of inputs, denoted by $\mathcal{N}(i)$, to be manufactured. These requirements define the edges $\mathcal{E}$ of the network, represented by an adjacency matrix $A$. To characterize the structural complexity of the production network $\cG$, we define the sourcing requirement $m$ as the maximum in-degree and the sourcing influence $\mu$ as the maximum out-degree.

The network is structured across different tiers of production. It begins with $r$ raw materials $\mathcal{R}$, which are primary inputs that do not require further processing ($|\mathcal{N}(i)| = 0, i \in \mathcal{R}$). At the other end of the chain are $c$ final goods $\mathcal{C}$, which are products not utilized in the production of others. 

Every product $i$ is sourced from a latent set of suppliers $\mathcal{S}(i)$ in a stochastic manner, as, in practice, modeling suppliers explicitly as nodes would require detailed firm-level data that is rarely observable at scale in the real world. For that reason, we assume a flexible sourcing logic in which a product can be produced by any one of several potential suppliers, provided that the required input products are available. This abstraction captures diversification across firms, facilities, or regions without explicitly modeling suppliers as separate nodes, allowing us to focus on the structural propagation of failures through product dependencies rather than the individual supplier network.

The resilience of this system is defined by its ability to maintain the production of final goods despite disruptions in the upstream supply chain. A product fails if it loses access to its required inputs or if its suppliers are unable to provide them. By quantifying how failures propagate from raw materials through intermediate levels to final goods, the resilience metric captures the systemic risk inherent in the network’s topology. This setup enables us to examine how specific structural parameters, such as $r$, $c$, $m$, and $\mu$, influence the overall stability of the supply chain in various failure scenarios (\cref{fig:supply_chain}).

% We consider a set of products, denoted by $\cK$, with cardinality $K = |\cK|$, where each product $i \in \cK$ can be produced by a number of suppliers and also requires certain inputs to be produced. Specifically, each product $i \in \cK$ has a set of requirements (inputs), denoted by $\cN(i)$ that it needs in order to be made. The products and their input requirements define the \emph{network}, $\cG(\cK, \cE)$. The network is also associated with the adjacency matrix $A$, and the cardinality of the edges is $M = |\cE|$.  

% The \textit{production network} $\cG$ -- henceforth called \textit{``network''} -- starts with \emph{raw materials} (or sources), which are materials that do not require any input, that is, they have $|\cN(i)| = 0$ and are the ``initial products'' that are used in the production of others. We denote the set of raw materials (or primary sector) by $\cR$ and its cardinality by $r =  |\cR|$. We use $\cC$ to denote the set of final goods (or final sector), which corresponds to products that are not required in the production of others, and use $c = |\cC|$ to denote their cardinality. Each product $i \in \cK$ can be sourced from a set of suppliers $\cS(i)$. \redit{We assume that any supplier of a product can source from any of the suppliers of the products on which it depends.}

% We define the sourcing requirement $m$ of $\cG$ as the maximum in-degree of any product, and the sourcing influence $\mu$ of $\cG$ as the maximum out-degree of any product. If $\cG$ is random, $m$ and $\mu$ are defined to be the corresponding expected in/out-degrees. 

\subsection{Cascading Failure Model} \label{sec:node_percolation}

To analyze the systemic stability of the supply chain, we model disruptions as a node percolation process. In our core framework, we consider a homogeneous failure model where each product $i$ is supported by an identical number of suppliers, $n$. Each supplier fails independently and exogenously with a probability $x \in (0, 1)$. This spontaneous and exogenous failure of individual suppliers serves as a trigger for broader network effects.

The survival of a product depends on both its internal supply base and the health of its upstream requirements. Specifically, a product $i$ is successfully produced (indicated by the variable $Z_i = 1$) if, and only if, two conditions are met: (i) all its required inputs $j \in \mathcal{N}(i)$ are available and (ii) at least one of its own suppliers $s \in \mathcal{S}(i)$ remain operational. This interplay between endogenous supplier failures and exogenous network dependencies is captured by the following recursive dynamics:
\begin{align} \label{eq:dynamics}
Z_i = \prod_{j \in \mathcal{N}(i)} Z_j \left( 1 - \prod_{s \in \mathcal{S}(i)} X_{is} \right),
\end{align}
where $X_{is}$ indicates the spontaneous failure of supplier $s$. We assume that suppliers fail independently with probability $x$, representing an exogenous shock to the system.
%We also let $u_i$ to be the probability that product $i$ fails spontaneously. In this simplified model
\footnote{To illustrate the resilience results, we have relied on the assumption that suppliers fail independently of each other with the same probability $x$, so that product $i$ fails spontaneously with probability $x^n$. In general, suppliers can also fail under more complex structures, such as correlations among themselves. In that case, our results can be extended by treating the product failure probabilities as general functions of the individual supplier shocks and the correlation parameters. We explore this extension in \cref{sec:extensions}.} 
%we have that $u_i = x^n$.

The resilience of the network is then characterized by the threshold failure probability $x$ that ensures at most an $\varepsilon$-fraction of the total products $K$ fail.

\subsection{The {Resilience} Metric}

To evaluate the behavior of complex production networks and identify potential vulnerabilities, we require a formal metric that quantifies the system's ability to maintain core functionality under stress. In our model, a more resilient network is one that can withstand larger exogenous shocks, represented by a higher probability of supplier failure $x$, while maintaining the survival of the vast majority of products.

Let $S$ be the random variable representing the total number of surviving products, where $S = \sum_{i \in \mathcal{K}} Z_i$ and $F$ be the number of failures (or the cascade size) with $F = K - S = \sum_{i \in \mathcal {K}}(1 - Z_i)$. We define resilience as the maximum allowable probability of supplier failure that maintains a target survival rate with high confidence.

\begin{mdefinition}[Resilience]\label{def:resilience-taxonomy} 

For $\varepsilon \in (0, 1)$, the resilience of a production network $\mathcal{G}$ is:
$$R_{\mathcal{G}} (\varepsilon) = \sup \left \{ x \in (0, 1) : \Pr \left [S \ge (1 - \varepsilon) K \right ] \ge 1 - \frac {1} {K} \right \}.$$
\end{mdefinition}

One can consider this metric as a stress test for the network: given a sense of supplier risk and network structure, it tells a manager how much disruption the system can absorb before performance starts to break down.

Operationally, $R_{\mathcal{G}}(\varepsilon)$ represents the highest failure threshold where at least $(1 - \varepsilon)$ fraction, for instance, 90\%, of the network survives with high probability. This definition allows us to establish a clear taxonomy for supply chain stability as follows:

\xpar{Architectural Fragility} A production network structure is defined as \textit{architecturally fragile} if and only if $R_{\mathcal{G}}(\varepsilon) \to 0$ as $K \to \infty$. Much like the concept of ``scale-free'' \citep{nair2011supply} networks in operations, the $K \to \infty$ behavior reveals the inherent stability of a growth pattern.  This classification signifies that the network's design is fundamentally unstable; as the system grows in complexity or scale, it becomes increasingly vulnerable to even infinitesimal shocks, making the limit a diagnostic for fragility\footnote{This dichotomy is robust; if a network is fragile for any fixed $\varepsilon$, it is fragile for all $\varepsilon$. This metric provides the foundation for our theoretical analysis, enabling us to categorize network topologies by asymptotic stability. Proofs regarding the equivalence of fragility characterizations are provided in \cref{app:dichotomy}.}.

Generally, while the limit $K \to \infty$ is meaningful for assessing tipping points and fragility, supply chain managers work with finite networks. To that end, specifically, we now distinguish between inherent fragility (an architectural property) and operational resilience (a status for a specific network size):

\xpar{Operational Resilience} A production network is \textit{operationally resilient at level $\delta > 0$} if and only if $R_{\mathcal{G}}(\varepsilon) > \delta$ for a fixed threshold $\delta > 0$ for any finite $K$ of interest. Requiring the condition to hold for any $K$ and for a fixed threshold $\delta$ ensures that the metric remains a practical tool for benchmarking real-world supply chains. This allows managers to use the metric as a real-time key performance index (KPI) to evaluate their current supply chain state, regardless of future growth.

\section{Main Results} \label{sec:main_results}

This section applies the resilience metric introduced in the previous section to characterize which network architectures remain resilient and which become fragile,  as scale increases.

\subsection{Resilient and Fragile Classes of Networks}\label{sec:high-level-structures}

The classification of a network as resilient or fragile depends fundamentally on its topology. To illustrate this, we contrast simple structures with complex multi-tier architectures.

We first consider a two-tier network architecture  \citet{bimpikis2018multisourcing,willems2008data} where the final goods are directly dependent on a small number of raw materials. In this two-tier system, each final good requires $m$ inputs and each raw material supplies $\mu$ products. We establish that if the sourcing requirement $m$ and the sourcing influence $\mu$ remain constant as the network scales, the system is inherently resilient. Here, failures are localized: a disrupted raw material impacts only a fixed number of downstream products, preventing systemic cascading failures.

Quantitatively, the two-tier network ensures operational resilience for $\delta = \left ( \frac {\varepsilon} {m (\mu + 1)} \right )^{1/n}$. We give some initial high-level intuition on why this happens, by calculating $\delta$:  At the operational level $\delta$ the probability of supplier failures is bounded by $\delta$, and on average, at most $r \delta ^n$ raw materials fail, which implies that at most $r \delta^n + \mu r \delta^n = (\mu + 1) r \delta^n \le m K \delta^n (\mu + 1)$ products fail in total,  on average (because with the sourcing requirement $m$, the number of raw materials satisfies $r \le mK$). To establish a lower bound on the resilience, we observe that cascade sizes satisfy $F \leq \varepsilon K$ almost surely, as long as~$\delta = \left ( \frac {\varepsilon} {m (\mu + 1)} \right )^{1/n}$~as previously stated.

As we transition from shallow two-tier networks to general multi-tier networks, the risk of cascading failures arises sharply. In deeply layered architectures, a final good may be connected to the primary sector through an exponentially large number of paths. Although this suggests redundancy, it also increases the statistical likelihood that at least one critical dependency is disrupted along every available path. This mechanism typically results in a heavy-tailed distribution of the cascade size $F$, modeled as a power law with scaling exponent two $(\Pr [F \ge f] \geq  {C}/{f}, f>0, C \text { constant})$, where the aggregate risk is dominated by rare, system-wide cascades rather than localized events \citep{nair2022fundamentals,wegrzycki2017cascade}.%, modeled as a power law,

% \begin{align*}
%     \Pr [F \ge f] \gtrsim \frac {1} {f}.
% \end{align*}

Such fragility is a generic property of directional input-output dependencies, arising from the structural layering of the production network. To illustrate this mechanism, we rely on a stylized hierarchical input–output network and provide an analytical proof in \cref{app:power_laws}, showing that the probability of large-scale cascades is bounded away from zero as the network grows.

Next, when the primary sector is large, specifically when the number of raw materials $r = \omega(K^{2/3})$, the network becomes fragile. In these cases, the sheer breadth and depth of the dependency structure ensure that even a vanishingly small failure probability $x$ triggers a systemic collapse. In contrast, networks with a fixed number of raw materials and constant dependencies maintain their resilience. This transition demonstrates that fragility is not the result of specific parameter settings but an emergent property of directional, hierarchical layering in large-scale supply chains.

Motivated by this observation, we generalize our claim by deriving resilience bounds for arbitrary networks using global graph features. Specifically, the global graph features that govern the bounds are the source requirement $m$, the sourcing influence $\mu$, the number of raw products $r$ and the number of final goods $c$; proved in Appendix \ref{app:theorem:resilience_graph_statistics}:
    
\begin{theorem}[Bounding the resilience with global graph features] \label{theorem:resilience_graph_statistics_2}
    For any network $\cG$, the resilience satisfies
    \begin{align*}
         \left ( \frac {\varepsilon} {2(m + r) (\mu + c)} + \sqrt {\frac {\log K} {rK}} \right )^{1/n}  \le R_{\cG}(\varepsilon) \le \left [ \frac {(1 - \varepsilon) K} {\sqrt 2 r^{3/2} + \sqrt {r \log K}} \right ]^{1/n}.
    \end{align*} Thus, if $r = \omega (K^{2/3})$, then the network is fragile. Conversely, if $r$, $c$, $m$, and $\mu$ are constant, the network is operationally resilient at $\delta = \left ( \frac {\varepsilon} {2(m + r) (\mu + c)} \right )^{1/n}$. 
    
\end{theorem}

The above parameters are chosen to capture the dominant structural features that govern the propagation of disruption, such as sourcing requirements, shared inputs, and supply dependencies, rather than fine-grained topological details. The objective of \cref{theorem:resilience_graph_statistics_2} is therefore to distinguish among broad classes of production structures that differ in resilience, not to model micro-level aspects of network architecture.

Indeed, many networks with substantially different local structures can share identical values of the four parameters. This abstraction is both intentional and essential: it allows us to identify structural limits that hold uniformly across entire families of networks. Although finer distinctions are important in specific settings and for obtaining sharper rates (see, for example, \cref{app:architectures}, where we analyze several concrete architectures), the core analysis emphasizes how these coarse architectural features fundamentally constrain resilience in a general and comparative sense.

% \begin{figure}[t]
%     \centering
%     \includegraphics[width=0.7\linewidth]{figures/resilience_slide.pdf}
%     \caption{\textbf{Illustration of fragile and resilient network structures.} On the one hand, the key characteristic of the \textbf{fragile} networks is having many raw materials ($r$), which makes the network look like a ``top hat''. Two forms of fragile networks are shown---one corresponding to lineage growth and one corresponding to a multi-tier network with a fixed number of layers and a large width. Among the two, the lineage growth exhibits higher fragility due to higher centralization. In the cases of more centralized networks, failures cause larger cascades in general. On the other hand, the key characteristic of \textbf{resilient} networks is that they have few raw materials ($r$) and finished goods ($c$) as well as sparsity, i.e., sourcing requirements ($m$) and sourcing influence ($\mu$) remain fixed, resembling a ``rolling pin'' structure. An example of a resilient network is a multi-tier network with a fixed width and multiple layers, characterized by fixed connections. Finally, networks can also be partially resilient, as in the case of an inverted lineage growth network. In the Figure, we have visualized the failures due to network effects in dark red, the failures in pink, and the unaffected nodes in light blue.} 
% \label{fig:qualitative_representation}
% \end{figure}

\begin{figure}[t]
    \centering
    \includegraphics[width=0.7\linewidth]{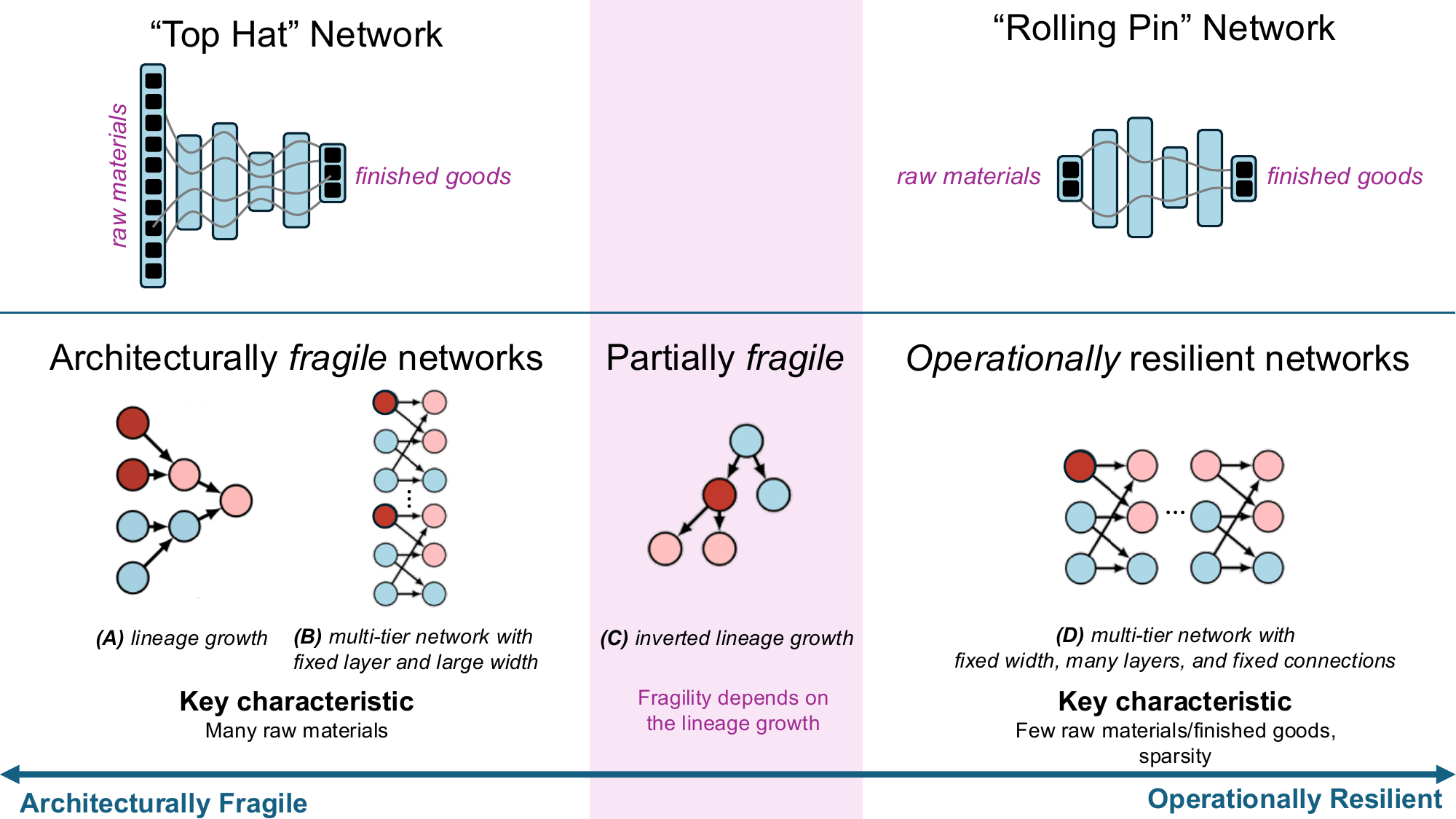}
    \caption{\textbf{Illustration of fragile and resilient network structures.} \textbf{Fragile} networks resemble a ``top hat'', characterized by a broad base of many raw materials ($r$) feeding into a complex, often highly centralized, multi-tier middle. Two forms are shown: \textbf{(A)} lineage growth, which exhibits the highest fragility due to centralization, where individual hub failures trigger massive non-linear cascades; and \textbf{(B)} a multi-tier network with a fixed number of layers and large width. \textbf{Resilient} networks resemble a ``rolling pin'', defined by few raw materials ($r$) and finished goods ($c$), with sourcing requirements ($m$) and sourcing influence ($\mu$) remaining fixed as the network grows. This sparsity ensures failures in one part of the chain cannot easily propagate to the broader system. Shown as \textbf{(D)} a multi-tier network with fixed width, many layers, and fixed connections. Networks can also be partially resilient, as in \textbf{(C)} an inverted lineage growth network, where a small primary sector limits systemic failure entry points. A critical managerial implication follows: expanding the supplier base to include many diverse raw materials may inadvertently shift a resilient ``rolling pin'' into a fragile ``top hat'', increasing the likelihood of catastrophic systemic collapse. Failures due to network effects are shown in dark red, cascading failures in pink, and unaffected nodes in light blue.} 
\label{fig:qualitative_representation}
\end{figure}

From a managerial perspective, the structural transition from resilience to fragility can be visualized through \textit{two} distinct geometric archetypes (\cref{fig:qualitative_representation}). In \cref{app:architectures}, we analyze several examples of such networks.

\subsection{Interpreting Structural Regimes Through Real-World Supply Chains}

The results mentioned above in \cref{theorem:resilience_graph_statistics_2} and in \cref{app:architectures} highlight that resilience in production networks is fundamentally shaped by architectural features rather than the reliability of individual suppliers. This distinction is critical for interpreting real-world supply chains, which often appear robust when assessed locally but remain vulnerable due to hidden structural dependencies. Our characterization of resilient and fragile regimes provides a structural lens through which widely observed disruption patterns can be understood.

Fragile ``top hat'' architectures are common in industries where production relies on a large, heterogeneous set of raw materials or specialized upstream inputs that feed into dense intermediate processing layers. Chemical manufacturing, advanced polymers, semiconductors, and pharmaceutical supply chains may use this structure: a broad base of distinct chemical precursors or rare inputs feeds into a relatively narrow set of processing stages that serve many downstream products \citep{xiong_semiconductor_2025}. While such architectures may be efficient from a specialization or cost perspective, our results show that they are inherently prone to cascading failures. As the number of raw materials grows, the probability that at least one upstream disruption disconnects large portions of the network increases sharply, even when individual failure probabilities are small. This mechanism aligns with documented cases where seemingly localized upstream shocks, such as failures in specialty chemicals, rare earths, or active pharmaceutical ingredients, rapidly propagate across multiple product lines and firms.

By contrast, resilient ``rolling pin'' architectures resemble supply systems in which both the number of raw materials and the number of final goods are constrained, and in which sourcing requirements and influence remain bounded as the system scales. Such structures are more commonly observed in modular manufacturing systems, standardized component ecosystems, the automotive industry, or industries that deliberately limit product variety relative to available inputs \citep{he_developing_2026}. In these networks, failures remain localized because disruptions encounter limited branching opportunities as they propagate. Theoretical bounds derived in \cref{theorem:resilience_graph_statistics_2} formalize why these systems can absorb nontrivial shocks without experiencing system-wide collapse, even as the number of intermediate production stages grows.

An important implication of the analysis is that depth alone does not generate fragility; rather, fragility emerges from the interaction between depth and width. Multi-tier supply chains are not inherently vulnerable simply because they involve many stages. Instead, vulnerability arises when deep chains are combined with expanding upstream breadth or highly centralized intermediate layers. This distinction helps reconcile why some deeply layered global supply chains, such as those with standardized platforms and controlled interfaces, remain stable, while others experience disproportionate disruption from relatively small shocks.

The results also clarify why common managerial responses, such as adding suppliers or increasing sourcing options, do not necessarily improve resilience. Expanding the number of raw materials or upstream inputs without controlling sourcing influence can unintentionally push a network into the fragile regime identified by the theory. From a structural standpoint, such interventions increase the number of potential failure entry points without reducing the network’s propensity to propagate failures once they occur. This insight helps explain why diversification strategies sometimes fail to prevent systemic disruption despite appearing sensible at the firm level.

Finally, the distinction between resilient and fragile regimes emphasizes that resilience exhibits threshold behavior rather than smooth tradeoffs. Networks do not gradually become more fragile as they grow; instead, once structural parameters cross critical boundaries, even infinitesimal shocks can trigger large cascades. This perspective underscores the importance of architectural design decisions made early in a supply network's evolution. Once a production system enters a fragile regime, incremental adjustments, such as marginal redundancy or modest buffering, are unlikely to fundamentally alter its systemic behavior.

The results in this section show that network architecture induces sharp structural regimes, with some classes remaining resilient as scale increases and others becoming inherently fragile.

\section{Robustness and Extensions} \label{sec:extensions}

In this section, we extend our structural resilience framework to incorporate finer network features and examine whether the logic of structural regimes in \cref{sec:main_results} persists beyond the high-level structural properties identified in \cref{sec:high-level-structures} and under supplier heterogeneity and correlated disruptions.

\subsection{Operationalizing Resilience in General Networks via Linear Programming}

So far, our bounds on the resilience in \cref{theorem:resilience_graph_statistics_2} depended on four global network characteristics ($r, c, m, \mu$); however, an interesting open question is how resilience can be expressed in terms of the specific and fine-grained structure of the network $\cG$. To address that, we develop linear programming (LP) approximations to study the resilience metric across general network topologies.

The resilience metric remains meaningful even in networks with cycles, such as those found in circular economies, where recycled materials flow back into upstream processes. In these general cases, calculating the exact expected number of failures ($\mathbb{E}[F]$) is computationally intractable. The LP formulation allows us instead to approximate the expected cascade size $\ev {} {F}$ as a function of the adjacency matrix $A$ of $\cG$ and the vector of the failure probabilities per-product $u = (u_1, \dots, u_K)^T$; see \cref{eq:lp} below.  Subsequently, we show that the cascade size can be systematically bounded through LP duality formulations and use that to derive generic bounds for the resilience in terms of the adjacency matrix $A$; see our main result in \cref{theorem:lp_duality_resilience_2}.

Moreover, in realistic settings, firms use several structural decoupling mechanisms, such as holding excess inventory, supplier redundancy, input standardization, and modular product design, which all serve to reduce the probability of production dependency. We model structural decoupling as the probability $y$ that an edge (dependency) remains in place during a shock; hence, by reducing $y$ we can analyze how such buffers weaken the propagation of failures, increasing the resilience metric, and effectively transforming the network into a more resilient substructure. Formally, the parameter $y$ can be viewed as a probabilistic sparsification of the production network, and a low value of $y$ corresponds to a large structural decoupling. The value of this approach lies in its ability to map complex structural decoupling policies back onto structural topology, and it can be easily generalized to account for heterogeneous inventories per product or supplier. By viewing the buffered supply chain as an ``edge-pruned'' network, we can leverage our theoretical bounds on the edge-pruned network to estimate the resilience of the original, more complex/coupled system.

In \cref{app:linear_program}, we show that the solution to the following LP effectively approximates the expected cascade size in the original network. Specifically, for a choice of $y = \varrho/m$, where $m$ is the sourcing requirement and $\varrho \in (0, 1)$ is the approximation factor, we demonstrate that the solution of the following LP is an $(1 + \varrho)$-approximation to the expected cascade in the edge-pruned network (see \cref{theorem:general_ub} and its proof in \cref{app:linear_program}): 
\begin{align} \label{eq:lp}
    \max_{\beta} & \quad \one^T \beta \\  \text{ s.t. } & \quad \beta \le yA^T \beta + u, \quad \zero \le \beta \le \one. \nonumber
\end{align} In \cref{eq:lp}, the decision variables $\beta = (\beta_1, \dots, \beta_K)^T$ model the probabilities of each one of the $K$ products failing in the edge-pruned network, $y$ denotes the edge-pruning probability that models structural decoupling, $A$ corresponds to the adjacency matrix of $\cG$ and $u = (u_1, \dots, u_K)^T$ is the vector of failure probabilities for each product; cf. \cref{sec:preliminaries},  where by assuming that products fail independently with exogenous shock probability $x$ and with $n$ suppliers per product, we had $u_i  = x^n$ for all $i$.

The above provides a linear-programming tool that approximates the expected cascade size from the network's adjacency matrix and failure probabilities. This tool enables managers to leverage empirical data to identify which products trigger the largest cascades, and can be solved efficiently in large networks in complexity, which is almost linear in the number of nodes and edges of the networks, using an iterative fixed point algorithm (see \cref{app:algorithm_lp}). 

In \cref{{theorem:lp_duality_resilience_2}}, we use the LP duality to link the resilience metric with Katz centrality, which is a measure of a node's global influence. This provides us with a systematic way to rank the criticality of products within the network and lower-bound the resilience. Formally, we extend the definition of resilience to include the edge-retention probability $y$, and denote the extended resilience metric by $R_{\cG}(\varepsilon; y)$; noting that $R_{\cG}(\varepsilon; y)$ is decreasing in $y$ and $R_{\cG}(\varepsilon; 1) = R_{\cG}(\varepsilon)$. Our result follows (the technical formulations and proofs are in \cref{app:linear_program}):

\begin{theorem} \label{theorem:lp_duality_resilience_2}
    If the structural decoupling is sufficiently large, which corresponds to a low value of $y$, specifically $y < 1/ \mu$, the resilience satisfies $$R_{\cG}(\varepsilon; y)  \ge \left (\frac {\varepsilon} {\one^T \beta_{\cG}^\katz (y)} \right )^{1/ n},$$ where $\beta_\cG^\katz(y) = (I - yA^T)^{-1} \one$ is the vector of Katz centralities computed on the production network $\cG$ with adjacency matrix $A$.
\end{theorem}

Generally, the lower bounds in \cref{theorem:resilience_graph_statistics_2,theorem:lp_duality_resilience_2} can be compared directly, with the LP-based bound (\cref{theorem:lp_duality_resilience_2}) being more meaningful when $r \cdot c = \Omega (K)$. This setting is typical in modern, highly specialized supply chains where the diversity of raw materials and finished goods is significant relative to the total number of intermediate products. In these scenarios, \cref{theorem:lp_duality_resilience_2} provides a more precise estimate of the lower bound for resilience, accounting for the specific path topology, while \cref{theorem:resilience_graph_statistics_2} provides a more conservative bound focusing on the cardinality of raw products, final goods, and sourcing dependencies.

Examples of such specialized supply chains include pharmaceutical manufacturing and advanced electronics (e.g., aerospace components). In these sectors, the ``intermediate'' production steps are often streamlined or highly integrated, but they depend on a vast array of specialized chemical precursors or rare-earth materials (large $r$) to produce a wide range of specific final formulations or high-precision components (large $c$). Here, the LP-based bound (\cref{theorem:lp_duality_resilience_2}) captures the critical role of specific intermediate ``hubs'' that might be overlooked by the global parameters in \cref{theorem:resilience_graph_statistics_2}. 

\subsection{Connection to the Risk Exposure Index (REI)}

Our resilience framework shares deep conceptual links with established supply chain metrics, most notably the Risk Exposure Index (REI; \citet{levi2016identifying,simchi2014superstorms,ham2022companies}). While our metric focuses on the threshold of systemic collapse due to random failures, REI measures the potential impact of a disruption at a specific node on the overall network performance.

By defining the potential impact of product $i$ as the change in the expected number of failures following an infinitesimal shock to product $i$ and the risk exposure index to be the maximum such impact across all products, we can identify the most critical node in the production network, i.e., the product whose failure triggers the largest cascades. In \cref{proposition:rei_katz}, we show that for networks with sufficient decoupling (i.e., $y<\min \{ 1/\mu, 1/m \}$), the Risk Exposure Index is affected by the node with the maximum Katz centrality in the reverse production network $\cG^R$ (namely the ``most upstream'' product); proved in \cref{app:rei}:

\begin{proposition} \label{proposition:rei_katz}
    If the structural decoupling is sufficiently large, corresponding to a low value of $y$, specifically  $y < \min \{1/\mu, 1/m \}$, and $x < (1 - \mu y)^{1/n}$,  then the REI depends on the most upstream product, namely: $$\mathrm {REI}_\cG(x; y) = n x^{n - 1} \max_{i \in [K]} \gamma_{i, \cG}^\katz(y),$$ where $\gamma_{ \cG}^\katz(y) = (I - yA)^{-1} \one$ is the vector of Katz centralities in the reverse production network $\cG^R$ with adjacency matrix $A^T$, and $\gamma_{i, \cG}^\katz(y)$ is its $i$th entry. 
\end{proposition}

This connection provides a powerful bridge to the network science literature, where Katz centrality is known to relate to the average path length, such as, for example, in the case of small-world networks \citep{watts1998collective,basole2014supply} and the structural influence of a node. The above result shows that small-world networks, where path lengths between products are small, may be less prone to disruptions compared to networks where there are long paths between products. 

While \cref{proposition:rei_katz} expresses REI in terms of Katz centrality on $\cG^R$, REI is not a separate object from our resilience metric: it is the dual/sensitivity companion of the primal LP, given in \cref{eq:lp}, that bounds resilience in \cref{theorem:lp_duality_resilience_2}, expressed in terms of the Katz centrality on $\cG$. In particular, REI measures the maximum potential impact of an individual product’s shock on the resilience bound, because the dual optimal solution characterizes how the primal (resilience) objective changes under marginal perturbations of shock constraints. Empirically, this primal-dual link implies that products with high REI should be precisely those whose idiosyncratic shocks most degrade our resilience metric. Later in the paper, we verify this prediction by regressing average REI against resilience in \cref{sec:rei-resilience-connection}; see \cref{tab:resilience_ttr}.

Managerially, these connections are significant because they enable local decision-making. Managers do not necessarily require a global controller to identify risks; by calculating the REI through our LP framework, they can pinpoint the most risky supplier. This allows for targeted interventions, such as increasing safety stock or dual-sourcing for high-REI nodes, that effectively lower systemic risk without requiring an overhaul of the entire network architecture.

\subsection{Heterogeneity in Suppliers and Products and Correlated Failures}

\xpar{Heterogeneity in Suppliers and Products} Real-world supply chains are rarely homogeneous; products may have varying numbers of suppliers, and certain sectors may be more prone to failure than others. We extend our model to incorporate these heterogeneities by defining a generalized resilience metric that accounts for the diverse supplier bases between products, i.e., the varying number of suppliers ($n_i$) per product $i$ and the heterogeneous and correlated failure probabilities ($x_s$) per supplier $s$.

Our analysis confirms that the core structural logic of the resilient-versus-fragile regime still holds under these conditions. Although heterogeneity shifts the specific numerical thresholds for resilience, for example, products with fewer suppliers act as bottlenecks that lower the overall threshold, the behavior of the network remains dictated by its topology.

\xpar{Correlated Failures and Common Shocks} A significant concern in supply chain management is the occurrence of correlated failures, where multiple suppliers or products fail simultaneously due to common shocks such as natural disasters, regional geopolitical instability, or industry-wide labor strikes. We formalize this by introducing a vector of high-order correlation parameters $(\rho)$ into the joint failure distribution.

Recall that with $n_i$ suppliers per product $i$, we have a total of $N = \sum_{i\in\cK} n_i$ suppliers that may fail jointly and cause production failures to propagate. Let us denote the random set of jointly failed suppliers by $F_\cS$ and assume that supplier shocks follow the joint distribution $F_\cS \sim \nu$. We define generalized resilience with respect to the worst-case distribution $\nu$, such that, in expectation, at most $N\ .\ x$ suppliers fail (i.e., the average supplier failure rate is $x$), and that at least $1 - \varepsilon$ of the products survive with high probability: %, assuming that the set of jointly failed suppliers (denoted by $F_\cS$) are distributed according to $\nu$:% whereby supplier shocks follow a joint distribution $\nu$
\begin{align} \label{eq:resilience_general}
    R_{\cG}(\varepsilon) & = \inf_{\nu} \;\; \sup_{x} \;\;\;\;  x \\ \text{s.t.} \quad &   0 < x < 1, \; \frac {1}{N} \ev {F_\cS \sim \nu} {|F_\cS|} \le x, \; \Pr_{F_\cS \sim \nu} \left [ F \ge \varepsilon K \right ] \le \frac {1} {K}. \nonumber 
\end{align}

When suppliers fail independently at random with probability $x$, the generalized resilience in \cref{eq:resilience_general} reduces to the resilience in \cref{def:resilience-taxonomy}. The generalized resilience measure is hard to compute exactly; however, as previously, it can be approximated with LPs. To achieve this, we rely on a statistical tool called the \textit{Bahadur representation} \citep{bahadur1961representation,yuan2021community}, which enables us to express the product-level vector of shock probabilities, $u$, as a function of the marginal failure probabilities of suppliers, $x_s$, and a correlation parameter, $\rho$); See \cref{sec:general_resilience_lower_bound}. We differentiate between within-product correlations (where a product's suppliers fail together) and between-product correlations (where shocks hit multiple sectors at once).

Our results indicate that correlation generally exacerbates fragility. In the case of between-product correlations, even resilient topologies can become fragile, as a single event is very likely to trigger large cascades. This addresses a major robustness concern: while structural decoupling, like inventory or dual-sourcing, can mitigate independent shocks, systemic resilience in the face of correlated common shocks requires a fundamental shift in network geometry. Detailed mathematical representations of these correlations and their impact on resilience bounds are provided in \cref{app:generalized_resilience}.

Taken together, these extensions demonstrate that the resilient-versus-fragile regime logic persists under heterogeneity, correlated failures, and more general network structures.

\section{Numerical Illustration on Empirical Network Topologies} \label{sec:experiments}

We numerically illustrate the proposed resilience metric on empirically observed multi-echelon supply chain networks, demonstrating its discriminative power across real-world topologies. Specifically, we utilize the multi-echelon supply chain dataset provided by \citet{willems2008data}, which contains 38 distinct industrial networks. These data represent various manufacturing contexts and have been used in the supply chain resilience and operations management literature; cf. \citet{blaettchen2021traceability,thevenin_material_2021,park_structural_2018}. 

\subsection{Numerical Estimation of Resilience} \label{sec:numerical-estimation-resilience}

In real-world networks, the exact calculation of the resilience threshold is complicated by non-standard topologies. We therefore utilize a Monte Carlo (MC) simulation approach to estimate resilience. For each network, we simulate 1,000 failure scenarios in a range of spontaneous supplier failure probabilities. We define the surviving fraction for each trial and determine the maximum shock that ensures a high survival rate with high confidence. We denote the resilience estimated from the MC simulations by $\hat{R}_{\mathcal{G}}(\varepsilon;y)$, or simply $\hat{R}_{\mathcal{G}}(\varepsilon)$ when $y=1$. To provide a single, comprehensive measure of the resilience of a network across all possible failure tolerances ($\varepsilon$), we introduce the average resilience as the area under the curve (AUC) of $\hat{R}_{\mathcal{G}}(\varepsilon;y)$ that integrates $\varepsilon$ out, denoted by $\bar{R}_{\mathcal{G}}(y)$ and defined as follows: $\bar{R}_{\mathcal{G}}(y) = \int_0^1 \hat{R}_{\mathcal{G}}(\varepsilon; y) d\varepsilon.$

This metric represents the area under the resilience-tolerance curve, providing a holistic view of the system’s structural integrity. We estimate the average resilience for three networks (see \cref{tab:statistics-willems} for the statistics of these networks) from \citep{willems2008data}, and report the results in \cref{fig:willems}. We observe that the three networks have three distinct resilience profiles with Network \#30 being the most resilient on average, followed by Network \#10 and subsequently Network \#20. Network \#30 has the smallest density and average degree (cf. \cref{tab:statistics-willems}), factors that serve against the propagation of shocks as we showed previously.

\begin{figure}[!h]
    \centering
    \subfigure[Distribution of tiers \citep{willems2008data}]{\includegraphics[width=0.3\textwidth]{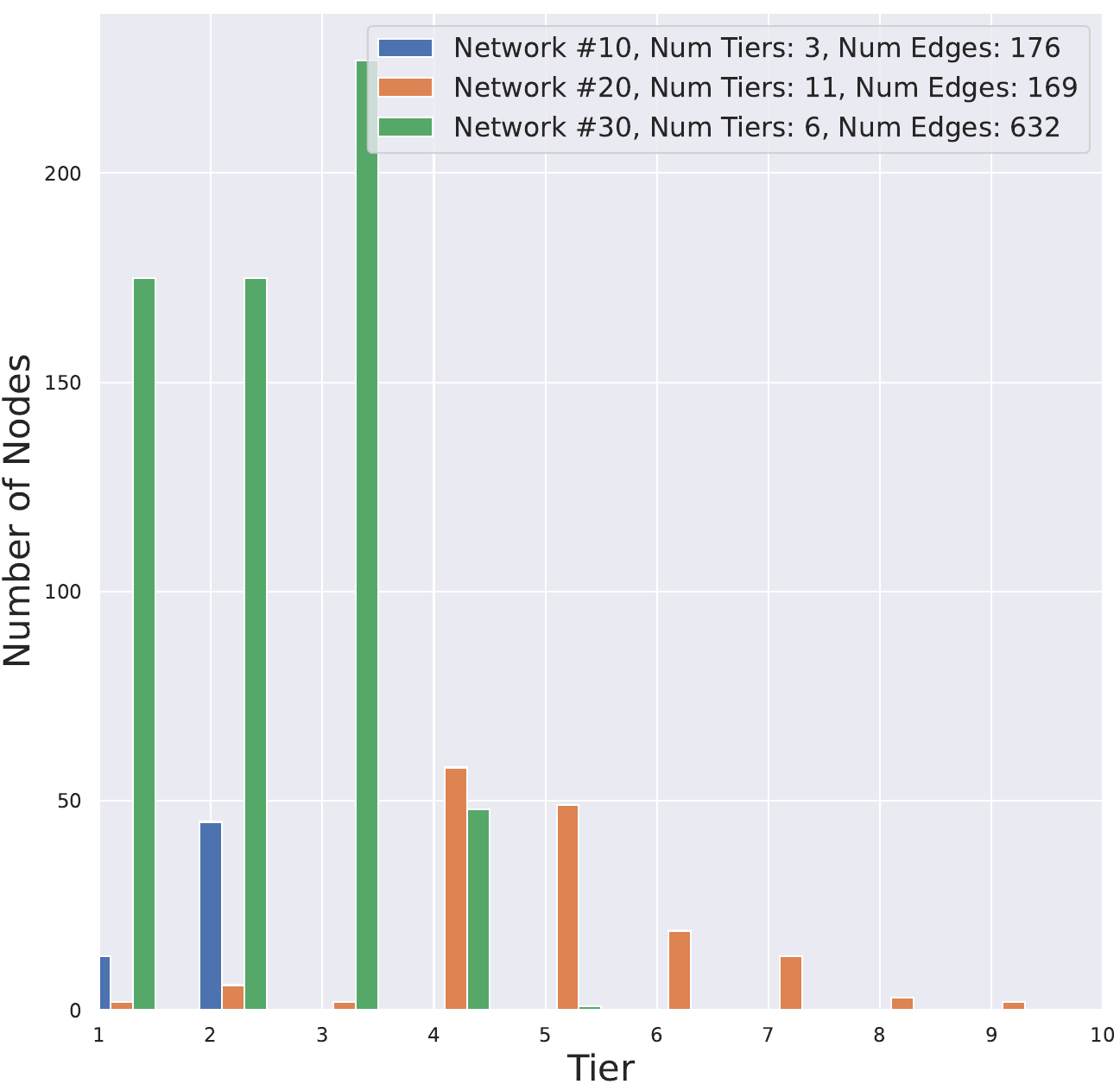}}
    \subfigure[Estimating $\hat R_{\cG}(\varepsilon)$ and $\Ravgsamp {\cG} (\varepsilon)$ \label{subfig:willems_resilience}]{\includegraphics[width=0.3\textwidth]{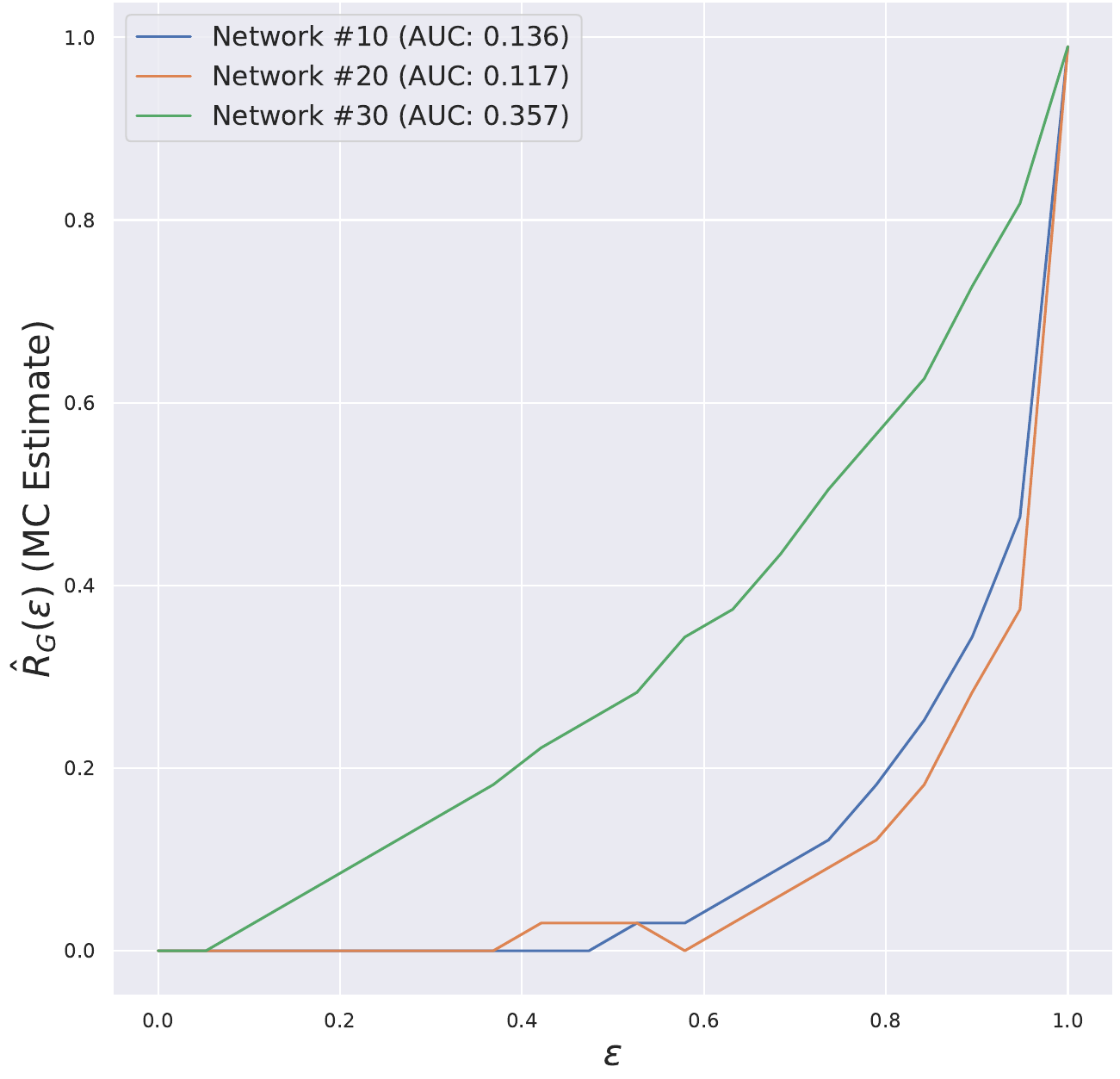}}
    \caption{Resilience estimation for three example networks from \citet{willems2008data}. We set the number of suppliers for each product to $n = 1$ and $y = 1$.}
    \shrink
    \label{fig:willems}
\end{figure}

\begin{table}[!h]
    \scriptsize
    \centering
    \begin{tabular}{llllllll}
    \toprule
         Network ID  & Size ($K$) & Avg. Degree & Density ($\frac {|\cE(\cG)|} {K^2 - K}$) & $\mu$ & $m$  & $\bar{R}_{\cG}$ & Confidence ($1 - 1/K$)  \\
    \midrule 
         \#10 & 58 & 3.03 & 0.053 & 27 & 13  & 0.136 & 98.27 \% \\
         \#20 & 156 & 1.08 & 0.006 & 29 & 3  & 0.117 & 99.35 \% \\
         \#30 & 626 & 1.00 & 0.001 & 2 & 48  & 0.357 & 99.84 \% \\ 
    \midrule
    \end{tabular}
    \caption{Network Statistics and $\bar{R}_{\cG}(y)$ estimated from Monte Carlo simulations for three out of the $38$ networks in \citet{willems2008data}'s dataset referring to the experiments of \cref{fig:willems} (with $n=1$ and $y=1$ as in \cref{fig:willems}).}
    \shrink
    \label{tab:statistics-willems}
\end{table}

\subsection{Impact of Correlation Structure}

In real-world settings, shocks are rarely independent; regional events or industry-wide crises create dependencies in failure probabilities. To further refine our empirical understanding, we analyze the effect of introducing supplier correlations on network resilience in three empirical networks of \citep{willems2008data}.

We chose the three smallest networks in \citep{willems2008data} to best illustrate the effect of increasing the correlation coefficient ($\rho$). Our simulations in \cref{fig:correlations} demonstrate that resilience decreases consistently as $\rho$ increases. The most pronounced decline occurs under joint supplier-product correlations, where shocks are coupled across both the vertical tiers and the horizontal sectors of the supply chain. In these scenarios, even networks that may appear structurally robust under i.i.d. failure assumptions quickly transition into fragile states. The data suggests that as $\rho$ approaches 1, the structural benefits of multi-sourcing are effectively neutralized, as the redundant suppliers fail in unison, leaving the network's directional dependencies exposed.

\begin{figure}[!h]
    \centering
    \includegraphics[width=0.9\linewidth]{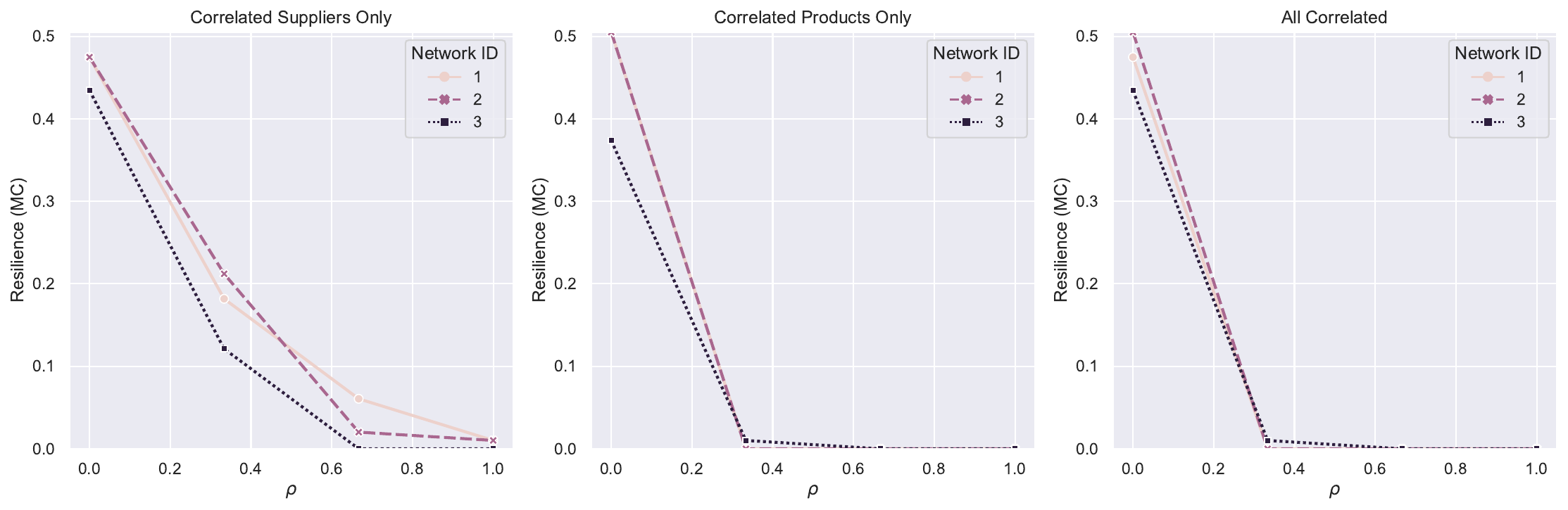}
    \caption{Effect of correlation on $R_{\cG}(\varepsilon)$ for the three smallest supply networks from~\cite{willems2008data}: Networks \#1 with $K = 8$, \#2 with $K = 13$, and \#3 with $K = 17$. For each case, the second-order correlation coefficient has been set to $\rho \in \{ 0, 0.25, 0.5, 0.75, 1 \}$, shown on the x-axis. Higher-order correlations are set to zero. We plot the Monte Carlo estimate $\hat R_\cG(\varepsilon)$ as a function of $\rho$. We have set $y = 1$, $\varepsilon = 0.2$ and $n = 10$. The correlation regimes are analytically given in \cref{app:generalized_resilience}. We chose Networks \#1, \#2, and \#3 due to their smaller sizes, which give a more gradual decline in resilience as $\rho$ increases and are more suitable for illustration.}
    \label{fig:correlations}
\end{figure}

\subsection{Structural Determinants of Resilience: Regression Analysis}

We further examine how the fundamental structural parameters identified in our theoretical sections, namely network size ($K$) and the number of raw materials ($r$), predict resilience across the entire 38-network dataset. The goal of this section is to demonstrate \textit{(i)} that the resilience metric is computationally tractable on large, empirically observed supply-chain networks, and \textit{(ii)} that it meaningfully discriminates across real-world network architectures in ways predicted by the theory. To achieve this, we perform a log-log regression to test the relationship between these topological features and the average resilience.

The results in \cref{tab:resilience_auc} provide empirical validation for our theoretical findings. We find that resilience decreases significantly as the number of raw materials ($r$) increases ($p < 0.01$). This confirms the ``Top Hat'' hypothesis: networks that rely on a broad, dispersed base of primary inputs are more susceptible to systemic failure. Interestingly, we observe that for these topologies, larger network size ($K$) correlates with higher resilience, suggesting that the growth in these specific multi-echelon networks often comes with increased modularity or depth rather than just width.

\subsection{Linking Resilience to the Risk Exposure Index (REI)}\label{sec:rei-resilience-connection}

Finally, we operationalize the connection between our resilience threshold and the Risk Exposure Index (REI). We average MC simulations to estimate the expected impact of the worst-case, single-node disruptions, denoted by $\reiavg {\cG} (y)$, and perform a log-log regression of $\reiavg {\cG} (y)$ against network resilience $\bar{R}_{\mathcal{G}}(y)$, size ($K$), and sourcing requirement ($m$). The negative coefficient for $\log \bar{R}_{\mathcal{G}}(y)$ ($p < 0.01$) in \cref{tab:resilience_ttr} demonstrates a strong inverse relationship between average resilience and risk exposure. This implies that networks characterized by high systemic risk (low $\bar{R}_{\cG}$) are also those where individual nodes possess extreme disruptive potential. Furthermore, the positive relationship with $\log(m)$ confirms that denser connectivity increases the maximum possible cascade size.

Finally, to complement the above results, in \cref{app:experiments_addendum}, we also give results from I-O output tables of countries taken from the World Input-Output Database \citep{timmer2015illustrated}.

\begin{table}[t]
\begin{minipage}[t]{.49\textwidth}
\centering
\tiny
\begin{tabular}{@{\extracolsep{3pt}}lc}
\\[-1.8ex]\hline
\hline \\[-1.5ex]
& \multicolumn{1}{c}{\textit{Dependent variable: $\log$-$\bar{R}_{\mathcal{G}}$}} \
\cr \cline{2-2}
\hline \\[-1.5ex]
 const & -1.757$^{***}$ (0.045) \\
 log-$K$ & 0.146$^{***}$ (0.012) \\
 log-$r$ & -0.029$^{***}$ (0.010) \\
\hline \\[-1.5ex]
 Observations & 34 \\
 $R^2$ & 0.871 \\
 Adjusted $R^2$ & 0.862 \\
 Residual Std. Error & 0.067 ($df=31$) \\
 $F$ Statistic & 104.385$^{***}$ ($df=2; 31$) \\
\hline
\hline \\[-1.5ex]
\textit{Note:} & \multicolumn{1}{r}{$^{*}p<0.1; ^{**}p<0.05; ^{***}p<0.01$} \\
\hline \\
\end{tabular}
\caption{Relation between the $\bar{R}_{\mathcal{G}}$, $K$ and $r$ for the 38 multi-echelon networks of \cite{willems2008data}. For each network, we set $y = 1 / (\max\{m, \mu\} + 10^{-5})$.}
\label{tab:resilience_auc}
\end{minipage}%
\begin{minipage}[t]{0.49\textwidth}
\tiny
\centering
\begin{tabular}{@{\extracolsep{3pt}}lc}
\\[-1.8ex]\hline
\hline \\[-1.5ex]
& \multicolumn{1}{c}{\textit{Dependent variable: $\log$-$\reiavg {\cG} (y)$}} \
\cr \cline{2-2}
\hline \\[-1.5ex]
 const & -14.782$^{***}$ (2.817) \\
 log-$m$ & 0.540$^{***}$ (0.167) \\
 log-$K$ & 0.932$^{***}$ (0.259) \\
 log-$\bar{R}_{\mathcal{G}}$ & -8.828$^{***}$ (1.589) \\
\hline \\[-1.5ex]
 Observations & 34 \\
 $R^2$ & 0.678 \\
 Adjusted $R^2$ & 0.646 \\
 Residual Std. Error & 0.629 ($df=30$) \\
 $F$ Statistic & 21.089$^{***}$ ($df=3; 30$) \\
\hline
\hline \\[-1.5ex]
\textit{Note:} & \multicolumn{1}{r}{$^{*}p<0.1; ^{**}p<0.05; ^{***}p<0.01$} \\
\hline\\
\end{tabular}
\caption{Relation between $\log \reiavg {\cG} (y)$ and $\log m$, $\log K$ and $\log\bar{R}_{\mathcal{G}} (y)$ for the 38 multi-echelon networks studied by \cite{willems2008data}. For each network, we set $y = 1 / (\max\{m, \mu \} + 10^{-5})$.} 
\label{tab:resilience_ttr}
\end{minipage}
\shrink
\end{table}

\section{Discussion and Conclusion} \label{sec:discussion}

Across the networks studied, production systems differ sharply in how they absorb or amplify disruptions \citep{ivanov2020viability,nair2009supply,elliott2023supply}. Some structures contain failures, while others allow small shocks to cascade into large production losses. These contrasts persist even among networks that appear similar in size or sourcing effort, mirroring mechanisms long studied in financial contagion and systemic risk theory, where interdependencies determine whether shocks dissipate or cascade; see, for instance, \citet{eisenberg2001systemic,chen2013axiomatic}. Viewed in this way, resilience reflects underlying differences in dependency design, cascade dynamics, and threshold behavior rather than the presence or absence of individual protection.

At the network level, resilience depends on whether dependency structures limit or amplify the cascading propagation of failures, rather than on the reliability of any single source. Failures matter less on their own and more because of how easily they cascade through the network. When dependencies are dense and broadly distributed upstream, disruptions travel quickly and widely across tiers, enabling large failure cascades. In contrast, networks that constrain critical inputs and align product scope with available resources limit cascading pathways. The structure of dependency, not the strength of individual nodes, governs whether failures cascade or dissipate.  

The results further indicate a sharp divide between resilient and fragile network architectures rather than smooth tradeoffs in performance \citep{elliott2022supply}. Some structures remain functional under nontrivial disruption, while others are such that even small shocks trigger widespread cascading failure as networks scale. Importantly, this fragility reflects threshold behavior: once dependencies exceed certain structural limits, disruptions shift abruptly from being contained to becoming system-wide. This helps explain why incremental interventions, such as adding sourcing options or modest buffers, often fail to change outcomes. When fragility is embedded in the architecture, local improvements have little impact on system behavior compared to structural differences in network design. A striking feature of these patterns is how they evolve as production networks grow. Structures that appear stable on smaller scales can become fragile as complexity increases, which helps explain why supply chains often become more vulnerable over time, despite incremental improvements.

Network shape offers a concise way to interpret these regime differences: Broad upstream dependence combined with expansive downstream product proliferation concentrates systemic risk and facilitates cascades. For example, when many products rely on distinct raw materials or when a single disruption disables multiple downstream variants, failures quickly propagate across the network. More focused architectures with controlled input and output widths restrict the channels through which failures spread, limiting the scope of cascades even under comparable shocks (e.g., when a component affects a small set of downstream variants rather than many). This perspective helps unify insights across different network models and empirical settings \citep{levi2016identifying,nair2009supply}. Network shape thus provides an intuitive explanation for persistent patterns of fragility in complex production systems. Similar shape-driven fragility has been documented in financial contagion and systemic risk theory, where network topology governs whether shocks remain localized or spread system-wide \citep{glasserman2015likely,acemoglu2015systemic}. 

Much of the existing supply chain risk literature emphasizes direct supplier failures and local disruptions, often evaluating vulnerability at the level of individual firms or tiers. The results here suggest that this perspective misses a critical driver of systemic risk:  Cascading effects, rather than direct supplier losses, drive large production disruptions by allowing indirect failures to propagate through dependency structures well beyond the initial shock. The results also distinguish between exposure to disruption and the amplification of disruption once it occurs. Resilience in this framework is less about preventing shocks altogether and more about limiting how much those shocks propagate through the network and escalate.  This explains why assessments focused on individual input sources or tiers tend to underestimate exposure when failure propagation dominates, making a system-wide perspective essential \citep{zhao2018supply}. This perspective is reinforced by treating resilience as a measurable system property rather than an informal label. Doing so clarifies how differences in dependency design and cascade dynamics translate into meaningful differences in tolerance to disruption. 

From a systemic perspective, correlated failures highlight important limits of conventional diversification logic. Much of the existing intuition assumes that expanding the number of sources of supply reduces risk, as seen, for example, in \citet{craighead2007severity}, by averaging out independent failures; however, this reasoning breaks down when sources share common exposure. The results show that even a modest correlation at the failure level can sharply reduce resilience by synchronizing disruptions and pushing networks past cascade thresholds. In this sense, correlation does not merely increase risk smoothly; it can fundamentally alter the behavior of the network. Accounting for shared risk sources aligns the analysis more closely with real disruption patterns, where shocks often affect input sources jointly rather than independently. Independence among sources of supply, therefore, matters as much as their number, particularly in networks otherwise designed to suppress cascades.

The discussion frames resilience as a property that emerges from deliberate dependency design rather than isolated protections or ad hoc responses. Resilience is not merely a qualitative descriptor, but a property that can be systematically compared across network designs by how they shape cascade dynamics and respond to structural thresholds. From this perspective, production networks do not simply experience shocks; they actively determine whether disruptions dissipate or amplify through their dependency structures. Viewing resilience in this way shifts attention from reacting to individual failures toward designing systems that suppress cascades and avoid threshold-driven collapse. As production networks continue to grow in scale and complexity, adopting a system-wide view of resilience becomes increasingly consequential.

In summary, we present a framework that helps supply chain managers assess systemic risk, identify structural vulnerabilities, and prioritize targeted interventions before disruptions occur. Our results yield several theoretical contributions and actionable managerial insights.

\subsection{Theoretical Contributions}

\xpar{In-depth analytical and empirical insights about Resilience} Our developed resilience index provides in-depth analytical and empirical tools to assess the resilience of production networks and to characterize structurally resilient versus fragile networks, going beyond well-established simulation-based approaches that rely on specific disruption scenarios or parameter choices, such as the Risk Exposure Index \citep{levi2016identifying,simchi2014superstorms}.

In contrast, our index admits a closed-form characterization that directly links resilience to fundamental structural network characteristics, enabling comparative statics and theoretical insights that are not accessible through purely simulation-based methods, as studied in several notable works, for instance, \citet{nair2011supply}. This analytical tractability allows us to formally identify the structural conditions under which resilience deteriorates.

More broadly, our framework reframes resilience as an intrinsic property of the production network rather than an outcome of particular shock realizations. This perspective enables ex ante evaluation of network designs and highlights how common structural features observed in real-world supply chains, such as tiered sourcing and asymmetric input dependencies, can endogenously generate heavy-tailed cascade risks.

\xpar{Operationalizing Resilience with Linear Programming}
We develop a fast, scalable LP-based method to estimate cascade size and resilience even in large networks, which may have cyclical dependencies. Researchers can use this to design studies involving multi-faceted \textit{``what-if''} scenarios, estimate the resilience of the network, and study more resilient networks or evaluate investment decisions in existing ones. Contrary to other recent works, see, for instance \citet{chen2023when} and the references therein, our formulation is tractable (easier to compute) and can accommodate a variety of configurations.

\xpar{Accounting for exogenous correlations} Past literature has considered supplier diversification or network structure as an approach to managing risk \citep{elliott2022supply,gabaix2011granular,nair2011supply,kim2015supply}, but has largely overlooked the possibility of the suppliers being correlated at the failure level beyond correlation purely due to structure. We extend this line of research by incorporating a correlation structure among the suppliers. This extension allows researchers to examine more complex scenarios in which the existence of correlations is exogenous to the supply network structure (for example, due to geography and/or geopolitics). This matters because correlated failures can undermine standard risk mitigation strategies, making networks that appear diversified far more fragile in practice. 

\xpar{Implications for network design} Our results show that supply chain resilience depends not only on local characteristics such as node capacity or path length, but also on global structural features, particularly the diversity of raw materials and the degree of sourcing dependencies. When these features become too large, disruptions are more likely to trigger cascading failures, even in well-capacitated networks.

These insights clarify resilience strategies for common network structures. In diamond-shaped supply chains, where multiple products rely on shared upstream suppliers \citep{ang2017disruption,wang2021ignorance}, protecting suppliers in the tiers with the most nodes is effective, as failures in those suppliers disrupt many production paths, due to the fact that the induced sub-networks, including all downstream tiers, are fragile, as they fall under the lineage growth topology. Moreover, in scale-free networks, which rely on a few highly connected hubs \citep{hearnshaw2013complex,kim2015supply}, inherent vulnerabilities arise from the number of raw materials \citep{albert_error_2000,borgatti2009social}. This is because, in the directed preferential attachment model, the number of nodes with zero in-degree, which correspond to the raw materials, is a constant fraction of $K$ \citep{tanimoto2009power}. Finally, small-world networks \citep{watts1998collective}, characterized by high clustering and short paths, are shown to respond better to shocks, consistent with prior literature \citep{basole2014supply}. This is because their short average path length reduces the risk of raw products participating in long supply paths, yielding a small sum of the Katz centralities for the nodes, which, in turn, increases the lower bound on resilience, as shown in \cref{theorem:lp_duality_resilience_2}. 

\subsection{Managerial Insights}

\xpar{Design for resilience with network shape in mind} 
Our theoretical results show that ``rolling pin'' shaped networks, those with a narrow, controlled number of inputs and outputs, are more resilient than ``top hat'' structures with broad raw material bases; see \cref{fig:qualitative_representation}. This suggests that managing the width of both ends of the production network (input diversity and product proliferation) is key.

\textit{SK Hynix's fire} incident at their Wuxi semiconductor plant in September 2013 underscores this extreme concentration risk. By relying on just two fabs for their entire DRAM output, a single fire at one location disrupted 15--30\% of the world’s supply, causing global price spikes and delays for smartphone and PC manufacturers. %Reducing the variety of critical inputs and avoiding such extreme centralization at key architectural points can increase resilience without compromising flexibility.
Careful consideration of the network can help increase resilience: Reducing the variety of critical inputs, avoiding such extreme centralization at key architectural points, and aligning product proliferation with available resources can increase resilience without compromising flexibility. 

\xpar{Apply the proposed Resilience metric as a quantitative risk index} To use the framework in practice, the idea is straightforward: first, map the network at a high level to capture factors like how many inputs each product depends on and how much overlap there is across suppliers; then use those features to see whether the network is operating in a stable or fragile range.

The resilience metric provides a rigorous measure of the extent of disruption a supply chain can tolerate while maintaining functionality. Managers can use this as a stress-testing tool to evaluate whether their network can withstand, say, a 10\% supplier failure rate while still producing 95\% of goods. If the metric is low, it signals the need to reinforce vulnerable nodes, build inventory, or reduce complexity. In contrast, when the metric is high, it indicates that the network can withstand significant disruptions without performance loss, specifically without cascading failures that lead to substantial production losses.

\xpar{Ensure supplier diversification accounts for correlated risks}
We extend our analysis to networks with correlated failures and show that when suppliers share risk factors (e.g., geographic or geopolitical proximity and shared raw inputs), the network becomes significantly more fragile. Therefore, managers should not only diversify among suppliers, but also ensure that those suppliers are \emph{independently resilient}. Geographical, operational, or organizational diversification can be the key to avoiding common-mode failures. 

A pertinent example is the \emph{2022 Semiconductor Neon Gas Supply Disruption}: the Russia–Ukraine conflict cut off about 50\% of the world’s semiconductor-grade neon gas virtually overnight. Because Ukraine’s leading producers (Ingas and Cryoin) were concentrated in a single conflict zone, this disruption created a cascading failure risk across global chip fabrication, illustrating how a geopolitically exposed source can undermine high-tech manufacturing networks despite local stockpiles.

\subsection{Future Research Directions and Limitations}

This paper develops a structural framework for understanding resilience in production networks, showing how cascading failures and systemic fragility emerge from network architecture rather than isolated disruptions. Our framework provides a simple way to diagnose structural vulnerability early, before detailed simulation or full operational modeling is required.

By characterizing resilient and fragile regimes in terms of global structural features, the analysis explains why some networks absorb shocks, while others amplify them. Because the framework intentionally abstracts from many operational details to isolate the role of structure and dependency, it also suggests several directions for future research.

An avenue is to extend the framework beyond production networks to other interconnected economic systems. Financial networks, input–output economies, and innovation ecosystems exhibit similar dependency structures in which shocks propagate across interlinked agents. Applying the resilience metric in these contexts could help determine whether the structural regimes identified here also govern systemic risk in other domains, or else their varying dependency structures give rise to distinct fragility patterns.

The second direction involves incorporating economic and operational primitives such as costs, capacities, quantities, and prices. The current model focuses on the feasibility and structural limits of production, rather than optimal operation. Embedding the resilience metric within optimization or simulation models would allow researchers to study how structural fragility constrains efficiency and profitability under disruption and to analyze trade-offs between cost efficiency and resilience.

 The resilience metric captures the ability of a network to absorb shocks and avoid large-scale cascades, but does not model post-disruption adjustment, learning, or reconfiguration. The present framework abstracts away from recovery dynamics and endogenous firm behavior, and its analysis relies on simplifying assumptions about static firm behavior. Future empirical and computational work could examine how firms respond to structural vulnerability over time and whether heterogeneity in product requirements, decentralized responses, as well as endogenous supplier selection, diversification, and buffering decisions, amplify or dampen cascading behavior and systemic vulnerability relative to the metrics identified here.

Finally, these directions highlight how the structural perspective developed here can benchmark future empirical testing and computational modeling of resilience in interconnected systems.

\ACKNOWLEDGMENT{% \section*{Acknowledgments} 

MP was partially supported by a LinkedIn Ph.D. Fellowship, a grant from the A.G. Leventis Foundation, and a grant from the Gerondelis Foundation. MAR was partially supported by the National Science Foundation under Grant
No. 2424684. 

\noindent Data and code can be found at: 

\begin{center}
    \url{https://github.com/papachristoumarios/supply-chain-resilience}
\end{center}

}

%%REFERENCES%%
%%%%%%%%%%%%%%%%%%%%%%%%%%%%%%%%%%%%%%%%%%%%%%%%%%%%%%%%%%%%%%%%%%%%%%%%%%%%%%%%%%%%%%%%%%%%%%%%%%%%%%%%%%%%%%%%%%%%%%%%%%%%%%%%%%%%
%% This template complies references using bibtex. You will need to use pomsref.bst file for biblography style.
%REFERENCES USING BIBTEX FILES
%%%%%%%%%%%%%%%%%%%%%%%%%%%%%%%%%%%%%%%%%%%%%%%%%%%%%%%%%%%%%%%%%%%%%%%%%%%%%%%%%%%%%%%%%%%%%%%%%%%%%%%%%%%%%%%%%%%%%%%%%%%%%%%%%%%%

\bibliographystyle{pomsref}

 \let\oldbibliography\thebibliography
 \renewcommand{\thebibliography}[1]{%
    \oldbibliography{#1}%
    \baselineskip14pt %Change this for line spacing within the same reference
    \setlength{\itemsep}{10pt}% %Change this for spacing between two referneces
 }
\bibliography{references}

%% Here starts the e-companion (EC). Place your appendix content here.
%%%%%%%%%%%%%%%%%%%%%%%%%%%%%%%%%%%%%%%%%%%%%%%%%%%%%%%%%%
\ECSwitch % Comment this line out if you do not need e-companion.
%%%%%%%%%%%%%%%%%%%%%%%%%%%%%%%%%%%%%%%%%%%%%%%%%%%%%%%%%%

%%% Main head for the e-companion
\ECHead{E-Companion for ``Structural Measures of Resilience for Supply Chains," Papachristou et al.}
%\ECHead{E-Companion}% for POM Journal Template
\crefalias{section}{appendix}
\crefalias{subsection}{appendix}
\crefalias{subsubsection}{appendix}

% ===================== EC-only TOC ======================
\makeatletter

% Keep original addcontentsline
\let\EC@origaddcontentsline\addcontentsline

% In the EC, write ONLY \section entries to a separate file (.ectoc).
% This prevents main-paper entries from appearing and suppresses figures/tables.
\renewcommand{\addcontentsline}[3]{%
  \def\EC@type{#2}%
  \def\EC@section{section}%
  \ifx\EC@type\EC@section
    \EC@origaddcontentsline{ectoc}{#2}{#3}%
  \fi
}

% Print the EC TOC from ectoc (with hyperlinks), and widen number box so "EC1" doesn't overlap titles
\newcommand{\ECtableofcontents}{%
  \section*{E-Companion Table of Contents}
  % \begingroup
  %   \setcounter{tocdepth}{1}% sections only
  %   % The "3em" is the width reserved for the section number (EC1, EC12, etc.)
  %   % Increase it if you still see overlap.
  %   \renewcommand*\l@section{\@dottedtocline{1}{0em}{3em}}%
  %   %\NoHyper
  %   \@starttoc{ectoc}%
  %   %\endNoHyper
  % \endgroup
  % --- EC ToC: link only the section number (EC.1, EC.2, ...)
\begingroup
  % 1) disable hyperref's default "whole line" ToC links (locally)
  \hypersetup{linktoc=none}

  % 2) keep your formatting
  \setcounter{tocdepth}{1}%
  \renewcommand*\l@section{\@dottedtocline{1}{0em}{3em}}%

  % 3) make only the numberline clickable, using the current ToC destination
  \makeatletter
  \renewcommand*\numberline[1]{%
    \hb@xt@\@tempdima{\hyperlink{\Hy@tocdestname}{##1}\hfil}%
  }%
  \makeatother

  \@starttoc{ectoc}%
\endgroup
}

\makeatother

% \renewcommand\thesubfigure{EC.\arabic{figure}(\alph{subfigure})}
% \renewcommand\p@subfigure{EC.\thefigure}
% \renewcommand{\thesubfigure}{EC.\arabic{figure}}

% \renewcommand{\thefigure}{A-\arabic{figure}}
% \renewcommand{\thefigure}{EC.\arabic{figure}}
% \crefname{figure}{Figure EC}{Figures EC}
% \renewcommand\thefigure{EC.\arabic{figure}} 
% \renewcommand\thefigure{\thesection.\arabic{figure}}  
%If you don't use BiBTex, you can manually itemize references as shown below.

% Use this for adding Appendix
% Options are (1) APPENDIX (with or without general title) or
%             (2) APPENDICES (if it has more than one unrelated sections)
% Outcomment the appropriate case if necessary
%
% \begin{APPENDIX}{<Title of the Appendix>}
% \end{APPENDIX}
%
%   or

% \begin{APPENDICES}

% --- Make \cref to subfigures show EC.<n>(a) instead of <n>(a) ---
\makeatletter
\renewcommand{\p@subfigure}{\thefigure}      % prefix for subfigure references
\renewcommand{\thesubfigure}{(\alph{subfigure})}
\makeatother

%\section*{E-Companion Organization}

\noindent This E-Companion is organized as follows: \cref{app:notation_table} provides a notation table for the quantities used throughout the paper. \cref{app:dichotomy} presents equivalent characterizations of fragility. Next, \cref{app:examples} discusses the first case of a resilient network (the two-layer network), the reachability collapse, and the emergence of power laws. Afterwards, \cref{app:main_result} provides a proof of the main result of the paper that characterizes networks into resilient (with fixed sourcing requirement, sourcing influence, fixed raw materials, and finished goods) and fragile (large base of raw materials). Then, \cref{app:architectures} introduces several network architectures and the regime to which each belongs. \cref{app:linear_program} develops the principles behind the LP-based formulation of the resilience and proves the main results, and \cref{app:rei} establishes connections to the REI. Finally, \cref{app:generalized_resilience} provides general definitions of resilience that account for exogenous correlation structures. Finally, \cref{app:experiments_addendum} provides additional experiments with input-output networks for different countries (world economy I-O tables). 

\ECtableofcontents

\clearpage

\section{Notation Table} \label{app:notation_table}

\begin{table}[!h]
\centering
\footnotesize
\caption{Summary of key notation} \label{tab:notation}
\begin{tabular}{@{}ll@{}}
\toprule
Symbol & Meaning \\ \midrule
$[K]$ & $K = \{ 1, \dots, K \}$ \\ [2pt]
$\| v \|_p$ or $\| V \|_p$ & $p$-norm of vector/matrix ($p = 2$ if $p$ omitted from notation \\ [2pt]
$\zero$ / $\one$ & column vector of all-0s/all-1s \\ [2pt] 
$\zero_S$ / $\one_S$ & column vector where the entries belonging to $S$ are zero/one \\ [2pt] 
$\min\{x , y\}$ & coordinate-wise min between vector $x, y$ \\ [2pt]
$\max\{x , y\}$ & coordinate-wise max between vector $x, y$ \\ [2pt]
 $\ge, \le, >, <$ & coordinate-wise ordering \\ [2pt]
 $\wedge , \vee$ & logical operations AND and OR \\ [2pt]
 $x_n \asymp y_n$ & $\lim_{n \to \infty} x_n / y_n = 1$ \\ [2pt] \midrule

$\cG$ & Production network (network of supply relationships) \\ [2pt]
$\cK$ & Set of products \\ [2pt]
$\cS(i)$ & Suppliers of product $i$ \\ [2pt]
$\cN(i)$ & Inputs of product $i$ \\ [2pt]
$\cG^R$ & Reverse production network (network of sourcing relationships) \\ [2pt]
$K$ & Total number of products (nodes) in the production network  \\[2pt]
$M$ & Total number of edges in the production network \\ [2pt]

$r$ & Number of \emph{raw products} (source nodes) \\[2pt]

$c$ & Number of \emph{finished goods}  \\[2pt]

$n_i$ & Number of suppliers per product. In the simple model $n_i = n$ \\[2pt]

$m$ & \textit{Sourcing requirement:} maximum \# of inputs (in-degree) any product requires  \\[2pt]

$\mu$ & \textit{Sourcing influence:} maximum \# of outputs (out-degree) any product influences  \\[2pt] \midrule

$x\in(0,1)$ & Per‑supplier failure probability \\[2pt]

$F$ & Random variable: number of failed products after cascading effects finish \\[2pt]

$S = K-F$ & Number of products that remain producible after the cascade \\[2pt]

$\varepsilon$ & Tolerable fraction of product failures in the resilience definition \\[2pt]

$R_\cG(\varepsilon)$ & \emph{Resilience metric of $\cG$ given fixed fraction $\varepsilon$}: largest $x$ s.t.\ $\Pr[F \le \varepsilon K] \ge 1-\tfrac1K$ \\[2pt]

\bottomrule
\end{tabular}
\end{table}

\newpage

\section{Equivalent characterizations of fragility} \label{app:dichotomy}

The resilience $R_{\cG}(\varepsilon)$ is the highest percolation probability such that at least $(1 - \varepsilon) K$ products survive with probability at least $1 - 1/K$. The dichotomy between resilient and fragile networks is strengthened by noting that if $\lim_{K \to \infty} R_{\cG}(\varepsilon) = 0$ for some fixed $\varepsilon \in (0, 1)$, then $\lim_{K \to \infty} R_{\cG}(\varepsilon) = 0$ for any fixed $\varepsilon \in (0, 1)$. Formally:

\begin{proposition}[Equivalent characterizations of fragility] \label{prop:diochotomy}
    If $\lim_{K \to \infty} R_{\cG}(\varepsilon) = 0$ for some fixed $\varepsilon \in (0, 1)$, then $\lim_{K \to \infty} R_{\cG}(\varepsilon) = 0$ for any fixed $\varepsilon \in (0, 1)$.
\end{proposition}

\xpar{Proof} Suppose there exists some fixed \(\varepsilon_0 \in (0,1)\) such that $\lim_{K \to \infty} R_\cG(\varepsilon_0) = 0$. We aim to show that for any other \(\varepsilon \in (0,1)\), it also holds that $\lim_{K \to \infty} R_\cG(\varepsilon) = 0$.  First, observe that \(R_\cG(\varepsilon)\) is monotonically non-decreasing in \(\varepsilon\). Hence, for all \(\varepsilon < \varepsilon_0\), we have
\[
0 \le R_\cG(\varepsilon) \leq R_\cG(\varepsilon_0).
\]
Taking the limit as \(K \to \infty\), we get
\[
\lim_{K \to \infty} R_\cG(\varepsilon) \leq \lim_{K \to \infty} R_\cG(\varepsilon_0) = 0,
\]
which implies
\[
\lim_{K \to \infty} R_\cG(\varepsilon) = 0 \quad \text{for all } \varepsilon < \varepsilon_0.
\]

Now, suppose for the sake of contradiction that there exists \(\varepsilon_1 > \varepsilon_0\) such that
\[
\limsup_{K \to \infty} R_\cG(\varepsilon_1) > 0.
\]
Then there exists some percolation probability $x_0 \in (0, 1)$ and an infinite subsequence $\{K_m\}$ such that
\[
\Pr_{x_0}[S \geq (1 - \varepsilon_1) K_m] \geq 1 - \frac{1}{K_m}.
\]
Since $(1 - \varepsilon_0)K_m < (1 - \varepsilon_1)K_m$, it follows that
\[
\Pr[S \geq (1 - \varepsilon_0) K_m] \geq \Pr[S \geq (1 - \varepsilon_1) K_m] \geq 1 - \frac{1}{K_m}.
\]
This implies that \(x_0\) is feasible for \(R_\cG(\varepsilon_0)\) for infinitely many \(K_m\), contradicting the assumption that \(\lim_{K \to \infty} R_\cG(\varepsilon_0) = 0\).

Therefore, our assumption must be false, and it must hold that: 
\[
\lim_{K \to \infty} R_\cG(\varepsilon) = 0 \quad \text{for all } \varepsilon \in (0,1).
\]

\newpage

\section{Driving factors of resilience and fragility} \label{app:examples}

\begin{figure}[!h]
    \centering
    \subfigure[Two-tier network (resilient) \label{fig:parallel_products}]{
    \begin{tikzpicture}[]
        \Vertex[x=-1,y=0,Pseudo]{rz}
        \Vertex[x=4,y=0,Pseudo]{rzz}

        \Vertex[x=-0.25,y=0,Pseudo,label={$\cR$}]{r0}
        \Vertex[x=4.25,y=0,Pseudo]{x}

        \Vertex[x=1,y=0,color=rose]{r1}
        \Vertex[x=2,y=0]{r2}
        \Vertex[x=3,y=0]{r3}

        \Vertex[x=-0.25,y=1.25,Pseudo,label={$\cC$}]{c0}

        \Vertex[x=1.5,y=1.25,color=pink]{c1}
        \Vertex[x=2.5,y=1.25,color=pink]{c2}

        \Edge[color=black,Direct](r1)(c1)
        \Edge[color=black,Direct](r2)(c1)
        \Edge[color=black,Direct](r3)(c1)
        \Edge[color=black,Direct](r1)(c2)
        \Edge[color=black,Direct](r2)(c2)
        \Edge[color=black,Direct](r3)(c2)

    \end{tikzpicture}}
    \subfigure[Hierarchical one-directional input-output network (fragile) \label{fig:random_dag_2}]{
    \begin{tikzpicture}[transform shape]
        \Vertex[x=-2, y=0, label=$1$, color=rose]{u1}
        \Vertex[x=0, y=0, label=$2$]{u2}
        \Vertex[x=2, y=0, label=$3$, color=pink]{u3}
        \Edge[Direct, color=black, label=$1- p$, style={dashed}](u1)(u2)
        \Edge[Direct, color=black, label=$p$, ](u2)(u3)
        \Edge[Direct, color=black, label=$p$, bend=30](u1)(u3) 
    \end{tikzpicture}} \quad

    %\shrink
    \caption{A resilient network (a) and a fragile network due to reachability collapse (b). The spontaneous failures (sources of failure) are colored in dark red, the resulting failures due to cascading effects are colored pink, and the unaffected nodes are colored light blue.}
    \shrink
    \label{fig:parallel_products_rdag}
\end{figure}

\subsection{Two-tier network}

To show that there is at least a resilient network, we rely on a two-tier network (cf. \cite{bimpikis2019supply,willems2008data} for related work that motivates this architecture). Here, our objective is to produce a set $\cC$ of finished goods, and each requires $m$ inputs.  We also introduce supply dependencies among raw materials, assuming that each raw material can supply $\mu$ products. \cref{fig:parallel_products} shows an example of this network together with an instance of the percolation process (the affected nodes are colored pink). Here, it is interesting to study both the resilience of the whole graph, i.e., the graph with vertex set $\cC \cup \cR$, as well as the resilience of the finished goods $\cC$ alone. We show that if the source requirement $m$ and the sourcing influence $\mu$ between the products are independent of $K$, then the production network is resilient. The resilience metric is lower bounded by $\left ( \frac {\varepsilon} {2 (\mu+1)m} \right )^{1/n}$ in both cases (finished goods alone or together with raw materials), as the number of products goes to infinity. 

We now state the result for the two-tier network (proof in \cref{app:proof:theorem:parallel_products}):

\begin{theorem} \label{theorem:parallel_products}
    Let $\cG$ consist of a two-tier network with $c$ finished goods and assume that $r = o(K^{2/3})$ raw materials can produce these products, and each raw material supplies at most $\mu$ finished goods, and each final good requires at most $m$ raw materials. 
    Then, the resilience satisfies $$R_{\cG}(\varepsilon) \ge \left ( \frac {\varepsilon} {2 (\mu + 1)m} + \sqrt {\frac {\log K} {2mK}} \right )^{1/n}.$$ Subsequently, if $\varepsilon$, $\mu$, and $m$ are constant with respect to $K$, then the network is resilient, with resilience scaling as $\Omega\left( \left( \frac{\varepsilon}{\mu m} \right)^{1/n} \right)$.
\end{theorem}

\subsubsection{Proof of \texorpdfstring{\cref{theorem:parallel_products}}{theorem:parallelproducts}: Two-tier network}\label{app:proof:theorem:parallel_products}

Let $F_{\cR}$ (resp. $F_{\cC}$) be the number of failed raw materials (resp. complex products). The quantity
\begin{align}
    z = \sup \{ x : (0 , 1) : \Pr [F_{\cC} \ge \varepsilon K] \le 1 / K \},
\end{align} is a lower bound to the resilience $R_{\cG}(\varepsilon)$. We have that $\{ F_{\cC} \ge \varepsilon K \} \implies \{ F_{\cR} \ge {\varepsilon K} / \mu \} $. Let $\delta = \frac {1} {z^n} \sqrt {\frac {\log K} {2r}}$ and let $\frac {\varepsilon K} {\mu} = (1 + \delta) \ev {} {F_R} = (1 + \delta) r z^n$. We apply the one-sided Chernoff bound and get $\Pr [F_{\cC} \ge \varepsilon K] \le \Pr \left [F_{\cR} \ge \frac {\varepsilon K} {\mu} \right ] = \Pr \left [F_{\cR} \ge (1 + \delta) \ev {} {F_{\cR}} \right ] \le e^{-2 \delta^2 \ev {} {F_{\cR}}^2 / r} = \frac 1 K$. Finally, by resolving the last equation $(1 + \delta)r z^n = \frac {\varepsilon K} {\mu}$, we get that $z = \left ( \frac {\varepsilon K} {r \mu} + \sqrt {\frac {\log K} {2 r}} \right )^{1/n}$. Also, we have that $r \le mK$ and therefore $R_{\cG}(\varepsilon) \ge z \ge \left ( \frac {\varepsilon} {\mu m} + \sqrt {\frac {\log K} {2mK}} \right )^{1/n}$.

\subsection{Cascading failures and emergence of power laws in  random hierarchical, one-directional input-output production networks} \label{app:power_laws}

Networks as simple as a hierarchical one-directional input-output random production network can exhibit cascade sizes that follow a power law, namely, the average cascade size $F$ is dominated by a few very large cascades rather than the many smaller ones. Our motivation stems from the network science literature \cite{thadakamaila2004survivability,albert_error_2000,wegrzycki2017cascade,dobson2005loading,nair2022fundamentals}, as well as the supply chain resilience literature, which has extensively studied random networks such as preferential attachment networks and random graphs~\citep{kim2015supply,magnanti_inventory_2006,nair2011supply}, and multi-tier networks \citep{bimpikis2019supply}.

Production networks typically exhibit a ``hierarchical order'', where complex products depend on the supply of simpler ones. In a network, raw materials and component parts are typically transformed into intermediate products and finished products through a series of production processes. For example, the production of a car engine may depend on the availability of simpler components such as computer chips. In its simplest form, this behavior can be captured by a hierarchical one-directional production network. We control the density of the network by a density parameter $p \in (0, 1)$. More specifically, the network is constructed as follows: The $K$ products are ordered from $1$ to $K$, and each product $i$ can tentatively depend on all the products before it. For a product $i$, to determine whether it depends on a product $j \le i$, a coin of bias $p$ is flipped. This structure yields a random acyclic graph network, which can be equivalently represented as a multi-tier network.

The following theorem characterizes the distribution of cascade size $F$, showing that it admits a power-law lower bound with exponent one.
Our proof in \cref{app:proof:theorem:power_laws} follows arguments similar to those made by \citet{wegrzycki2017cascade}. 

\begin{theorem} \label{theorem:power_laws}
    Let $\cG$ be a hierarchical one-directional input-output random network as described above. Then, for large enough $K$, the total number of failures has the following asymptotic distribution: 
    \begin{align} \label{eq:power_law}
    \Pr[F = f] \asymp \frac {x^n} {K(1 - (1 - x^n)(1 - p)^f)} \ge \underbrace{\frac {x^n} {K \left (1 + (1 - x^n) \log \left ( \frac 1 {1 - p} \right ) \right )}}_{C(K, p, x, n) > 0} \frac{1}{f}, 
    \end{align}where $C(K, p, x, n) > 0$ is a constant that depends on the model parameters.
\end{theorem}

The above result implies that $F$ has a power-law tail lower bound, i.e., $\Pr [F \ge f] = \sum_{f' = 1}^K \Pr[F=f'] \ge \Pr[F = f] \ge {C}/ {f}$. Having proven \cref{theorem:power_laws}, the next question is: \emph{How can we calculate the probability that a fractional cascade emerges?} Conceptually, we want to make the probability that there are at least $\varepsilon K$ failures go to 0 as $K$ grows for any fixed $\varepsilon \in (0, 1)$. Given the result of \eqref{eq:power_law}, this is only possible if we set $x = 0$. Conceptually, we could characterize the hierarchical one-directional input-output random production network as a \emph{``architecturally fragile''} architecture, because even the tiniest shock can be devastating to the entire network. 

% Identifying fragile networks is important because it allows companies and organizations to identify potential vulnerabilities and risks in their operations. 

% This can then be used to design more robust networks through targeted interventions. Thus, once fragile networks are identified, companies can take a number of steps to make them more robust, such as diversifying their supplier base, increasing inventory levels, and implementing contingency plans. 

\subsubsection{Proof of \texorpdfstring{\cref{theorem:power_laws}}{theorem:powerlaws}: Power law cascade size and fragility}\label{app:proof:theorem:power_laws}

    For notational convenience, we define $\mathsf{rdag}(K, p)$ as the hierarchical one-directional input-output production network (also referred to as ``random directed acyclic graph'' or ``random DAG'') with $K$ products and edge probability $p$.
    
    Let $\cG \sim \mathsf{rdag}(K, p)$, have nodes $1, 2, \dots, K$ (in this order). Let $P_{k, f}$ be the probability of having $f$ distinct failures in the hierarchical one-directional input-output random production network with $k$ nodes, conditioned on a failure on node 1. We have $P_{1, 1} = 1$ and $P_{k, f} = 0$ for $f > i$ and $f < 1$. To devise a recurrence formula for $P_{i, f}$, note that for the $i$-th node we have the following: 

    \begin{compactenum}
        \item $i$ is affected by the cascade. That happens if at least one connection to $f - 1$ infected nodes up to node $i - 1$, or if $i$ fails due to percolation. This happens with probability $\big \{ [ 1 - (1 - p)^{f - 1} ] + x^n - [ 1 - (1 - p)^{f - 1} ] x^n \big \} P_{k - 1, f - 1} = \big [ 1 - (1 - p)^{f - 1}(1 - x^n) \big ]P_{k - 1, f - 1}$.
   
        \item $i$ is not affected by the cascade. That means that $i$ has $\ge 1$ functional supplier, and no connections to (sourcing from) the $f$ infected nodes. That happens with probability $(1 - p)^f (1 - x^n) P_{k - 1, f}$.
        
    \end{compactenum}

This produces the following recurrence,
    \begin{align} \label{eq:recurrence}
        P_{k, f} = \left [ 1 - (1 - p)^{f - 1}(1 - x^n) \right ]P_{k - 1, f - 1} +   (1 - p)^f (1 - x^n) P_{k - 1, f}.
    \end{align}

    To determine the distribution of the failures $F$, we assume that the cascade can start at any node with equal probability $1 / K$ and that the probability of failure for any given node is $x^n$. Also, since a cascade in $\mathsf{rdag}(K, p)$ starting from node 1 is the same as starting from node $i$ in $\mathsf{rdag}(K + i - 1, p)$, the distribution obeys the following:
    \begin{align}
        \Pr [F = f] = \frac {x^n} {K} \sum_{k \in [K]} P_{k, f}
    \end{align}

    We let $Q_{K, f} =  \sum_{k \in [K]} P_{k, f}$, so that $\Pr [F = f] = \frac {x^n} {K} Q_{K, f}$. Summing \cref{eq:recurrence} for $k \in [K]$ and using the definition of $Q_{K, f}$ yields a recurrence relation for $Q_{K, f}$, that is, $Q_{K, f} = (1 - p)^f (1 - x^n) Q_{K - 1, f} + [ 1 - (1 - p)^{f - 1} ] (1 - x^n) Q_{K - 1, f - 1}$. We take the limit for $K$ large, let $q_f = \lim_{K \to \infty} Q_{K, f}$, and solve the recurrence $q_{f} = (1 - p)^f (1 - x^n) q_{f} + [ 1 - (1 - p)^{f - 1} ] (1 - x^n) q_{f - 1}$ to get $q_f = \frac {1} {1 - (1 - x^n) (1 - p)^f}$. Since $e^x \ge x$, we have $(1 - p)^f \le \log (1 - p) f$ and subsequently $1 - (1 - x^n) (1 - p)^f \le f \left (1 + (1 - x^n) \log \left ( \frac 1 {1 - p} \right ) \right )$. Therefore, for sufficiently large $K$, 
    \begin{align*}
        \Pr [F = f] & \asymp \frac {x^n q_f} {K} \asymp \frac {x^n} {K(1 - (1 - x^n) (1 - p)^f)} \ge \underbrace {\frac {x^n} {K \left (1 + (1 - x^n) \log \left ( \frac 1 {1 - p} \right ) \right )}}_{C(K, p, x, n) > 0} \frac 1 f.
    \end{align*}

\newpage

\section{Main Result: Resilience bounds based on global graph features} \label{app:main_result}

We derive resilience bounds for arbitrary networks using global graph features. Specifically, the global graph features governing the bounds are the source requirement ($m$), the sourcing influence ($\mu$), the number of raw products ($r$), and the number of finished goods ($c$). We restate and then prove the main result
    
\begin{theorem}[Bounding the resilience with global graph features] \label{theorem:resilience_graph_statistics}
    For any network $\cG$, the resilience satisfies
    \begin{align*}
         \left ( \frac {\varepsilon} {2(m + r) (\mu + c)} + \sqrt {\frac {\log K} {rK}} \right )^{1/n}  \le R_{\cG}(\varepsilon) \le \left [ \frac {(1 - \varepsilon) K} {\sqrt 2 r^{3/2} + \sqrt {r \log K}} \right ]^{1/n}.
    \end{align*} Thus, if $r = \omega (K^{2/3})$, then the network is fragile. Conversely, if $r$, $c$, $m$ and $\mu$ are constant, the network is operationally resilient for $\delta = \left ( \frac {\varepsilon} {2(m + r) (\mu + c)} \right )^{1/n}   $. 
    
\end{theorem}

\subsection{Proof of \texorpdfstring{\cref{theorem:resilience_graph_statistics}}{theoremresilienceupperbound}: Resilience bounds based on global graph features} \label{app:theorem:resilience_graph_statistics}

\xpar{Lower bound} We construct $\cG'$ as follows: We start by $\cG$, and additionally, for each final good $j \in \cC$ in $\cG$ and each raw product $i$ we add an edge from $i$ to $j$ in $\cG'$ if there is a path from $i$ to $j$ in $\cG$. By construction $\cG \subseteq \cG'$ and thus $R_{\cG}(\varepsilon) \ge R_{\cG'} (\varepsilon)$. In $\cG'$ we have that the sourcing requirement $m'$ satisfies $m' \le m + r$ and the sourcing influence satisfies $\mu' \le \mu + c$. The result follows by applying \cref{theorem:parallel_products} to $\cG'$ and using the above inequalities. The value of $\delta$ comes from the fact that the term $\sqrt {\frac {\log K} {rK}}$ is always non-negative.

\xpar{Upper bound} Let $F_{\cR}$ correspond to the number of failures of the raw products. Let $x_{\cR} = \sup \{ x \in (0, 1) : \Pr [F_{\cR} \le (1 - \varepsilon) K] \ge 1 - 1 / K \}$. $x_{\cR}$ is an upper bound to $R_{\cG}(\varepsilon)$ since the event $\{ F \le (1 - \varepsilon) K \}$ implies $\{ F_{\cR} \le (1 - \varepsilon) K \}$. Moreover, by the Chernoff bound, we have that for any $\varepsilon' > 0$:

\begin{align*}
    \Pr \left [ \frac {F_{\cR}} {r} \le (1 + \varepsilon') x^n \right ] \ge 1 - e^{-2 (\varepsilon')^2r }.
\end{align*}

Letting $(1 + \varepsilon') x^n = (1 - \varepsilon) K / r$, $e^{-2 r (\varepsilon')^2} = 1/K$ and solving for $x$ would produce an upper bound to $x_{\cR}$ and subsequently an upper bound to $R_{\cG}(\varepsilon)$. Solving the system yields the following result. 

\begin{align*}
    \varepsilon' = \sqrt {\frac {\log(K)} {2 r}} \qquad \text{and} \qquad x = \left [ \frac {(1 - \varepsilon) K} {\sqrt {2} r^{3/2} + \sqrt {r \log K}} \right ]^{1/n}.
\end{align*}

To determine conditions where the resilience goes to zero as $K \to \infty$ we focus on the denominator of the upper bound: First, the term $\sqrt {r \log K}$ is at most $\sqrt {K \log K} < K$ and cannot grow faster than $K$. Second, the term $\sqrt 2 r^{3/2}$ grows faster than $K$ as long as $r = \omega \left ( K^{2/3} \right )$.

\newpage

\section{Overview of results for several families of networks} \label{app:architectures}

\begin{figure}[!h]
    \centering
    \includegraphics[width=0.75\linewidth]{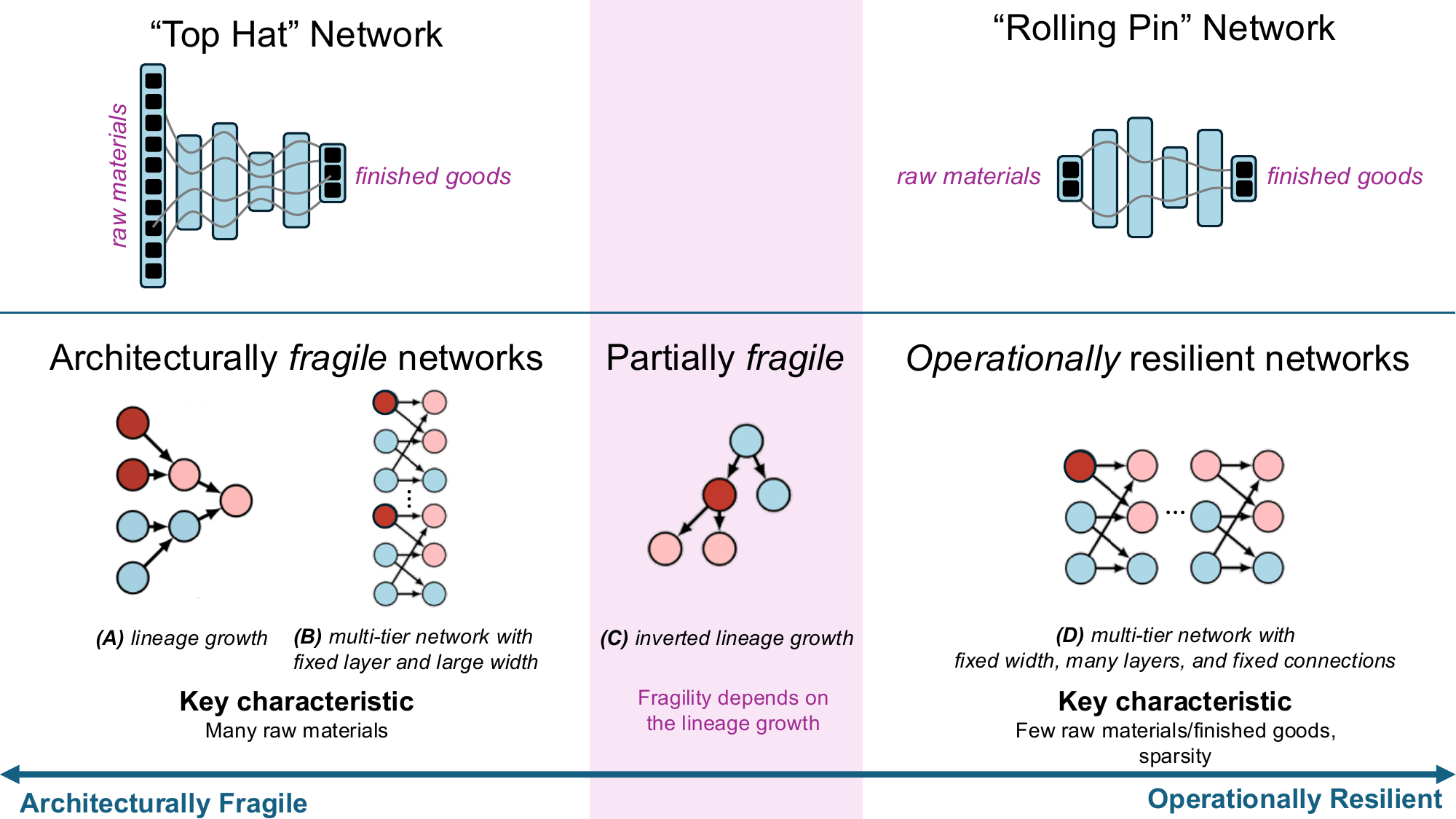}
    \caption{\textbf{Illustration of fragile and resilient network structures.} On the one hand, the key characteristic of the \textbf{fragile} networks is having many raw materials ($r$), which makes the network look like a ``top hat''. Two forms of fragile networks are shown -- one corresponding to lineage growth and one corresponding to a multi-tier network with fixed layers and a large-width network. Among the two, the lineage growth exhibits higher fragility due to higher centralization. In the cases of more centralized networks, failures cause larger cascades in general. \\
    On the other hand, the key characteristic of \textbf{resilient} networks is that they have few raw materials ($r$) and finished goods ($c$) as well as sparsity, i.e., sourcing requirements ($m$) and sourcing influence ($\mu$) remain fixed, resembling a ``rolling pin'' structure. An example of a fragile network is a multi-tier network with a fixed width and multiple layers, characterized by fixed connections. Finally, networks can also be partially resilient, as in the case of an inverted lineage growth network.} 
\label{fig:qualitative_representation_ec}
\end{figure}

\noindent In addition to our main results, we study the resilience of several stylized production networks that help to demonstrate the interplay between these parameters. These structures and their parameterization are as follows.

\begin{compactenum}
    \item \emph{Hierarchical one-directional input-output networks (Global Macro-Economies and Cross-Sectoral Dependencies):} $K$ products ordered as $1,2,\ldots, K$ and connected with directed edge probability $p$, conditioned on maintaining topological order; see \cref{app:power_laws}. 
    
    These networks represent the ``broad-scale'' view of national or global economies where products are ordered by their degree of processing (e.g., from raw ores to semiconductors to consumer electronics). The random directed edges represent the vast array of potential input requirements across sectors.
    
    \xpar{Insight} Our finding (\cref{theorem:power_laws}) that these networks are inherently fragile reflects the systemic risk in globalized trade. As a country's industrial base becomes more complex, the likelihood that any single final product remains ``reachable'' from a functional primary source diminishes, justifying the need for high-level decoupling strategies at the national policy level.
    
    \item \emph{Two-tier network with dependencies (Commodity-Based Manufacturing and ``Flat'' Supply Chains):} A set of $r$ raw materials that are used to produce $K - r$ finished goods, each complex product requires $m$ raw inputs (source dependencies), and each raw material is sourced by $\mu$ finished goods (supply dependencies). 
    
    This model describes industries where finished goods are directly assembled from a set of raw materials or standard commodities. Examples include the basic food processing industry or simple metal fabrication, where the distance between the extraction site and the final assembler is minimal.

    \xpar{Insight} We prove (\cref{theorem:parallel_products}) that these networks are resilient. Managerially, this justifies the use of ``flat'' sourcing structures. When the depth of the supply chain is limited to two tiers, the risk of cascading failures is naturally boxed in, making this the safest, albeit often the least specialized, architecture.

    \item \emph{Hierarchical production networks:} 

    \begin{compactitem}
    \item \emph{(i) Lineage growth (Specialized Component Assembly):}  A tree production network with depth $D$ and fanout $m \ge 1$, that is, each product has $m$ inputs, independent of other products. Represents complex assembly industries like aerospace or automotive. A single final product (the root) depends on a vast ``fanout'' of specialized sub-components, which in turn depend on further tiers of raw materials. Failure propagates ``backwards'' from small upstream suppliers to large downstream OEMs (original equipment manufacturers). 
    
    \item \emph{(ii) Inverted lineage growth (Resource Distribution):} A tree production network generated by a branching process (also known as the Galton-Watson process) with branching distribution $\cD$ with mean $\mu$, which has (random) extinction time $\tau$ and probability of extinction $\eta^*(\cD) = \Pr [\tau < \infty]$. In this regime, the percolation starts from the root node (raw material) and proceeds to the leaves (see \cref{app:forward_forward_tree} and \cref{subfig:forward_forward}). 
    
    This represents the ``distribution'' or ``utility'' view, such as a single refinery or mine (the root) supplying a wide range of downstream chemical products or industries. Here, the risk is a ``forward'' cascade where a single upstream shock halts an entire industrial ecosystem.
    \end{compactitem}

    \xpar{Insight} These two archetypes highlight the ``Super-Spreader'' risk. In Lineage growth, the OEM is vulnerable to the sheer number of upstream nodes; in Inverted growth, the entire network is at the mercy of a single root bottleneck.

    We study hierarchical networks in \cref{app:hierarchical}. 

    \item \emph{Multi-tier network (Multi-Stage Process Industries such as Chemicals, Textiles, or Pharmaceuticals):} We study the multi-tier network with width $r$, depth $D$, and probability of edge creation $p$. The multi-tier network corresponds to a network with $D$ tiers, each of which has $r$ products, and there are only connections between tier $d$ and $d + 1$ for all $d \in \{ 1, 2, \dots, D - 1\}$. These connections correspond to random edges that are sampled independently with probability $p$. 
    
    The multi-tier network structure represents industries that move through distinct, sequential ``stages'' or ``tiers'' of processing. In chemical manufacturing, Tier 1 might be raw petrochemicals, Tier 2 intermediate polymers, and Tier 3 specialized resins. 
    
    \xpar{Insight} The multi-tier network model allows managers to study the ``width'' vs. ``depth'' trade-off. Our analysis of multi-tier networks helps quantify how much redundancy (width $r$) is needed within each processing tier to prevent a failure in one stage from completely disconnecting the next, providing a blueprint for designing robust multi-stage industrial processes.
    
    We study the multi-tier network in \Cref{app:trellis}.
\end{compactenum} 
 
\cref{tab:resiliences} and \cref{fig:qualitative_representation_ec} summarize our results for the resilience of the above architectures and the corresponding theorems where the results are proved. 

% \begin{theorem}[Resilient and Fragile Network Taxonomy] \label{theorem:taxonomy}
%     The hierarchical one-directional input-output random production network, lineage growth, and multi-tier network with $m \ge 1$ and constant width are always fragile with $R_{\cG} (\varepsilon) \to 0$ as $K \to \infty$ for all $\varepsilon$. In contrast, two-tier networks with dependencies are resilient. The inverted lineage growth network satisfies $\Pr \left [ R_{\cG}(\varepsilon) = 0 \right ] > 0$.
% \end{theorem}

% \cref{theorem:taxonomy} is gradually proven in the following sections. 

\begin{table}[t]
    \centering
    \tiny
    \begin{tabular}{p{3cm}lll}
    \toprule
        \textbf{Topology} &\textbf{Is the network Fragile?} & \textbf{Main Result} & \textbf{Resilience} \\
    \midrule
        Hierarchical one-directional input-output networks & Yes & \cref{theorem:power_laws}  & $\Oeps { \left ( \frac 1 K \right )^{1/n}}$ \\ 
        Two-tier network & No & \cref{theorem:parallel_products}  & $\Omegaeps {\left ( \frac {1} {\mu m} \right )^{1/n}}$ \\

        Lineage Growth & Yes & \cref{theorem:tree_resilience,eq:quick} & $\Oeps { \left ( \frac {m \log K} {K}\right )^{1/n} }$ for $m \ge 2$ \\
        &   & & $\Oeps {\left ( \frac {1} {K} \right )^{1/n}}$ for $m = 1$\\
        Inverted Lineage Growth & No, with probability $1 - \eta^*$ & \cref{theorem:gw_resilience} & See \cref{theorem:gw_resilience} \\
        Multi-tier network w/ fixed width, many layers, and fixed connections & No & \cref{proposition:trellis} & $\Omegaeps {\left ( \frac {1} {r(1 + p)} \right )^{2/n}}$ \\
        Multi-tier network w/ fixed layer and large width & Yes & \cref{proposition:trellis} & $\Oeps {\left ( \frac {1} {\sqrt K} \right )^{1/n}}$ \\ \midrule
        Any Network & Yes if $r = \omega(K^{2/3})$ & \cref{theorem:resilience_graph_statistics}  & $\Oeps {\left ( \frac {K} {r^{3/2}} \right )^{1/n}}$ \\
         & No, for $m$, $\mu$, $r$, and $c$ fixed & & $\Omegaeps {\left ( \frac {1} {(m + r) (\mu + c)}\right )^{1/n}}$ \\
     \bottomrule 
    \end{tabular}
    \caption{Resilience bounds for several architectures. The notation $\Omegaeps {\cdot}, \Oeps {\cdot}$ suppresses the dependence on $\varepsilon$ (as we assume that $\varepsilon$ is fixed in our definition of resilience).}
    \label{tab:resiliences}
\end{table}

\subsection{Hierarchical Production Networks} \label{app:hierarchical}

A network can be organized hierarchically, with different levels representing different stages of the production process. The raw materials or components that go into the production of a product are at the bottom of the hierarchy, and the finished product is at the top. Each level of the hierarchy represents a stage of production in which materials or components are transformed into a more advanced or finished product \citep{elliott2022supply}. This hierarchical structure helps visualize the flow of materials and information throughout the network and identify potential bottlenecks or inefficiencies. Another possible hierarchy involves the production of a growing number of finished goods from raw materials. In this section, we study these two hierarchies, which we call \emph{lineage growth} and \emph{inverted lineage growth} (referring to the directions of percolation with respect to network growth), production networks visualized in \cref{fig:tree_percolation}. More specifically, we consider:
\begin{compactitem}
    \item The \emph{lineage growth network} (\cref{subfig:forward_backward}) where the tree grows from the root, but percolation begins at the leaves and proceeds backwards to the root. For the scope of this paper, we study lineage growth in deterministic $m$-ary trees. In \cref{sec:forward_backward_tree} we prove that, as expected, such networks are, in fact, fragile and give lower bounds on resilience. 
    \item The \emph{inverted lineage growth network} (\cref{subfig:forward_forward}) at which the tree grows from the root, and then the percolation starts from the root and proceeds to the leaves. Here, the production network is generated by a stochastic \emph{branching process}; the Galton-Watson (GW) process. In \cref{app:forward_forward_tree}, we prove that, under specific conditions, such networks are fragile with nonzero probability, but otherwise resilient.  
\end{compactitem}

\begin{figure}[t]
    \centering
    \subfigure[Lineage growth\label{subfig:forward_backward}]{
    \begin{tikzpicture}[transform shape, scale=0.85]
        \Vertex[x=0, y=0, color=pink]{u1}
        \Vertex[x=1, y=0.5, color=pink]{u3}
        \Vertex[x=1, y=-0.5]{u2}
        \Vertex[x=2, y=-1.5]{u4}
        \Vertex[x=2, y=-0.5]{u5}
        \Vertex[x=2, y=0.5, color=rose]{u6}
        \Vertex[x=2, y=1.5, color=rose]{u7}

        \Edge[color=black, Direct](u2)(u1)
        \Edge[color=black, Direct](u3)(u1)
        \Edge[color=black, Direct](u4)(u2)
        \Edge[color=black, Direct](u5)(u2)
        \Edge[color=black, Direct](u6)(u3)
        \Edge[color=black, Direct](u7)(u3)

        \Vertex[x=-2, y=-2, Pseudo]{x1}
        \Vertex[x=3.5, y=-2, Pseudo]{x2}
        \Vertex[x=-2, y=2, Pseudo]{y1}
        \Vertex[x=3.5, y=2, Pseudo]{y2}

        \Edge[color=blue, Direct, label={Tree Growth}](x1)(x2)
        \Edge[color=red, Direct, label={Percolation}](y2)(y1)

    \end{tikzpicture}} 
    \subfigure[Inverted lineage growth\label{subfig:forward_forward}]{
    \begin{tikzpicture}[transform shape, scale=0.85]

        \Vertex[x=0, y=0]{u1}
        \Vertex[x=1, y=0.5]{u3}
        \Vertex[x=1, y=-0.5, color=rose]{u2}
        \Vertex[x=2, y=-1.5, color=pink]{u4}
        \Vertex[x=2, y=-0.5, color=pink]{u5}
        
        \Edge[color=black, Direct](u1)(u2)
        \Edge[color=black, Direct](u1)(u3)
        \Edge[color=black, Direct](u2)(u4)
        \Edge[color=black, Direct](u2)(u5)
        
        \Vertex[x=-2, y=-2, Pseudo]{x1}
        \Vertex[x=3.5, y=-2, Pseudo]{x2}
        \Vertex[x=-2, y=2, Pseudo]{y1}
        \Vertex[x=3.5, y=2, Pseudo]{y2}

        \Edge[color=blue, Direct, label={Tree Growth}](x1)(x2)
        \Edge[color=red, Direct, label={Percolation}](y1)(y2)
    \end{tikzpicture}}
    \subfigure[\label{fig:gw_subcritical}]{\includegraphics[width=0.5\textwidth]{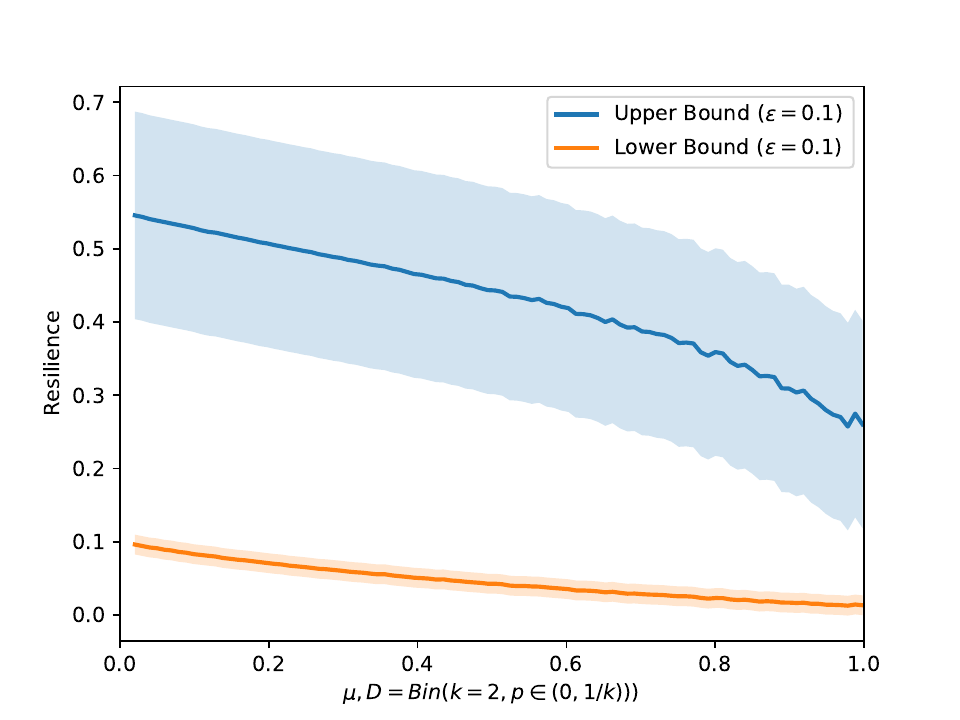}}
    %\shrink
    
    \caption{Lineage growth (a) and Inverted lineage growth (b) networks. Node failures are drawn in pink. (c): Resilience bounds for a subcritical GW process with branching distribution as a function of $\mu$; note the decreasing trends in both upper and lower bounds, $\ev {\cG} {\overline R_{\cG}(\varepsilon)}$ and $\ev {\cG} {\underline R_{\cG}(\varepsilon)}$, with increasing $\mu$. In panels (a) and (b), the spontaneous failures (sources of failure) are colored in dark red, the resulting failures due to cascading effects are colored pink, and the unaffected nodes are colored light blue.} 
    \label{fig:tree_percolation}
    \shrink
\end{figure}

\subsubsection{Lineage growth network (\texorpdfstring{$m$}{m}-ary tree)} \label{sec:forward_backward_tree}

In the case of the lineage growth, we consider an $m$-ary tree with height $D$ and fanout $m \ge 1$. The levels of the tree correspond to the ``tiers'' with raw materials placed in the tier $D$ and more complicated products placed in higher tiers. Each product has $n$ potential suppliers, and each product in the tier $d \in [D - 1]$ has exactly $|\cN(i)| = m$ inputs from the tier $d + 1$. Tier $d = D$, which corresponds to the raw products, has no inputs. \cref{subfig:forward_backward} shows how the percolation process evolves in a tree with $m = 2$ and $D = 3$, where failures (drawn in pink) propagate from the two faulty raw materials to the root.  

In addition to the power law result (\cref{theorem:power_laws}), the case of the $m$-ary tree is another example that motivates the resiliency measure $R_{\cG}(\varepsilon)$. Specifically, let us think about the probability of a catastrophic failure in a tree, that is, one that affects a substantial proportion of the suppliers in the production network. A raw material failing to be produced can cause its parent product not to be produced and inductively create a cascade up to the root. The complete cascade will start from the failed product in tier $D - 1$ since some products in tier $D - 1$ may be made if their corresponding raw materials are produced. However, no product can be produced from tier $D - 2$ onward. As a result, only $o(K)$ products survive. The probability of such an event equals: 
\begin{align} \label{eq:quick}
    \Pr [S = o(K)]  \ge \Pr [\text{$\ge 1$ raw material malfunctions}] = 1 - (1 - x^n)^{m^{D - 1}} \ge 1 - e^{-x^n m^{D - 1}}. 
\end{align}
It is easy to see that if $x = \Omega \left ( \left ( {m \log K}/{K} \right )^{1/n} \right )$, then a catastrophe occurs with a high probability in the tree structure, meaning that failure probabilities as small as $\left (  {m \log K}/ {K} \right )^{1/n} + o(1)$ can cause catastrophes with probability approaching one (and therefore the Lineage growth network is a fragile architecture). Therefore, it is interesting to study cases where such a scenario does not happen; on the contrary, we have many products that survive. The following theorem formalizes the lower bounds and provides an additional upper bound for the resilience of the lineage growth (proof deferred to \cref{app:proof:theorem:tree_resilience}). 

\begin{theorem} \label{theorem:tree_resilience}
    Let $\cG$ be a lineage growth with fanout $m$ and depth $D$. Then, 
    \begin{align}
        \left [ 1 - \left ( 1 - \frac 1 {K} \right )^{\frac {1} {(1 - \varepsilon) K}} \right ]^{1/n} \le R_{\cG}(\varepsilon) \le \begin{cases}
            \left ( \frac {2} {K (1 - \varepsilon)}\right )^{1/n} = \left ( \frac {2} {D (1 - \varepsilon)}\right )^{1/n}, & m = 1 \\
            \left ( \frac {(1 - \varepsilon) \log m} {\log K} \right )^{1/n} \asymp \left ( \frac {(1 - \varepsilon)} {D} \right )^{1/n}, & m \ge 2
        \end{cases}.
    \end{align}
    Therefore, the network is fragile. 
\end{theorem}

% Since $\varepsilon \in (0, 1)$ the above quantity behaves asymptotically as $O \left ( \left ( \frac {\log m} {\log K} \right )^{1/n} \right )$ for all values of $\varepsilon$. Therefore, the resilience goes to 0 with rate $\log K$. However, note that in \cref{eq:quick}, we showed a better rate of $O \left ( \left ( \frac {m\log K} {K} \right )^{1/n} \right )$, and therefore we state that the resilience goes to 0 with a rate of $O \left ( \left ( \frac {m\log K} {K} \right )^{1/n} \right )$ (for $m \ge 2$). 

\subsubsection{Inverted lineage growth network (Branching process)} \label{app:forward_forward_tree}

We consider a random hierarchical network in which the products at each level $D$ are denoted by $\cK_d$. Starting from one raw material, we branch out through a Galton-Watson (GW) process \citep[Ch.~5, \S5.1]{galtonwatson_notes} such that every product $i \in \cK_d$ at the level $d \ge 1$ creates $\xi_{i}^{(d)}$ supply dependencies, where $\{ \xi_{i}^{(d)} \}_{i \in \cK_d, d \ge 0}$ are generated i.i.d. from a distribution $\cD$, with mean $\ev {\cD} {\xi_{i}^{(d)}} = \mu > 0$. Subsequently, the number of products at each level obeys
\begin{align} \label{eq:gw_1}
    |\cK_{d + 1}| = \begin{cases} 
        \sum_{i \in \cK_d} \xi_i^{(d)}, & d \ge 2 \\
        1 & d = 1
    \end{cases}. 
\end{align} 

Adding the node percolation process, we start a percolation of the children of the root node $r$ and subsequently proceed to their children, etc. The number of the surviving products $S$ in this case can be expressed as $S = \sum_{d : |\cK_d| \ge 1} \sigma_d$, where $\{ \sigma_d \}_{d \ge 0}$ follow another branching process, namely
\begin{align} \label{eq:gw_2}
    \sigma_{d + 1} = \begin{cases} 
        \sum_{1 \le i \le \sigma_d} \xi_i^{(d)} \left ( 1 - \prod_{s \in \cS(i)} X_{is} \right ), & d \ge 2 \\
        1 - \prod_{s \in \cS(r)} X_{rs} & d = 1
    \end{cases}. 
\end{align} 

In \cref{subfig:forward_forward}, we show such an example in which failures propagate from the raw product to products of increasing complexity. In the case of the GW process, the network has a random number of nodes $K$. For this reason, to characterize resilience and fragility, instead, we focus on quantities $\Pr_K [R_{\cG}(\varepsilon) = 0]$. We generalize the definition of resilience as follows. 

\begin{mdefinition}
    A network $\cG$ is $(q, \delta)$-resilient for some $q, \delta \in (0, 1)$ if and only if $\Pr_K [R_{\cG}(\varepsilon) \ge \delta] \ge q$.
\end{mdefinition}

We prove the following result about the inverted lineage growth network (proof deferred to \cref{app:proof:theorem:gw_resilience}): 

\begin{theorem} \label{theorem:gw_resilience}
    Let $\cG$ be generated by a GW process in which the number of children of each node is generated by a distribution $\cD$ with mean $\mu > 0$ and extinction time $\tau$. Let $G_{\cD}(\eta) = \ev {\xi \sim \cD} {e^{s\xi}}$ be the moment generating function of $\cD$, and let $\Pr [\tau < \infty]  = \eta^* = \inf \{ \eta \in [0, 1]: G_{\cD}(\eta) = \eta \}$ be the extinction probability of the GW process. Then the following are true: \emph{(i)} If $\mu < 1$, then $\cG$ is $(1, \underline \delta(\mu))$-resilient, \emph{(ii)} If $\mu (1 - x^n) > 1$, then $\cG$ is $(\eta^*, \underline \delta(\mu))$-resilient. 
    
    Moreover, the expected upper bound on the resilience is, for $\mu \in (0, 1) \cup (e^2, \infty)$, given by  $\overline \delta (\mu) = \ev {\cG} {\overline R_{\cG} (\varepsilon)} = \sum_{1 \le k < \infty} \Pr [\tau = k] \overline x(\mu, \tau, \varepsilon)$ with 
    {\footnotesize
    \begin{align}
     \overline x (\mu, \tau, \varepsilon) & = \inf \left \{ x \in \left [0, \one \{ \mu < 1 \} + \left ( 1 - \frac {1} {\mu} \right )^{1/n} \one \{ \mu > 1 \} \right ] : (1 - x^n) \frac {\mu^\tau (1 - x^n)^\tau - 1} {\mu (1 - x^n) - 1}  \le \frac {1 - \varepsilon} {2} \frac {\mu^\tau - 1} {\mu - 1} \right \}. \label{eq:ub_gw}  
    \end{align}}

    The expected lower bound on the resilience is, for $\mu \in (0, 1) \cup (e, \infty)$, given by $\underline \delta(\mu) = \ev {\cG} {\underline R_{\cG} (\varepsilon)} = \sum_{1 \le k < \infty} \Pr [\tau = k] \underline x(\mu, \tau, \varepsilon)$ with
    {\footnotesize
    \begin{align}
        \underline x(\mu, \tau, \varepsilon) & = \sup \left \{ x \in \left [0, \one \{ \mu < 1 \} +  \left ( 1 - \frac {1} {\mu} \right )^{1/n} \one \{ \mu > 1 \} \right ] : \frac {\mu^\tau - 1} {\mu - 1} - (1 - x^n) \frac {\mu^\tau (1 - x^n)^\tau - 1} {\mu (1 - x^n) - 1} \le \varepsilon \right \}. \label{eq:lb_gw}
    \end{align}}

\end{theorem}

% \begin{proofsketch}
%     If $Z_r = 0$, which happens with probability $x^n$, then the number of surviving products is $S = 0$. If $Z_r = 1$, which happens with probability $1 - x^n$, the cascade behaves as a GW process with mean $\mu_x = \mu (1 - x^n)$. Now, conditioned on the fact that $Z_r = 1$ we bound the percentage of the expected number of surviving products over the total expected number of products to devise the upper and lower bound. For the upper bound we require $\ev {\cG, x} {S} \le \frac {1 - \varepsilon} {2} \ev {\cG} {K}$, and for the lower bound we require $\ev {\cG, x} {F} = \ev {\cG} {K} - \ev {\cG, x} {S} \le \varepsilon$. We prove that if the process takes infinite time to terminate, then the only possible solution is $x = 0$ which makes the network fragile. In all of the other cases, we study the existence of an upper and a lower bound to the resilience by solving two inequalities, whose roots are studied in (the auxiliary) \cref{lemma:gw_roots}. 
% \end{proofsketch}

Applying \cref{theorem:gw_resilience} for the case where $\cD$ is a point-mass function that equals $\mu$ with probability $1$, yields the following corollary for deterministic structures.

\begin{corollary}
    Let $\cD$ have $\Pr_{\xi \sim \cD} [\xi = \mu] = 1$ for $\mu > 1$. Then $\cG$ is fragile.
\end{corollary}

For the subcritical regime, we plot the expected resilience bounds in \cref{fig:gw_subcritical} for a subcritical GW process with branching distribution $\cD = \mathsf{Bin}(k, p \in (0, 1 / k))$, as a function of $\mu = k p$.

We want to remark here that the work of \citet{elliott2022supply} considers a hierarchical supply network similar to the one presented in \cref{sec:forward_backward_tree} -- though by assuming a different percolation process. In Section II of their paper, they observe that their reliability metric decreases as the interdependency increases, and the reliability increases as the number of suppliers increases. This is in agreement with the lower bound presented in \cref{theorem:tree_resilience} for the lineage growth, which increases as $n$ increases and decreases as the sourcing requirement $m$ increases, the probability that there is at least a raw material failure (\cref{eq:quick}) increases. Moreover, in their paper, if the shocks are below a value, then for large depths, the reliability goes to zero, which is conceptually in agreement with the upper bounds presented in \cref{theorem:tree_resilience}, which go to zero as $D$ grows. Finally, for the inverted lineage growth network -- which is not studied by \citet{elliott2022supply} -- we again get a result that is in agreement with their results, since as the average sourcing influence $\mu$ increases, we observe that the upper and lower bounds in the resilience decrease, as empirically shown in \cref{fig:gw_subcritical}. 

\subsection{Multi-tier network} \label{app:trellis} 

The multi-tier network we study is defined with three parameters: the width $r$, the depth $D$, and the edge sampling probability $p$. Specifically, there are $D$ tiers and each tier has $r$ products, and, thus, $K = rD$. Edges are sampled randomly only between the nodes of tier $d$ and $d + 1$ for $d \in \{ 1, \dots, D - 1\}$ independently with probability $p$. Directly applying \cref{theorem:resilience_graph_statistics} we get the following result regarding the multi-tier network:

\begin{theorem} \label{proposition:trellis}
    If the multi-tier network has $D = o(K^{1/3})$ tiers, then it is fragile with resilience that goes to zero at the rate $\Oeps {\left ( \frac {1} {\sqrt K} \right )^{1/n}}$ as $K\to\infty$. On the other hand, if the number of raw products and the edge probability for the multi-tier network are fixed ($r$ and $p$ are both $O(1)$), then the multi-tier network is resilient and its resilience can be lower bounded by $\Omegaeps {\left ( \frac {1} {r(1 + p)} \right )^{2/n}}$ as $K\to\infty$.
\end{theorem}

\subsection{Proofs of results} \label{app:proofs}

\subsubsection{Helper Lemma}

Before proceeding to the main result, we use the following helper lemma:

\begin{lemma} \label{lemma:upper_bound_resilience}
    Let $\varepsilon \in (0, 1)$ and let $\overline R_{\cG}(\varepsilon) \in (0, 1)$ be a percolation probability such that $\Pr_{x = \overline R_{\cG}(\varepsilon)} [S \ge (1 - \varepsilon) K] \le \frac 1 2$. Then, for $K \ge 3$, we have $ R_{\cG}(\varepsilon)<\overline R_{\cG}(\varepsilon)$. 
\end{lemma}

\xpar{Proof} If $S(x)$ is the number of products that survive at a given probability of percolation $x$, and $x_1 \le x_2$ are two percolation probabilities, then a straightforward coupling argument shows that $S(x_1) \ge S(x_2)$, and subsequently, for every $s \in [0, K]$ we have $\Pr_{x = x_1} [S \ge s] \ge \Pr_{x = x_2} [S \ge s]$. Now, in order to arrive at a contradiction, let $\overline R_{\cG}(\varepsilon) \le R_\cG (\varepsilon)$, and $s = (1 - \varepsilon)K$. Then $1 - 1 / K \le \Pr_{x = R_{\cG}(\varepsilon)} [S \ge (1 - \varepsilon) K] \le \Pr_{x = \overline R_{\cG}(\varepsilon)} [S \ge (1 - \varepsilon) K] \le 1 / 2$, which produces a contradiction.

\subsubsection{Proof of \texorpdfstring{\cref{theorem:tree_resilience}}{theorem:treeresilience}: Lineage growth network}\label{app:proof:theorem:tree_resilience}

\xpar{Helper lemmas: Upper and Lower Bounds on $\ev {} {S}$} To prove the result, we first prove upper and lower bounds on $\ev {} {S}$ given in the lemma that follows:

\begin{lemma} \label{lemma:probability_functional}
    Let $q_d$ be the probability that a product in tier $d$ can be produced. Then 
    \begin{equation}
         q_{d} = \begin{cases} \left ( 1 - x^n \right )^{\frac{m^{D - d + 1} - 1} {m - 1}}, & m \ge 2 \\
         \left ( 1 - x^n \right )^{D - d + 1}, & m = 1
         \end{cases}
    \end{equation}
\end{lemma}

    \xpar{Proof} Let $q_d = \Pr [\text{a product in tier $d$ can be produced}] = \Pr [\exists \text{a functional supplier at tier $d$}]$. To calculate $q_d$, note that all the inputs for a product node at tier $d$ succeed with probability $q_{d + 1}^m$, and then the probability that at least one supplier is functionally conditioned on all the inputs working is $1 - x^n$. This yields the following recurrence relation $q_d = q_{d + 1}^m (1 - x^n)$ with $q_{D + 1} = 1$. Solving this recurrence relation, we get $q_{D - d} = \left ( 1 - x^n \right )^{\sum_{l = 0}^d m^l}$ for $d \in [D]$. This yields $q_d = \begin{cases} \left ( 1 - x^n \right )^{\frac {m^{D - d + 1} - 1} {m - 1}}, & m \ge 2 \\
         \left ( 1 - x^n \right )^{D - d + 1}, & m = 1
         
         \end{cases}
    $. \qed

\begin{lemma}[Upper bound when $m \ge 2$] \label{lemma:upper_bound_tree}

Under the tree structure and $m \ge 2$ the expected size obeys $\ev {} {S} \le \frac {K x^n (D - 1)} {2} = U(x)$. 

\end{lemma}

From \cref{lemma:probability_functional} we have $q_d =\left ( 1 - x^n \right )^{\frac {m^{D - d + 1} - 1} {m - 1}}$. Using the inequality $(1 - t)^a \le \frac 1 {1 + ta}$ for $a > 0$ and $t \in (0, 1)$ we get that 
    \begin{align*}
        \ev {} {S} & = \sum_{d = 1}^D m^{d - 1} q_d \le \sum_{d = 1}^D m^{d - 1} \frac {1} {1 + x^n \frac {m^{D - d + 1}} {2(m - 1)}} \asymp \int_{t = 1}^D \frac {m^{t - 1}} {1 + x^n \frac {m^{D - t + 1}} {2(m - 1)}}.
    \end{align*}
By letting $u = \frac {x^n m^{D - t + 1}} {2 (m - 1)}$ we get that the above integral equals
\begin{align*}
    & \frac {m^D x^n} {2 \log m (m - 1)} \log \left ( \frac {(1 - x^n)(1 - p)^{\varepsilon K} (1 + (1 - x^n)(1 - p)^{K})} {(1 - x^n)(1 - p)^{K} (1 + (1 - x^n)(1 - p)^{\varepsilon K})} \right ) \bigg |_{u_2 = x^n m^D / 2(m-1)}^{u_1 = x^n m / 2(m-1)} \\ & \le \frac {m^D x^n} {2 \log m (m - 1)} \log \left ( m^{D - 1} \right ) \le \frac {K x^n (D - 1)} {2} = U(x).
\end{align*} \qed

\begin{lemma}[Lower bound when $m \ge 2$] \label{lemma:lower_bound_tree}
    Under the tree structure and $m \ge 2$ the expected cascade size obeys $\ev {} {S} \ge K (1 - x^n(D - 1)) = L(x)$ where $K = m^D - 1$ is the number of products. 
\end{lemma}

        \xpar{Proof}  From \cref{lemma:probability_functional} we have $q_d =\left ( 1 - x^n \right )^{\frac{m^{D - d + 1} - 1} {m - 1}}$. By Bernoulli's inequality $$q_d \ge 1 - x^n \left (\frac {m^{D - d + 1} - 1} {m - 1} \right).$$
    
    Since every level has $m^{d - 1}$ nodes we have that 
    \begin{align*}
        \ev {} {S} & = \sum_{d = 1}^D m^{d - 1} q_d \ge \sum_{d = 1}^D m^{d - 1} \left [ 1 - x^n \left (\frac {m^{D - d + 1} - 1} {m - 1} \right) \right ] \ge K (1 + x^n) - x^n \frac {1} {m - 1} \sum_{d = 1}^D m^{d - 1} m^{D - d + 1} \\ &  = K (1 + x^n) - x^n \frac {1} {m - 1} \sum_{d = 1}^D m^D = K (1 + x^n) - x^n D (K - 1) \ge K (1 - x^n (D - 1)).
    \end{align*}
\qed

    \bigskip

     \xpar{Proof of the main result: \cref{theorem:tree_resilience}} We now proceed with the proof of \cref{theorem:tree_resilience}:

    \xpar{Lower bound} Depending on the range of $m$ we have two choices

    \begin{compactitem}
        \item  \textbf{Case where $m = 1$ (path graph).} For every level $\tau \in [D]$ of the path, we have that $\Pr[S \ge D - \tau] = \Pr \left [ \bigcap_{d > \tau} \{ Z_i = 1 \} \right ] = \Pr [Z_1 = 1] \Pr [Z_2 = 1 | Z_1 = 1] \dots = \prod_{d > \tau} (1 - x^n) = (1 - x^n)^{D - \tau}$. We let $\tau = \varepsilon D$ for some $\varepsilon \in (0, 1)$ and thus $\Pr [S \ge (1 - \varepsilon) D] = (1 - x^n)^{(1 - \varepsilon)D}$. We want to make this probability at least $1 - 1 / D$, and therefore, the resilience of the path graph is $R_{\cG}(\varepsilon) \ge \left ( 1 - \left ( 1 - \frac 1 D \right )^{\frac 1 {(1 - \varepsilon) D)}} \right)^{1/n}$. Since $K = D$ we get the desired result.  

        \item \textbf{Case where $m \ge 2$ (tree graph).} Let $\cK_d$ be the products of tier $d$. We let $\tau = \sup \{ d \in [D] : \exists i \in \cK_d : Z_i = 0 \}$ be the bottom-most tier for which a product failure occurs. If at level $\tau$ a failure occurs, then all levels above $\tau$ are deactivated. The probability that all products up to tier $\tau$ operate is given by 
        \begin{align*}
             \Pr [\text{all products up to tier $\tau$ are producible}] & = \Pr \left [ \bigcap_{d > \tau} \bigcap_{i \in \cK_d} \{ Z_i = 1 \} \right ] \\&= \prod_{d = D}^{\tau + 1} \Pr \left [ \bigcap_{i \in \cK_d} \{ Z_i = 1 \} \bigg | \bigcap_{d' > d} \bigcap_{i \in \cK_{d'}} \{ Z_i = 1 \} \right ] \\
            & = \prod_{d = D}^{\tau + 1} (1 - x^n)^{m^d} = \left ( 1 - x^n \right )^{\sum_{d = D}^{\tau + 1} m^d} \\&= \left ( 1 - x^n \right )^{\frac {m^{D + 1} - m^{\tau + 1}} {m - 1}}
        \end{align*}
        
        Also, $\{ \text{all products up to tier $\tau$ operate} \} \implies \{ S \ge \frac {m^{D + 1} - m^{\tau + 1}} {m - 1} \}$. Therefore, the tail probability of $S$ for $\tau \in [D]$ is given by $\Pr \left [S \ge \frac {m^{D + 1} - m^{\tau + 1}} {m - 1} \right ] = \Pr \left [ S \ge \underbrace {\left ( 1 - \frac {m^{\tau + 1}} {m^{D + 1}} \right )}_{:= 1 - \varepsilon}  \frac {m^{D + 1}} {m - 1}  \right ] \ge \left ( 1 - x^n \right )^{(1 - \varepsilon) \frac {m^{D + 1}} {m - 1}}$. For large enough $D$ we approach the continuous distribution and thus $\Pr [S \ge (1 - \varepsilon) K] \ge (1 - x^n)^{(1 - \varepsilon)K}$. Letting the above be at least $1 - 1 / K$, we get $R_{\cG}(\varepsilon) \ge \left [ 1 - \left ( 1 - \frac 1 {K} \right )^{\frac {1} {(1 - \varepsilon) K}} \right ]^{1/n}$. 
            \end{compactitem}

    \xpar{Upper bound} To derive an upper bound, we have the following cases, depending on the value of $m$

    \begin{compactitem}
        \item     \textbf{Case $m = 1$.} We follow the same logic as the $m \ge 2$ case, and upper bound $\ev {} {S} \le \sum_{d \ge 0} (1 - x^n)^d = \frac {1} {x^n}$ which yields an upper bound $R_{\cG}(\varepsilon) < \left ( \frac {2} {D (1 - \varepsilon)}  \right )^{1/n} \to 0$ as $D \to \infty$. 
        \item \textbf{Case $m \ge 2$.} By Markov's Inequality we get that $\Pr [S \ge (1 - \varepsilon) K] \le \frac {\ev {} {S}} {(1 - \varepsilon) K}$. \cref{lemma:upper_bound_tree}, implies that $\ev {} {S} \le \frac {K D x^n} {2}$, thus $\Pr [S \ge (1 - \varepsilon) K] \le \frac {K D x^n} {2 (1 - \varepsilon) K}$. To make the RHS equal to $1 / 2$, it suffices to pick $x = \left ( \frac {1 - \varepsilon} {D} \right )^{1/n}$. By \cref{lemma:upper_bound_resilience} we get that $R_{\cG}(\varepsilon) < \left ( \frac {1 - \varepsilon} {D} \right )^{1/n} \to 0$ as $D \to \infty$. 

    \end{compactitem}

\subsubsection{Proof of \texorpdfstring{\cref{theorem:gw_resilience}}{theorem:gwresilience}: Inverted lineage growth}\label{app:proof:theorem:gw_resilience}

\xpar{Helper Lemmas} In order to prove \cref{theorem:gw_resilience}, we first prove this auxiliary lemma:

\begin{lemma} \label{lemma:gw_roots}
    For $\tau$ finite, $\frac {\one \{ \mu > 1 \}} {\log \mu} < \alpha < \frac 1 2$ , and $0 < \beta < \one \{ \mu < 1 \} + \one \{ \mu > 1 \} \frac {\log \mu - 1} {\mu} $, let 
    \begin{align*}
        \phi(z)  = z \frac {\mu^\tau z^\tau - 1} {\mu z - 1} - \alpha \frac {\mu^\tau - 1} {\mu - 1}, \text{ for } z \neq \frac 1 \mu, \text { and } \;
        \psi(z) = \frac {\mu^\tau - 1} {\mu - 1} - z \frac {\mu^\tau z^\tau - 1} {\mu z - 1} - \beta, \text{ for } z \neq \frac 1 \mu.
    \end{align*}

    Then

   \begin{compactenum}
       \item If $\mu < 1$, then there exist $z_1, z_2 \in (0, 1)$ such that $\phi(z_1) = \psi(z_2) = 0$. 
       \item If $\mu > e^2$, then there exists $z_1 \in (1/\mu, 1)$ such that $\phi(z_1) = 0$.
       \item If $\mu > e$, then there exists $z_2 \in (1/\mu, 1)$ such that $\phi(z_2) = 0$. 
   \end{compactenum} 

\end{lemma} 

\xpar{Proof}

    \xpar{Analysis for $\phi(z)$} We do case analysis: 
    \begin{compactitem}
        \item If $\mu < 1$ then $\phi$ is defined everywhere in $[0, 1]$ and is also continuous. It is also easy to prove that $\phi$ is increasing in $[0, 1]$ since its the product of two non-negative increasing functions, $z$ and  $ \frac {(\mu z)^\tau - 1} {\mu z - 1} =\sum_{i = 0}^{\tau - 1} (\mu z)^i$. Moreover, note that $\phi(0) < 0$ and $\phi(1) > 0$. Therefore, there exists a unique solution $z_1 \in (0, 1)$ such that $\phi(z_1) = 0$. 
        
        \item If $\mu > e^2$, we study $\phi$ in $(1/\mu, 1]$. Again, $\phi$ is increasing (for the same reason as above), continuous in $(1/\mu, 1]$, and has $\phi(1) > 0$. We also have that, by using L'H\^ospital's rule,
        \begin{align*}
            \lim_{z \to 1 / \mu} \frac {\mu^\tau z^\tau - 1 } {\mu z - 1} & = \lim_{z \to 1 / \mu} \frac {(\mu^\tau z^\tau - 1)'} {(\mu z - 1)'} = \lim_{z \to 1 / \mu} \frac {\mu^\tau \tau z^{\tau - 1}} {\mu} = \tau \implies \\
            \lim_{z \to 1 / \mu} \phi(z) & = \frac {\tau} {\mu} - \alpha \frac {\mu^\tau - 1} {\mu - 1} < \frac {\tau (1 - \alpha \log \mu)} {\mu} < 0 \text{ for } \alpha > \frac {1} {\log \mu}.
        \end{align*}

        Therefore, for $\alpha \in (1/\log \mu, 1/2)$, there exists a unique solution $z_1 \in (1/\mu, 1]$ such that $\phi(z_1) = 0$. 
    \end{compactitem}

    \xpar{Analysis for $\psi(z)$} Note that $\psi$ is a decreasing function of $z$. We do case analysis:

    \begin{compactitem}
        \item If $\mu < 1$, then $\psi$ is defined everywhere in $[0, 1]$ and is continuous in $[0, 1]$. We have that $\psi(0) > 0$ and $\psi(1) < 0$ therefore there exists a unique solution $z_2$ such that $\psi(z_2) = 0$. 
        \item If $\mu > e$, then $\psi$ is decreasing and contunuous in $(1/\mu, 1]$, with $\psi(1) < 0$. We also have that $\lim_{z \to 1 / \mu} \psi(z)  = \frac {\mu^\tau - 1} {\mu - 1} - \frac {\tau} {\mu} - \beta > \frac {\tau (\log \mu - 1 - \mu \beta)} {\mu} > 0 \text { for } \beta < \frac {\log \mu - 1} {\mu}$.
    \end{compactitem}

    % For $\tau \to \infty$, it's easy to see that $z_1 \to 1$. 

\qed

Subsequently, we prove \cref{theorem:gw_resilience}:

\xpar{Proof of the main result: \cref{theorem:gw_resilience}}

\begin{compactitem}
    \item \xpar{Upper bound} Let $\tau = \inf \{ d \ge 1 : |\cK_d| = 0 \}$ be the extinction time of the GW process. In order to establish an upper bound on the resilience, it suffices to set the expected number of surviving products to be at most $\frac {1 - \varepsilon} {2} \ev {\cG} {K}$, since by Markov's inequality the probability of a fraction of at least $(1 - \varepsilon)$-fraction of products surviving would be at most $1 / 2$ and by \cref{lemma:upper_bound_resilience} we would get an upper bound on the resilience $R_{\cG}(\varepsilon)$; that is, $\Pr_{\cG, x} [S \ge (1 - \varepsilon) \ev {\cG} {K}] \le \frac {\ev {\cG, x} {S}} {(1 - \varepsilon) \ev {\cG} {K}} \le \frac 1 2$. Conditioned on $Z_1 = 1$, which occurs with probability $1 - x^n$, the surviving products grow as a GW process with mean $\mu_x = (1 - x^n) \mu$. Therefore due to \cref{lemma:upper_bound_resilience}, we it suffices that $\ev {\cG, x} {\cS | \tau} = \frac {1 - \varepsilon} {2} \ev {\cG} {K | \tau}$ holds for any extinction time $\tau$. This is equivalent to
    \begin{align}
        (1 - x^n) \frac {\mu_x^\tau - 1} {\mu_x - 1} & \le \frac {1 - \varepsilon} {2} \frac {\mu^\tau - 1} {\mu - 1} \iff (1 - x^n) \frac {\mu^\tau (1 - x^n)^\tau - 1} {\mu (1 - x^n) - 1}  \le \frac {1 - \varepsilon} {2} \frac {\mu^\tau - 1} {\mu - 1}  \label{eq:upper_bound_gw}
    \end{align}

    We have the following cases.

    \begin{compactenum}    
        \item If $\mu < 1$ then $\Pr [\tau < \infty] = 1$ (i.e., the process goes extinct after a finite number of steps), then the upper bound on the resilience is always finite due to \cref{lemma:gw_roots} which can be found by numerically solving \cref{eq:upper_bound_gw}. \cref{fig:gw_subcritical} is an example illustrating the resilience bound. 
        \item If $\mu (1 - x^n) > 1$, then $\Pr [\tau = \infty] > 0$ and in this case \cref{eq:upper_bound_gw} is only feasible if and only iff $x = 0$, at which case the upper bound on the resilience is 0, and the GW process is not resilient. If $\tau < \infty$, which happens with non-zero probability, the  upper bound on the resilience is finite when $\mu > e^2$ due to \cref{lemma:gw_roots}. 
    \end{compactenum}

    For a specific triplet $(\mu, \tau, \varepsilon)$, let $\overline x(\mu, \tau, \varepsilon)$ be the smallest possible solution to \cref{eq:upper_bound_gw}, which exists for $\mu \in (0, 1) \cup (e^2, \infty)$ due to \cref{lemma:gw_roots}. Then the expected upper bound on resilience $\ev {\cG} {\overline R_\cG (\varepsilon)}$, can be expressed as $\ev {\cG} {\overline R_\cG (\varepsilon)} = \ev {\tau} {\overline R_\cG (\varepsilon)} = \sum_{1 \le k < \infty} \Pr [\tau = k] \overline x(\mu, \tau, \varepsilon) > 0$.

    \item \xpar{Lower bound} Similarly to the upper bound, in order to devise a lower bound, it suffices to set $\ev {\cG} {K | \tau} - \ev {\cG} {S | \tau}$ to be at most $\varepsilon$ for any extinction time $\tau$, since, again, by Markov's inequality, we are going to get that the probability that at least a $\varepsilon$-fraction of products fails is at most $\frac {1} {\ev {\cG} {K | \tau}}$; namely, $\Pr_{\cG, x} [F \ge \varepsilon \ev {\cG} {K} | \tau] \le \frac {\ev {\cG, x} {F | \tau}} {\varepsilon \ev {\cG} {K | \tau}} \le \frac 1 {\ev {\cG} {K | \tau}}$. This yields 
    \begin{align}
        \frac {\mu^\tau - 1} {\mu - 1} - (1 - x^n) \frac {\mu_x^\tau - 1} {\mu_x - 1}  \le \varepsilon  \iff         \frac {\mu^\tau - 1} {\mu - 1} - (1 - x^n) \frac {\mu^\tau (1 - x^n)^\tau - 1} {\mu (1 - x^n) - 1}  \le \varepsilon \label{eq:lower_bound_gw}
    \end{align}

    \noindent Similarly to the upper bound, we have the following cases.

    \begin{compactenum}
        \item In the subcritical regime $\mu < 1$, we can again prove that the lower bound is always finite due to \cref{lemma:gw_roots}. 
        \item In the supercritical regime $\mu (1 - x^n) > 1$, we have that when $\tau < \infty$, which happens with positive probability then for $\mu > e$ from \cref{lemma:gw_roots} we get the existence of the resilience. When $\tau = \infty$, we have again that the only way \cref{eq:lower_bound_gw} can hold is iff $x = 0$. 
    \end{compactenum}

     For a specific triplet $(\mu, \tau, \varepsilon)$, let $\underline x(\mu, \tau, \varepsilon)$ be the largest possible solution to \cref{eq:lower_bound_gw}, which exists for $\mu \in (0, 1) \cup (e, \infty)$ due to \cref{lemma:gw_roots}. Then the expected lower bound on the resilience $\ev {\cG} {\underline R_\cG (\varepsilon)}$, can be expressed as 
 $$\ev {\cG} {\underline R_\cG (\varepsilon)} = \ev {\tau} {\underline R_\cG (\varepsilon)} = \sum_{1 \le k < \infty} \Pr [\tau = k] \underline x(\mu, \tau, \varepsilon) > 0.$$
 
    \xpar{Determining $\Pr [\tau < \infty]$ when $\mu (1 - x^n) > 1$} It is known from the analysis of GW processes (see, e.g., \citet{galtonwatson_notes}) that the extinction probability $\Pr [\tau < \infty]$ can be found as the smallest solution $\eta \in [0, 1]$ to the fixed-point equation $\eta = G_{\cD}(\eta)$ where $G_{\cD}(s) = \ev {\xi \sim \cD} {e^{s \xi}}$ is the moment generating function of the branching distribution $\cD$. 

\end{compactitem}

\newpage

\section{Linear programming operationalization of resilience in general networks} \label{app:linear_program}

Thus far, we have focused on simple network structures where analytical expressions for $\ev{}{F}$ are tractable,
and, subsequently, by bounding $\ev {} {F}$ we can determine the upper and/or lower bounds for $R_{\cG}(\varepsilon)$, because by Markov's inequality the probability of the event $\{ F \ge \varepsilon K \}$ is bounded above by $\ev {} {F} / \varepsilon K$. These include hierarchical one-directional input-output random production networks, Two-tier network, and other hierarchical structures considered in \cref{app:architectures}. However, in general, computing $\ev{}{F}$ is not easy \citep{chen2010scalable}. 
Moreover, it is certainly possible for a network to contain cycles, e.g., when complex products are used in the production of simpler products, or to represent recycling and other forms of circular flows of materials or resources in modern economies; cf. circular economies \citep{geissdoerfer2017circular}. In this section, we offer techniques to address the following questions for general network topologies:  
\begin{compactenum}
\item \emph{How do we efficiently calculate the cascade size $F$?}
\item \emph{How do we identify network vulnerabilities in the failure of their most critical nodes?}
\end{compactenum}  In the following, we answer the above questions and provide a systematic way to treat general graphs that undergo a joint percolation process. More specifically, we systematically bound the expected number of failures and subsequently derive bounds for the resilience metric. Finally, we show that our analysis has deep connections to financial networks. 

Additionally, we assume that firms hold inventories, which is not modeled as an explicit stocking policy or cost-based decision, but serves as a structural decoupling mechanism that weakens dependencies between stages during shocks. Formally, each edge in the network is kept with probability $y$. Generally, if product $i$ has enough stock of input $j$, it can remain operational even if $j$ fails. This creates a subnetwork $\cG_y \subseteq \cG$, which, under reasonable assumptions for $x$ and $y$, can be used to identify vulnerable products whose failure can trigger the largest cascades. Formally, we denote the resilience when the edge-retention probability is $y$ by $R_{\cG}(\varepsilon; y)$, which is always greater than the resilience without structural decoupling $(y=0)$; see \cref{prop:rel_resilience}. 

We denote the cascade size (number of failures) in $\cG_y$ with $\ev {} {F_y}$. To approximate $\ev{}{F_y}$, we formulate the following linear program and its dual, parameterized by the shock vector among the products $u = (u_1, \dots, u_K)$ which depends on the adjacency matrix\footnote{For each directed edge between two products $i \to j$ we set $A_{ij} = 1$.} $A$ of the network $\cG$:

\begin{align}
	p_\cG^*(u; y) & = \max_{\beta \in [\zero, \one]}  \one^T \beta \quad &
	\text{s.t.} \quad & \beta \le y A^T \beta + u, \label{eq:upper_bound_lp} \\
        d_{\cG}^*(u; y) & = \min_{\gamma, \theta \ge \zero}  u^T \gamma + \one^T \theta \quad & \text{s.t.}  \quad  & (I - yA) \gamma  + \theta \ge \one. \label{eq:upper_bound_lp_dual} 
\end{align}

We denote optimal solutions for the primal and dual problems by $\beta_{\cG}^*(u; y)$ and $\left ( \gamma_{\cG}^*(u; y), \theta_{\cG}^*(u; y) \right )$, respectively. We omit arguments and simply write $p^*, d^*, \beta^*$ and $\gamma^*$ when the context is clear.

The following theorem clarifies the relationship between the LP aforementioned and the cascade size $\ev {} {F_y}$, by upper and lower bounding $\ev {} {F_y}$ in terms of the  LP solution $p^*$; proof in \cref{app:proof:theorem:general_ub}.

\begin{theorem} \label{theorem:general_ub}
    Let $\cG$ be a network that undergoes a joint percolation process with the probability of supplier failure $x$ and the edge-retention probability $y$. If $p^*$ is the optimal value of the primal problem in \cref{eq:upper_bound_lp} for $u = \one x^n$, then for $y = \varrho / \mu$, with $\varrho \in (0, 1)$, we have $\ev{}{F_y} \le p^* \le \left(1 + \varrho + O(\varrho^2)\right) \ev{}{F_y}$. The solution to LP can be found as the unique fixed point $\beta^*$ of the contraction $\Phi(\beta) = \min\{\one , y A^T \beta + x^n \one\}$.
\end{theorem}

\cref{theorem:general_ub} shows that there is a systematic way of bounding $\ev {} {F_y}$ by solving a linear program or a fixed-point equation (if the edge retention probability is less than $1/\mu$). We expand on this computational angle in \cref{app:algorithm_lp}. It is instructive to note that such an elegant approach to bounding the expected number of failures---which otherwise is \textit{\#P-hard} to compute \citep{chen2010scalable,borgs2014maximizing}---becomes possible by introducing sampling at the edge level, which reduces network dependencies. 

\cref{theorem:general_ub} demonstrates the utility of the LP formulation along with the granular (edge-level) network structure for bounding the expected number of failures when the edge retention probability ($y$) is small. We can complement \cref{theorem:general_ub} with bounds that use a global network feature (sourcing requirement $m$) to show how $\ev {} {F_y}$ approaches $\ev {} {F}$ from below, giving sharp bounds as $y\to 1$: 

 \begin{theorem} \label{theorem:Fy_approximation}
    %For $y < 1 / \mu$, we have $\ev {} {F_y} \ge \frac {\ev {} {F} - K(1-y^{\mu}))} {y^{\mu}}$. Moreover, $\ev{}{F_y} \le \left( \frac{3}{4} + o(1) \right) \ev{}{F}$.
    For all $0< x, y < 1$, we have: $$\frac {1}{y^{m}}\ev {} {F} + \left (1 - \frac{1}{y^{m}} \right ) K \leq \ev {} {F_y} \le (1-(1-x^n)(1-y)^m) \ev{}{F}.$$ %\left( \frac{3+x^n}{4}  \right) \ev{}{F}$.
 %Moreover, $\ev {} {F_y} \le \left ( \frac {3} {4} + o(1)  \right ) \ev {} {F}$.
 \end{theorem}

After defining $p^*_{\cG}$ and its dual $d^*_{\cG}$, the next question is how they are related to resilience. In the following (proved in \cref{app:proof:theorem:lp_duality_resilience}), we show that we can bound the resilience using the sum of Katz centralities, defined as $$\beta_\cG^\katz(y) = (I - yA^T)^{-1} \one.$$

\begin{theorem}[\cref{theorem:lp_duality_resilience_2} of the main text] \label{theorem:lp_duality_resilience}
    If $y < 1/ \mu$, the resilience satisfies $R_{\cG}(\varepsilon; y)  \ge \left (\frac {\varepsilon} {\one^T \beta_{\cG}^\katz (y)} \right )^{1/ n}$.
\end{theorem}

\subsection{An Algorithm for Solving the LP Efficiently} \label{app:algorithm_lp}

The last part of \cref{theorem:general_ub} provides an algorithm to find $\beta^*$ by solving the fixed point problem, i.e., if $\beta^{(t)}$ corresponds to the failure probabilities in iteration $t$, then $\beta^*$ can be computed via the following contraction map iterations (valid when $y < 1/ \mu$): 
\begin{align}
    \beta^{(t)} = \min\{\one , yA^T \beta^{(t - 1)} + x^n \one \}. \label{eq:contarction_iteration}
\end{align}

To compute a solution with precision $\eta$, we iterate \cref{eq:contarction_iteration} for $\log(2K/\eta)/\log(1/\varrho)$ steps, where $\varrho < 1$ is the Lipschitz constant for the contraction map and $2K$ bounds the $L_1$ norm of the initial condition.  This algorithm has a runtime of $O \left ( \frac {(K + M)\log (K / \eta)} {\log (1 / \varrho)} \right ) = \tilde O(K + M)$ to yield a solution that is close to $\beta^*$ by some accuracy $\eta$, and subsequently a $1 + \eta + \varrho + O(\varrho^2)$ approximation of $\ev{}{F_y}$. The runtime matches the lower bound (up to logarithmic factors) of $\Omega(K + M)$ for the influence maximization problem, which uses an equivalent formulation of cascade sizes \citep{borgs2014maximizing}. 

% We also show that the following approximation bounds hold for $F_y$ and that $\ev {} {F_y}$ can never give a better approximation than $(3+x^n)/4$: 

%  \begin{theorem} \label{theorem:Fy_approximation}
%     %For $y < 1 / \mu$, we have $\ev {} {F_y} \ge \frac {\ev {} {F} - K(1-y^{\mu}))} {y^{\mu}}$. Moreover, $\ev{}{F_y} \le \left( \frac{3}{4} + o(1) \right) \ev{}{F}$.
%     For all $0< x, y < 1$, we have: $$\frac {1}{y^{\mu}}\ev {} {F} - (\frac{1}{y^{\mu}} -1) K \leq \ev {} {F_y} \le (1-(1-x^n)(1-y)^m) \ev{}{F}.$$ %\left( \frac{3+x^n}{4}  \right) \ev{}{F}$.
%  %Moreover, $\ev {} {F_y} \le \left ( \frac {3} {4} + o(1)  \right ) \ev {} {F}$.
%  \end{theorem}

% \textcolor{red}{Combining \cref{theorem:general_ub,theorem:Fy_approximation}, we get that the LP yields a $3/4 + o(1)$ approximation to $\ev {} {F}$ for $y < 1/m$.} 
%\cref{theorem:general_ub} shows that there is a systematic way of bounding $\ev {} {F}$ and $\ev {} {F_y}$ by solving a linear program or a fixed-point equation (if the edge retention probability is less than $1/{\mu}$). It is instructive to note that such an elegant upper bound on the expected number of failures becomes possible by introducing sampling at the edge level, which reduces network dependencies. 

\subsection{Proofs of results}

\subsubsection{Helper Lemma}

 A simple coupling argument for two values $0 < y_1 \le y_2 < 1$ shows that edge retentions in $\cG_{y_1}$ are greater than in $\cG_{y_2}$: 
\begin{lemma} \label{prop:rel_resilience}
    If $0 < y_1 \le y_2 \le 1$, then $R_{\cG}(\varepsilon; y_1) \ge R_{\cG}(\varepsilon; y_2)$. In particular, $R_{\cG}(\varepsilon; y) \ge R_{\cG}(\varepsilon)$ for all $y \in (0, 1]$.
\end{lemma} 

\xpar{Proof} We prove this result using a coupling argument: Let \(\cG\) be the original production network. For each edge \((i,j) \in E(\cG)\), generate a uniform random variable \(U_{ij} \sim \cU[0,1]\). Construct two edge-subsampled networks:
\begin{itemize}
    \item \(\cG_{y_1}\), where edge \((i,j)\) is retained if \(U_{ij} \leq y_1\),
    \item \(\cG_{y_2}\), where edge \((i,j)\) is retained if \(U_{ij} \leq y_2\),
\end{itemize}
with \(0 < y_1 \leq y_2 \leq 1\). Clearly, this coupling ensures that \(\cG_{y_2} \supseteq \cG_{y_1}\) almost surely.

Recall that in our model, each product fails if at least one of its required inputs fails or becomes inaccessible. Since \(\cG_{y_1}\) contains at most as many dependencies as \(\cG_{y_2}\), and fewer dependencies imply a lower likelihood of cascading failures, it follows that the cascade size in \(\cG_{y_1}\) is stochastically smaller than in \(\cG_{y_2}\).

Therefore, for any fixed failure probability \(x\), the probability that at least a \((1 - \varepsilon)\) fraction of the products survive is higher in \(\cG_{y_1}\) than in \(\cG_{y_2}\). Taking the supremum over such \(x\) values implies
\[
R_\cG(\varepsilon; y_1) \geq R_\cG(\varepsilon; y_2),
\]
as desired.

Because $y < 1$ we can obtain an upper bound on resilience by limiting failures to at most $\varepsilon K$ products. By the Markov inequality: $\Pr [F_y \ge \varepsilon K] \le {\ev {} {F_y}} /{\varepsilon K}$, which requires bounding the expected number of failed products, $\ev{}{F_y} = \sum_{i \in \cK} \Pr[Z_i = 0]$%which requires an upper bound on
; recalling the notation of \Cref{eq:dynamics}, $Z_i$ is an indicator variable for the failure of product~$i$.

\subsubsection{Proof of \texorpdfstring{\cref{theorem:general_ub}}{theorem:generalub}: Linear program to bound the expected cascade size}\label{app:proof:theorem:general_ub}

    \xpar{Derivation of the LP} Let $u_i$ be the probability that the product $i$ fails spontaneously\footnote{In the simple case that suppliers fail independently at random with probability $x$, we have $u_i = x^n$; more generally, the derivation holds for any $u_i \in [0, 1]$.}.

    For $i \in \cK$ let $\beta_i = \Pr [Z_i = 0] \in [0, 1]$ and $\beta = (\beta_1, \dots, \beta_K)^T$. By the union bound, we have that 
	\begin{align*}
	\beta_i & = \Pr [(\exists j \in \cN(i) : (i, j) \text{ survives edge pruning } \wedge Z_j = 0) \vee (\forall s \in \cS(i), X_{is} = 0)] \\
	& \le \Pr [\exists j \in \cN(i) : (i, j) \text{ survives edge pruning } \wedge Z_j = 0] + \Pr [\forall s \in \cS(i), X_{is} = 0] \\
	& \le y \sum_{j \in N(i)} \beta_j + u_i.
	\end{align*}

	\xpar{Upper and lower bounds of the LP} The number of failed products equals $\ev {} {F} = \sum_{i \in \cK} \beta_i$. Thus, finding the upper bound on $\ev {} {F_y}$ corresponds to solving the following LP,
	\begin{align}
		p^*_{\cG}(u; y) = \max _{\beta \in [\zero, \one]} \quad & \sum_{i \in \cK} \beta_i \qquad \text{s.t.} \qquad & \beta \le y A^T \beta + u,
        \label{eq:EN-optimization-problem-financial-clearing}
	\end{align} which proves $\ev {} {F_y} \leq p^*$.
	%When $y \| A^T \|_1 < 1$, which is equivalent to $y < \frac {1} {\mu}$, this problem is the financial clearing problem of \citet{eisenberg2001systemic}, and from Lemma 4 of \citet{eisenberg2001systemic}, we know that we can also compute $\beta$ by solving the fixed point equation $\beta = \one \wedgeMIN (y A^T \beta + u) = \Phi(\beta)$. Since $y < 1 / \mu$, the mapping is a contraction and has a unique fixed point due to Banach's theorem. 
    To obtain the upper bound, note that if $\beta^*$ is an optimal solution, then: 
    \begin{align*}
        p^*_\cG(u; y) = \sum_{i = 1}^K \beta_i^* & \stackrel{(i)}{\le} \sum_{i = 1}^K y \sum_{j \in N(i)} \beta_j^* + \sum_{i = 1}^K u_i \stackrel{(ii)}{\le} y \sum_{i = 1}^K \sum_{j \in N(i)} \beta_j^* + \ev {} {F_y} \\
        & \le y \sum_{j = 1}^K \sum_{i: j \in N(i)} \beta_j^* + \ev {} {F_y} \le y \mu \sum_{i = 1}^K \beta_i^* + \ev {} {F_y} \\
        & \le y \mu p^*_\cG(u; y) + \ev {} {F_y},
    \end{align*} where (i) follows from the union bound constraint and (ii) from the fact that the expected number of products that fail due to spontaneous failure of their supplier $\sum_{i = 1}^K u_i$ lower bounds $\ev{}{F_y}$. Rearranging the terms, we get: 
    \begin{align*}
        p^*_\cG(u; y) \le \frac {1} {1 - \mu y} \ev {} {F_y}.
    \end{align*}
    
    Using the fact that $\frac {1} {1 - \mu y} = 1 + \mu y + O((\mu y)^2)$ we get the right-hand side when $\varrho = \mu y<1$. 

  \xpar{Fixed-point Algorithm} When $y \| A^T \|_1 < 1$, which is equivalent to $y < 1 / \mu$, the LP in \cref{eq:EN-optimization-problem-financial-clearing} is the financial clearing problem of \citet{eisenberg2001systemic}, and from  \citet[Lemma 4]{eisenberg2001systemic}, we know that we can compute $\beta$ by solving the fixed point equation $\beta = \min\{\one, y A^T \beta + u\} = \Phi(\beta)$. Since $y < 1 / \mu$, the mapping is a contraction and has a unique fixed point due to Banach's theorem.

\subsubsection{Limitations of the LP-based bound: Inapproximability for \texorpdfstring{$y > 1/\mu$}{yll1m}} \label{app:innaproximability}

We show that the LP-based bound may become very loose when $y > 1/ \mu$. In the following theorem we construct a family where the gap between the LP and the expected cascade size grows as $\Omega(K)$: 

\begin{theorem}
Fix an integer $\mu \ge 2$ and let $y > 1/\mu$. Then for every $K$ divisible by $\mu+1$, there exists a directed production network $\cG_K$ with $K$ nodes and maximum sourcing influence $\mu$, together with a homogeneous shock vector $u_i = \lambda / K$ for some fixed $\lambda > 0$, such that the LP has optimal value $p^* = K$, while the true expected cascade sizes satisfy $\ev {} {F_{y}} \le \ev {} {F} \le \lambda(\mu+1) + o(1)$.
Consequently, $p^* - \ev {} {F}  = \Omega(K)$ and $p^* - \ev {} {F_y} = \Omega(K)$.
\end{theorem}

\xpar{Proof} Let $q= \mu + 1$ and assume $K = qM$ for some integer $M$. Construct $G_K$ as the disjoint union of $M$ identical components, each component being the complete
directed graph on $q$ vertices. In each component, every vertex has exactly $m = \mu$ inputs and outputs.

We first bound the true cascade. Since each vertex in a component depends on
all other $\mu$ vertices in that same component, a component fails if and only if at least one of its vertices suffers a spontaneous shock.  Thus, if $I_a$ is the indicator that component $a$ is hit by at least one spontaneous failure, then the total number of failed products is $F = q \sum_{a=1}^M I_a$. Now, because $u_i = \lambda / K$ for all $i$, the probability that a given component is hit is
\begin{align*}
    \Pr [I_a = 1] & = \Pr [\text{at least one component fails in $a$}] \\
    & = 1 - \Pr [\text{no components fail in $a$}] \\
    & = 1 - (1 - u_i)^q \\
    & \le 1 - (1 - u_i q) \\
    & = q u_i \\
    & = \frac {q \lambda} {K} 
\end{align*}

Therefore, $\ev {} {F} \le q \lambda = (\mu + 1) \lambda = O(1)$. Moreover, due to the coupling result between $F$ and $F_y$ we also get that $\ev {} {F_y} \le \ev {} {F} \le (\mu + 1) \lambda = O(1)$. 

\noindent Next, we consider the LP

\begin{align}
    \max_{\beta \in [\zero, \one]^K} \one^T \beta \quad \text{s.t.} \quad \beta \le y A^T \beta + u
\end{align}

We can quickly verify that $\beta = \one$ is an optimal solution to the LP: First, since every vertex has exactly $\mu$ inputs, the union bound constraint yields $yA^T \one + u = y\mu \one + \frac{\lambda}{K}\one$. Because $y\mu>1$, this implies $\one \le yA^T \one + u$ for all sufficiently large $K$. Hence $\beta = \one$ is a feasible solution to the LP with value $K$, so. Since always $\beta_i \le 1$, we also have $p^* \le  K$. Therefore, $\beta = \one$ is an optimal solution with $p^* = K$. Thus, the gap between the two is $p^* - \ev {} {F} = K - O(1) = \Omega(K)$ and $p^* - \ev {} {F_y} = \Omega(K)$, as claimed.

\qed

\subsubsection{Proof of \texorpdfstring{\cref{theorem:lp_duality_resilience}}{theoremlpdualityresilience}} \label{app:proof:theorem:lp_duality_resilience}

If $p^*$ is the optimal value of the primal given in \cref{eq:upper_bound_lp}, and $(\tilde \gamma, \tilde \theta)$ is a feasible dual solution with the objective value $\tilde d$, then from weak duality we have $\ev {} {F} \le p^* \le \tilde d$. If we let $\tilde d \le \varepsilon$, then $\Pr [F \ge \varepsilon K] \le \frac {\ev {} {F}} {\epsilon K} \le \frac {\tilde d} {\varepsilon K}$. Therefore any $x$ such that $\tilde d \le \varepsilon$ will be a lower bound on $R_{\cG}(\varepsilon)$ as it will make the bound at most $1/K$. This is equivalent to 
    
    \begin{align}
        R_{\cG}(\varepsilon) \ge x \ge \left ( \frac {\varepsilon -  \one^T \tilde \theta} {\one^T \tilde \gamma} \right )^{1/n}, \qquad \text{for all} \qquad \tilde \theta, \tilde \gamma \ge \zero \text { s.t. } (I - yA) \tilde \gamma + \tilde \theta \ge \one. 
    \end{align}

    We are interested in the values of $(\tilde \gamma, \tilde \theta)$ that maximize this lower bound, and therefore the best lower bound is given by 
    \begin{align}
        \underline R_{\cG}(\varepsilon) = \max_{\gamma, \theta \ge \zero} \left ( \frac {\varepsilon -  \one^T \theta} {\one^T \gamma} \right )^{1/n}, \qquad \text{s.t.} \qquad (I - yA)\gamma + \theta \ge \one.     
    \end{align}

    Observe that increasing any $\theta_i$ from $\theta_i = 0$ would decrease the lower bound; therefore, the optimal $\theta$ is $\theta^* = \zero$. Moreover, due to monotonicity,

    \begin{align} \label{eq:lower_bound_resilience_dual}
        \underline R_{\cG}(\varepsilon) = \left ( \frac {\varepsilon} {\min_{\gamma \ge \zero, \gamma \neq \zero} \one^T  \gamma} \right )^{1/n}, \qquad \text{s.t.} \qquad (I - yA)\gamma \ge \one.     
    \end{align} yielding our final result. We take the dual of the denominator, which corresponds to $\max_{\beta \ge \zero} \one^T \beta$ subject to $\beta \le y A^T \beta + \one$. If $y < 1/m$, from the main result of \cite{eisenberg2001systemic}---or the KKT conditions---we get that the optimal solution is $\beta_{\cG}^\katz (y)$, where $\beta_\cG^\katz(y) = (I - yA^T)^{-1} \one$ is the sum of the Katz centralities of the network nodes (products). %Similarly for its dual, the optimal solution is $\gamma_{\cG}^\katz (y)$.

\subsection{Proof of \texorpdfstring{\cref{theorem:Fy_approximation}}{theoremFyapproximation}} \label{app:proof:theorem:Fy_approximation}
    To compare $\ev{}{F_y}$ and $\ev{}{F}$ for a given value of $0<x,y<1$, we consider the  percolation processes on $\cG$ and $\cG_y$, through a natural coupling that identifies the status of initial supplier failures on the two graphs. (Recall that the node percolation process in either graph starts with the initial failure of the suppliers with probability $x$ indefinitely at random.) For each product $i\in\cK$, let $F_i$ and $F_{y,i}$ be indicator variables of its failure in $\cG$ and $\cG_y$, respectively. Their coupling implies that $F_{y,i} \leq F_{i}$, which gives the trivial bound $\ev{}{F_y} \leq \ev{}{F}$. Note that in the notation of \cref{sec:node_percolation} we have $F_i = 1-Z_i$.
    
    \xpar{Upper bound} To obtain a potentially tighter upper bound on $\ev{}{F_y}$ when $y< 1$, note that $\ev{}{F_{y,i}} = \Pr [F_{y,i} = 1] = \Pr [F_{y,i} = 1 | F_i = 1] \Pr [F_i = 1]$ because under the natural coupling of the processes a product that survives in $\cG$ does so in $\cG_y$ as well: $\Pr [F_{y,i} = 1 | F_i = 0] = 0$. Next, we can write $\Pr [F_{y,i} = 1 | F_i = 1] = 1- \Pr [F_{y,i} = 0 | F_i = 1]$ and obtain a lower bound on $\Pr [F_{y,i} = 0 | F_i = 1]$ by the probability that product $i$ does not fail on its own and all its sourcing requirements are removed in $\cG_y$. The former is given by $(1-x^n)$ and the latter is lower bounded by $(1-y)^m$, when the source requirement is $m$. By independence, the product of the two gives the desired lower bound:  $$\Pr [F_{y,i} = 0 | F_i = 1] \geq (1-x^n)(1-y)^m.$$ 
    
    Subsequently, we have 
    \begin{align*}
        \ev{}{F_{y,i}} = \Pr [F_{y,i} = 1] = (1- \Pr [F_{y,i} = 0 | F_i = 1]) \Pr [F_i = 1] \leq (1-(1-x^n)(1-y)^m) \ev{}{F_{i}}.
    \end{align*} Summing over all products $i\in\cK$, we get the upper bound claimed: $$\ev{}{F_{y}} \leq (1-(1-x^n)(1-y)^m) \ev{}{F}.$$
    %Given the subsampled graph $\cG_y$, if a product $i$ survives in $\cG$, then it survives in $\cG_y$; and if a product fails in $\cG$, then it survives in $\cG_y$ with probability at least $q = (1 - y)^\mu (1 - x^n)$, which is the probability that the product does not fail on its own and none of its inputs are retained in $\cG_y$. In expectation, we get that $\ev {} {F_y} \le (1 - q) \ev {} {F}$. The quantity $q$ is minimized for $y = 1/\mu$ and satisfies $q \ge (1 - 1/\mu)^{\mu} (1 - x^n) \ge 1 / 4 (1 - x^n)$. This implies the $3 / 4 + x^n / 4$ upper bound.   

    \xpar{Lower bound} The lower bound follows a similar argument by conditioning in the reverse direction and noting that 
    \begin{align*}
        \ev{} {F_i} & =  \Pr [F_i = 1] 
        \\ & = \Pr [F_i = 1 | F_{y,i} = 1] \Pr [F_{y,i} = 1] + \Pr [F_i = 1 | F_{y,i} = 0] \Pr [F_{y,i} = 0]  
        \\ & \le \Pr [F_{y,i} = 1] + (1-y^m) (1 - \Pr [F_{y,i} = 1]) 
        \\ & = \ev {}{F_{y,i}}  + (1-y^m) (1 - \ev{}{F_{y,i}}),
    \end{align*} 
    where we use the fact that $\Pr [F_i = 1 | F_{y,i} = 1] = 1$ and the bound $\Pr [F_i = 1 | F_{y,i} = 0] \leq 1-y^m$. The latter is an upper bound on the probability that a node fails in $\cG$ conditioned on its survival in $\cG_y$, which can only happen if at least one incoming edge incident to $i$ is removed in $\cG_y$ and the probability of the latter is upper bounded by $1-y^m$ when the sourcing requirement is $m$. Summing over the products $i\in\cK$ and rearranging both sides gives the claimed lower bound: $\ev {}{F}/y^m - K(1/y^m -1) \leq \ev{}{F_y}$.
    
    % By linearity of expectation, this gives $\frac {\ev {} {F} - Kq'} {1 - q'} \leq \ev {} {F_y}$. 
    
    % The value of $q'$ corresponds to the probability that for each node, at least $f \ge 1$ neighbors have failed but none of them is sampled in $G_y$ and it is given by: 
    % \begin{align*}
    %     q' = 1 - \left ( 1 - x^n (1 - y) \right )^\mu, 
    % \end{align*}

    % yielding the final result. 

\newpage

\section{Connections to existing supply chain resilience measures: The Risk Exposure Index} \label{app:rei}

Our metric has connections to important metrics already present in the network literature \citep{levi2016identifying,simchi2014superstorms,ham2022companies}, such as the Risk Exposure Index (REI). To define REI, suppose that a product in the network is disrupted by an infinitesimal shock, assuming the same responses from the other nodes. In our case, this corresponds to a change in the probability of shock of the product $i$ from $x$ to $x + \delta$ (for some sufficiently small $\delta$) and its impact on the size of cascading failures, which corresponds to the potential impact $\potentialimpact(i)$. Since it is difficult to quantify an exact formula for the change of $\ev {} {F}$, we will instead focus on the change in $\ev {} {F_y}$ given by \cref{eq:upper_bound_lp}. Specifically, when the partial derivatives from left and right exist, the potential impact for a node is defined as:

{\small
\begin{align} \label{eq:potential_impact}
    \potentialimpact_i (x; y) & = \lim_{\delta \to 0} \frac {p_\cG^*\left ((x^n, \dots, (x + \delta)^n, \dots, x^n)^T; y\right ) - p_\cG^* \left ( (x^n, \dots, x^n, \dots, x^n )^T; y\right )} {\delta} = (nx^{n - 1}) \sum_{j = 1}^{K} \frac {\partial \beta_j^*(u; y)} {\partial u_i} \bigg |_{u = \one x^n}
\end{align}}

Subsequently, the Risk Exposure Index of $\cG$ (REI, \cite{levi2016identifying}) is given as the worst possible magnitude of $\potentialimpact_i(x; y)$, i.e.,
\begin{align} \label{eq:rei}
    \rei_{\cG}(x; y) = \max_{i \in \cK} |\potentialimpact_i(x; y)| = (nx^{n - 1}) \max_{i \in [K]} \left | \sum_{j = 1}^{K} \frac {\partial \beta_j^*(u; y)} {\partial u_i} \bigg |_{u = \one x^n} \right |
\end{align}

In general though, the partial derivatives may not exist. However, we can always show that there is an efficient algorithm for calculating the potential impacts as well as the REI, given by the following result, which leverages duality and the KKT conditions: 

\begin{theorem} \label{theorem:rei_lp}
    If \cref{eq:potential_impact} is differentiable, the potential impact $\potentialimpact_i(x; y)$ can be computed as $\potentialimpact_i(x; y) = n x^{n - 1} \gamma_i^*(x^n \one; y)$ where $\gamma_i^*(x^n \one; y)$ can be found as the solution to the dual problem:

    \begin{align*}
     d_{\cG}^*(u; y) & = \min_{\gamma, \theta \ge \zero}  u^T \gamma + \one^T \theta \quad & \text{s.t.}  \quad  & (I - yA) \gamma  + \theta \ge \one.
     \end{align*}

    \noindent Subsequently, $\rei_\cG(x; y) = (n x^{n - 1}) \max_{i \in [K]} \gamma_i^*(x^n \one; y)$. 

    \noindent Moreover, if \cref{eq:potential_impact} is non-differentiable, we can show that $$\potentialimpact_i^+(x; y) = n x^{n - 1} \min_{(\gamma, \theta) \in \cD} \gamma_i, \quad \text{and} \quad \potentialimpact_i^-(x; y) = n x^{n - 1} \max_{(\gamma, \theta) \in \cD} \gamma_i$$ where $\cD = \arg \min_{\gamma, \theta \ge 0} \left \{ x^n \one^T \gamma + \one^T \theta  : (I - yA)\gamma + \theta \ge \one \right \}$. Subsequently,  $$\rei^+_\cG(x; y) = (n x^{n - 1}) \max_{i \in [K]} \min_{(\gamma, \theta) \in \cD} \gamma_i, \quad \text{and} \quad \rei^-_\cG(x; y) = (n x^{n - 1}) \max_{i \in [K]} \max_{(\gamma, \theta) \in \cD} \gamma_i$$
\end{theorem}

Subsequently, we can also show that when firms have sufficient structural decoupling ($y < \min \{ 1 / \mu, 1/ m \}$) and the shocks are sufficiently small, REI is proportional to the maximum Katz centrality of the nodes in $\cG^R$, because the dual program has a closed form solution: 

\begin{corollary}[\cref{proposition:rei_katz} of the main text] \label{collorary:rei_katz}
    If $y < \min \{ 1 / \mu, 1/ m \}$, and $x < (1 - \mu y)^{1/n}$,  then $\mathrm {REI}_\cG(x; y) = n x^{n - 1} \max_{i \in [K]} \gamma_{i, \cG}^\katz(y).$
\end{corollary}

\subsection{Proofs of results}

\subsubsection{Proof of \texorpdfstring{\cref{theorem:rei_lp}}{theoremreilp}: Potential impact of a product}

Let $e_i$ be the $i$-th baseis vector and let $\Delta_i(\delta) = \left [ (x + \delta)^n - x^n \right ] e_i$ and $f(\Delta) := p^*_\cG(x^n \one + \Delta; y)$, which characterizes the optimal value as a result of the perturbation of the $(I - yA^{T}) \beta - x^n \one \leq 0$ constraint to $(I - yA^{T})\beta - x^n \one \leq \Delta$. 

\xpar{Differentiable case} If $f$ is differentiable, which is equivalent to $d^*_G(x^n \one; y)$ admitting a unique, optimal $(\gamma^*, \theta^*)$ solution, then it follows that because $\gamma^*$ is an optimal Lagrange multiplier for the perturbation of the aforementioned constraint by $\Delta$, i.e., 

\begin{align*}
\potentialimpact_i(x, y) &= \lim_{\delta \to 0} \frac{p^*_\cG(x^n \one + \Delta_i(\delta); y) - p^*_\cG(x^n \one; y)}{\delta} \\
&= \lim_{\delta \to 0} \frac{f(\Delta_i(\delta)) - f(\Delta_i(0))}{\delta} \\
&= \left.\frac{df(\Delta_i(\delta))}{d\delta}\right|_{\delta=0} \\
&= D_\delta \Delta_i(\delta)|_{\delta=0} \cdot \nabla f(\Delta)|_{\Delta=\Delta_i(0)=0} \\
&= (0, \ldots, 0, n(x + \delta)^{n-1}|_{i\text{-th position}}, 0, \ldots, 0)|_{\delta=0} \cdot \gamma^*(x^n \one; y). \\
&= nx^{n-1} \cdot \gamma^*_i(x^n \one; y).
\end{align*}

\xpar{Non-differentiable case} If $f$ is not differentiable, the optimum may not be unique. Instead, we define the set $\cD = \arg \min_{\gamma, \theta \ge 0} \left \{ x^n \one^T \gamma + \one^T \theta  : (I - yA)\gamma + \theta \ge \one \right \}$. By Danskin’s theorem for pointwise minima, for any direction $h \in \Rbb^K$ we have that the directional derivative of $f$ is $f'(h) = \min_{(\gamma, \theta) \in \cD} h^T \gamma$. For the right derivative we have $h = +e_i$ and thus $f'(e_i) = \min_{(\gamma, \theta) \in \cD} \gamma_i$. Therefore:

\begin{align*}
\potentialimpact^+_i(x; y) = \lim_{\delta \downarrow 0} \frac{p^*_\cG(x^n \one + \Delta_i(\delta); y) - p^*_\cG(x^n \one; y)}{\delta} &= \lim_{\delta \downarrow 0} \frac{f([(x + \delta)^n - x^n] \cdot e_i) - f(0)}{\delta} \\
&= \lim_{\delta \downarrow 0} \frac{f([(x + \delta)^n - x^n] \cdot e_i) - f(0)}{(x + \delta)^n - x^n} \cdot \frac{(x + \delta)^n - x^n}{\delta} \\
&= \lim_{\delta \downarrow 0} \frac{(x + \delta)^n - x^n}{\delta} \cdot \lim_{t \downarrow 0} \frac{f(t \cdot e_i) - f(0)}{t} \\
&= nx^{n-1} f'(e_i) \\
& = nx^{n - 1} \min_{(\gamma, \theta) \in \cD} \gamma_i
\end{align*}

Similarly for the left derivative we have that:

\begin{align*}
    \potentialimpact^+_i(x; y) &  = \lim_{\delta \uparrow 0} \frac{f([(x + \delta)^n - x^n] \cdot e_i) - f(0)}{\delta} \\ & =  - n x^{n - 1} f'(-e_i) \\ & = - n x^{n - 1} \min_{(\gamma, \theta) \in \cD} (-\gamma_i) \\ & = n x^{n - 1} \max_{(\gamma, \theta) \in \cD} \gamma_i
\end{align*}

Subsequently,  $$\rei^+_\cG(x; y) = (n x^{n - 1}) \max_{i \in [K]} \min_{(\gamma, \theta) \in \cD} \gamma_i, \quad \text{and} \quad \rei^-_\cG(x; y) = (n x^{n - 1}) \max_{i \in [K]} \max_{(\gamma, \theta) \in \cD} \gamma_i$$

\subsubsection{Proof of \texorpdfstring{\cref{collorary:rei_katz}}{rei_katz}: Connection of the Potential impact of a product to Katz centrality}

We consider the primal LP and the dual LP: 

\begin{align*}
    \max_{\beta \ge \zero} \one^T \beta & \quad \text{s.t.} \quad \beta \le \one, \; \beta \le yA^T \beta + u \\
    \min_{\gamma, \theta \ge 0} u^T \gamma + \one^T \theta & \quad \text{s.t.} \quad (I - yA) \gamma + \theta \ge \one
\end{align*}

If $y < 1/\mu$ and $0 < x < (1 - \mu y)^{1/n}$ we know from \cref{theorem:general_ub} that the optimal primal solution is $\beta^* = (I - yA^T)^{-1} x^n$ and furthermore the derivative  of \cref{eq:potential_impact} exists. From strong duality, we have that for the optimal solution $(\gamma^*, \theta^*)$ of the dual, we have

\begin{align}
    \one^T \beta^* = x^n \one^T \gamma^* + \one^T \theta^* \label{eq:strong_dual} \\
    \gamma^* (I - yA) + \theta^* \ge \one \label{eq:strong_dual_feas} \\ 
    \gamma^*, \theta^* \ge \zero
\end{align}

Increasing $\theta^*$ in \cref{eq:strong_dual} would loosen \cref{eq:strong_dual_feas}, so at optimum we have $\theta^* = \zero$, which yields 

\begin{align}
     \one^T (I - yA^T)^{-1} \one = \one^T \gamma^* \label{eq:strong_dual_2}  \\
    \gamma^* (I - yA) \ge \one \nonumber
\end{align}

Further, due to the KKT conditions the constraint $ \gamma^* (I - yA) \ge \one$ is tight, that is $\gamma^* (I - yA) = \one$. If $y < 1/m$, the matrix $I - yA$ is invertible, which yields $\gamma^* = (I - yA)^{-1} \one$ which satisfies \cref{eq:strong_dual_2}. This yields the final constraints $y < \min \{ 1 / \mu, 1 / m \}$ and $0 < x < (1 - \mu y)^{1/n}$.

\newpage

\section{Generalizing Resilience to Heterogeneous and Correlated Failures} \label{app:generalized_resilience}

% \subsection{Extensions of the Resilience metric to heterogeneous number of suppliers and correlations} \label{sec:heterogeneities}

Thus far, for clarity and analytical tractability, we assumed that each product has $n$ suppliers and that these suppliers fail independently and identically (i.i.d.) with probability $x$. We now generalize the model to incorporate heterogeneities such as varying numbers of suppliers and correlated failure probabilities.

In this general model, each product $i$ has $n_i = |\cS(i)| = O(1)$ suppliers. We define the universe of suppliers as $\cU = \bigcup_{i \in \cK} \cS(i)$, with total size $N = |\cU|$. Each supplier $s \in \cU$ fails with probability $x_s = \Pr[X_s = 1]$.

Here, supplier failures are not assumed i.i.d., but instead follow a joint distribution $\nu$ with marginals $\{ x_s \}_{s \in \cU}$.

We define the resilience with respect to a joint distribution $\nu$ as the largest $x \in (0, 1)$ such that:
\emph{(i)} there are on average $\sum_{s \in \cU} x_{s} = x N$ failures, \emph{(ii)} at least $1 - \varepsilon$ of the products survive with high probability assuming that the set of failed suppliers ($F_\cS$) are distributed according to $\nu$ (i.e., $F_\cS \sim \nu$). The resilience of the graph is taken to be the resilience over the worst possible such joint distribution, i.e.: 
\begin{align} \label{eq:resilience_general-app}
    R_{\cG}(\varepsilon)  = \inf_{\nu} \;\; \sup_{x} \;\;\;\;  x & \quad\\ \text{s.t.} \;\;\;\;\;\;\;  & 0<x<1, \; \sum_{s \in \cU} x_{s} \le x N, \nonumber \\  &  \nu \text{ has marginals } \{ x_s \}_{s \in \cU}, & \nonumber \\ &  \text{and} \; \Pr_{F_\cS \sim \nu} \left [ F \ge \varepsilon K \right ] \le \frac {1} {K}. \nonumber 
\end{align}

If the marginals are deterministic (i.e., $x_s \in \{0, 1\}$), then $R_{\cG}(\varepsilon)$ is NP-Hard to compute and hard to approximate as well (see \Cref{app:generalized_resilience}).
Even when the marginals are fractional, numerically evaluating resilience is computationally intractable, as one has to search for all possible joint distributions and calculate the resilience for each joint distribution. Therefore, we are interested in efficiently computable upper and lower bounds. 

\subsection{The Bahadur Representation} \label{sec:general_resilience_lower_bound}

First note that the definition of resilience according to \cref{eq:resilience_general-app} considers the worst possible joint distribution; therefore, the generalized resilience value is upper-bounded by resilience in the case of i.i.d. failures, and the upper bound of \cref{theorem:resilience_graph_statistics} holds for the general definition of resilience. To devise a lower bound on resilience, one might consider maximizing $\ev {} {F}$ and using the Markov inequality to bound the probability of the event $\{ F \ge \varepsilon K \}$. However, such an optimization program has $O(K)$ variables and $O(2^N)$ constraints, making it impractical to solve.

In the sequel, we focus on low-dimensional approximations of the joint failure distribution $\nu$.

Our approach is based on the Bahadur decomposition~\citep{bahadur1961representation,yuan2021community}, which expresses the joint distribution $\nu$ as a product of marginals and multi-way correlation terms. 
According to the Bahadur representation, the probability that a subset $T \subseteq \cU$ of suppliers fails is: %the joint distribution $\nu$ of active suppliers can be written. 

{\small
\begin{align} \label{eq:bahadur_expansion}
    \nu(T) & = \prod_{s \in T} x_s \prod_{s \notin T} (1 - x_s) \left \{ \sum_{j = 1}^N \sum_{\{ s_1, \dots, s_j \} \in \binom {\cU} j} \rho_{s_1, \dots, s_j} (T) \prod_{j' = 1}^j \hat x_{s_j'}(T) \right \},
\end{align}}where $\hat x_{s_j} = \frac {\one \{ s_j \in T \} - x_{s_j}} {\sqrt {x_{s_j} (1 - x_{s_j})}}$, $\rho_{s_1, \dots, s_j}(T)$ denotes the $j$-th order correlation between $s_1, \dots, s_j$. 

From this expression, we can derive the failure probability for each product $i \in \cK$ as: $u_i = \sum_{T \subseteq \cU \setminus \cS(i)} \nu(\cS(i) \cup T)$. 

To obtain a tractable low-dimensional representation of $\nu$, we introduce structured assumptions that preserve interpretability and enable the computation of resilience. 
In the following, we give some simplified models that account for correlations between suppliers within each product, across products, and across all suppliers.

\xpar{Homogeneous failure probabilities and within-product correlations} The first simplification of \cref{eq:bahadur_expansion} assumes homogeneous failure probabilities and allows correlated failures among suppliers of the same product (but not between products), requiring only $O(N)$ parameters. 
This model accounts for failures within a product or a sector in the economy; for instance, failure of a supplier of a scarce raw material may increase the demand on alternative suppliers, elevating their risk of failure.

Throughout Appendix EC.8, the Bahadur representation is applied to the joint distribution of supplier-level failure indicators, not to the product-level indicators $Z_i$. Suppliers have fixed marginal failure probabilities $x_s$, and while correlations are introduced through the Bahadur coefficients of the supplier-level joint distribution. Product-level failure probabilities $u_i = \Pr [Z_i = 0]$ are therefore derived quantities, obtained by aggregating over correlated supplier failures, and are not fixed a priori.

Assuming $x_s = x$ for all $s \in \cU$ and that the $j$-th order correlation between suppliers of product $i$ is $\rho_{ij} \in [0, 1]$, we obtain the following expression for its failure probability:
\begin{align} \label{eq:correlated_suppliers}
   u_i = x^{n_i} + \sum_{j = 2}^{n_i} \binom {n_i} {j} \rho_{ij} (1 - x)^{j/2} x^{n_i - j/2}.
\end{align}
As expected, increasing $\rho_{ij}$ raises the failure probability $u_i$, thereby reducing resilience. When $\rho_{ij} = 0$, we recover the i.i.d. failure model with $u_i = x^{n_i}$.

When $\rho_{ij} = 1$, each product effectively behaves as if it has a single supplier, thereby network classes that are indexed by their number of products ($K$) maintain their asymptotic resilience classification (see \Cref{def:resilience-taxonomy}). This is no longer the case when correlations are introduced between products; then all networks are, in general, fragile.  

\xpar{Homogeneous failure probabilities with correlations between products but not within products} The second simplification considers the case of homogeneous failures and correlated products, but uncorrelated suppliers within a product, which requires $O(K)$ parameters. This case captures scenarios where natural disasters impact multiple sectors simultaneously, even when they lack direct supply-chain dependencies.

Thus, if we assume that all suppliers $s \in \cU$ fail with probability $x_s = x$ independently within each product, but products are correlated with each other, with the $j$-th order correlation coefficient being $\rho_j \in [0, 1]$, the product-level failure probability $u_i$ induced by correlated supplier failures, is given by:

\begin{align} \label{eq:correlated_products}
    u_i = \sum_{b = 0}^{K - 1} \binom {K - 1} {b} \left ( x^{n_i} \right )^{b + 1} \left ( 1 - x^{n_i} \right )^{K - 1 - b} \left \{ 1 + \sum_{j = 2}^K \sum_{r = 0}^{\min \{j, b + 1 \}} \binom j r \binom {K - j} {j - r} \rho_j \left (1 - x^{n_i} \right )^{r - j/2} \left  ( x^{n_i} \right )^{j/2 - r}  \right \}.
\end{align}

When $\rho_j = 1$, the failure of a single product is enough to trigger $F \ge \varepsilon K$ failed products for $\varepsilon \in (0, 1 - 1 / K)$. By the union bound, this happens with probability at most $\sum_{i = 1}^K x^{n_i} \le Kx$, and setting the probability of this event to be at most $1 / K$, yields a lower bound for the resilience which is $\Omegaeps{1/K}$. Regarding the upper bound, the probability that all products fail is at most $x$, and therefore the resilience is $\Oeps {1/K}$ in that case for any $\varepsilon$, thus making the network fragile.  Therefore the resilience is $\Theta_{\varepsilon} (1/K)$.

\xpar{Homogeneous supplier failure probabilities with joint within and between product correlations} 
The final simplification considers the case of homogeneous failures and correlations among all suppliers, which requires $O(N)$ parameters. This case models the scenario where a firm produces more than one product; for example, an automaker may produce both engines and chips in the same facility, so a localized event (e.g., a fire) affects multiple products simultaneously.

Extending \Cref{eq:correlated_products}, the marginal failure probability for product $i$ under full joint correlations, with the $j$-th order correlation coefficient given by $\rho_j \in [0, 1]$, becomes:
\begin{align} \label{eq:joint_correlations}
        u_i = \sum_{b = 0}^{N - n_i} \binom {N - n_i} {b} x^{b + n_i} \left ( 1 - x \right )^{N - n_i - b} \left \{ 1 + \sum_{j = 2}^N \sum_{r = 0}^{\min \{j, b + n_i \}} \binom j r \binom {N - j} {j - r} \rho_j \left (1 - x \right )^{r - j/2}  x^{j/2 - r}  \right \}.
\end{align}

Similarly to the case of \cref{eq:correlated_products}, when $\rho_j = 1$, the failure of a single product is enough to trigger $F \ge \varepsilon K$ failed products for $\varepsilon \in (0, 1 - 1 / K)$. By the union bound, this happens with probability at most $\sum_{i = 1}^K x^{n_i} \le Nx \le Kx \max_{i \in [K]} n_i $, and setting the probability of this event to be at most $1 / K$, yields a lower bound for the resilience, which is $\Omegaeps{1/(K \max_{i \in K} n_i)}$. Similarly to the above, the probability that all products fail is at most $x$, and therefore the resilience is $\Oeps {1/K}$ in that case for any $\varepsilon$, thus the network is fragile in that case. If $n_i = O(1)$, we again obtain the rate $\Theta_{\varepsilon} (1/K)$ for the resilience. 

\xpar{Bounds based on the Bahadur representation} To bound the resilience assuming simplified marginals $\{ u_i \}_{i \in [K]}$ (e.g. \cref{eq:joint_correlations,eq:correlated_suppliers,eq:correlated_products}), we can directly apply the results of \cref{theorem:resilience_graph_statistics} and \cref{theorem:lp_duality_resilience}, by inverting the corresponding marginals (see \cref{fig:inverse_marginal}):

\begin{proposition}[Bounding the resilience with correlations] \label{prop:correlation}
    Let $\cG$ be a network in which suppliers fail with a correlation vector $\rho \in [0, 1]^J$, where $J$ is the highest index for which the correlation is non-zero, and $n_i = n$, yielding marginal failure probabilities $u_i = u(x; \rho)$ that are monotonically increasing in $x$. If $y < 1/m$, then resilience satisfies: $$ u^{-1} \left ( \frac {\varepsilon} {\one^T \beta_{\cG}^{\katz}(y)} ; \rho \right ) \leq R_{\cG}(\varepsilon; y, \rho) \le u^{-1} \left ( \frac {(1 - \varepsilon) K} {2r^{3/2} + \sqrt {r \log K}} ; \rho \right ).$$ 
\end{proposition}

\begin{figure}
    \centering
    \includegraphics[width=\linewidth]{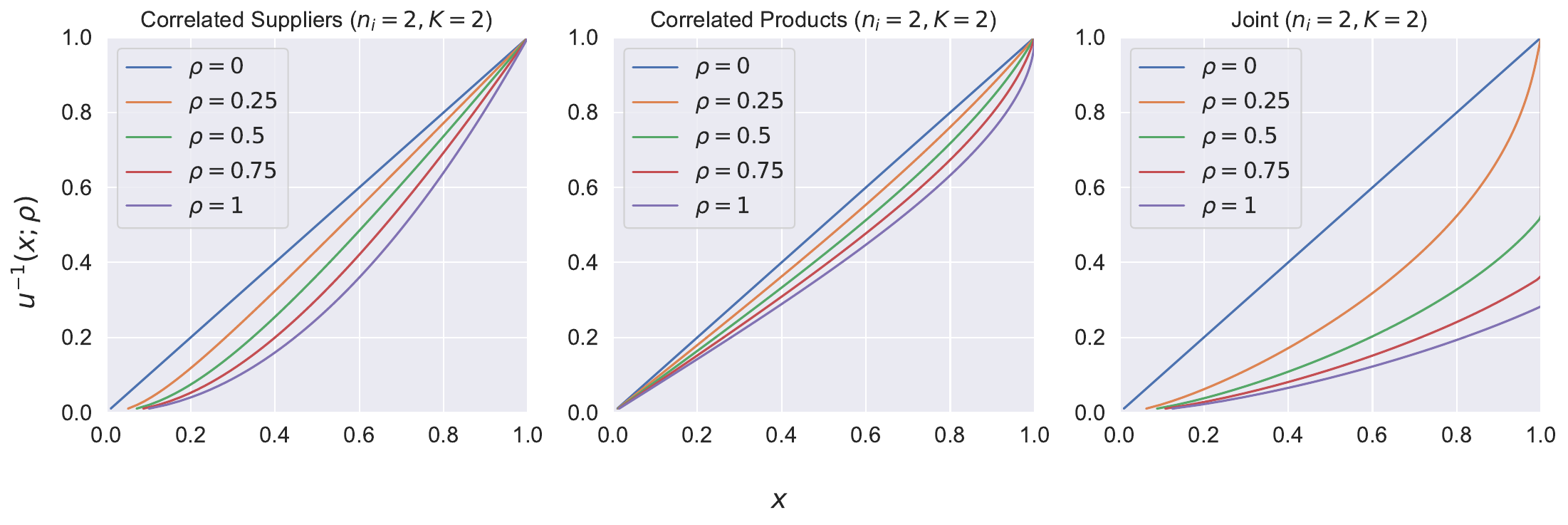}
    \caption{Value of inverted marginal $u^{-1}(x; \rho)$ -- which corresponds to a lower bound to the resilience metric -- for the cases studied in \cref{eq:correlated_products,eq:correlated_suppliers,eq:joint_correlations}. All correlations have been set equal to $\rho \in \{ 0.25, 0.5, 0.75, 1 \}$. The network is assumed to have $K = 2$ products and each product has $n_i = n = 2$ suppliers.}
    \label{fig:inverse_marginal}
\end{figure}

\subsection{Hardness with deterministic marginals} 

If supplier failures are deterministic, the hardness follows from the following decision problem, which is equivalent to calculating the resilience.

\begin{definition}[\textsc{Resilience-Deterministic}] 
    Given a production network $\cG$, an average number of supplier failures $x$, and a non-negative integer $f$, does there exist a deterministic joint distribution $\nu$ of supplier failures such that the number of failures is $f$?
\end{definition}

\begin{theorem} \label{theorem:resilience_deterministic_hardness}
    \textsc{Resilience-Deterministic} is NP-hard. 
\end{theorem}

% \begin{proof}

\xpar{Proof} The proof relies on a reduction from the \textsc{3-Set-Cover} problem, where the input consists of $c$ elements $V = \{ v_1, \dots, v_{c} \}$ and $r$ sets $S = \{ s_1, \dots, s_{r} \}$ where each set has cardinality 3, and a number $B$. The question is whether there is a collection of $B$ subsets that cover all the elements. 

To construct the reduction, we consider a bipartite production network with raw products $\cR$, labeled by $\{ s_1, \dots, s_{r} \}$, and finished goods $\cC$, labeled by $\{ v_1, \dots, v_{c}\}$, with $K = c + r$ products, where the left partition (raw products $\cR$) corresponds to the sets in \textsc{3-Set-Cover} and the right partition (consumer goods $\cC$) corresponds to the elements in \textsc{3-Set-Cover}. Each product/node has a single supplier ($n_i = 1$). Each raw product $s_i \in \cR$ is used to make three complex products such that for each complex product $v_j \in \cC$ we have $v_j \in s_i$ in the \textsc{3-Set-Cover} instance. We set $f = B + c$ and $x = B / K$. The reduction runs in polynomial time, creating a graph with $O(K)$ nodes and $O(K)$ edges. 

\noindent ($\implies$) Assume that there is a set cover $\mathcal J \subset S$ of size $B=|\mathcal J|$. Then we choose the raw products $J \subseteq \cR$ corresponding to $\mathcal J$ and the set $x_i = 1$ for the unique supplier of the product $i$ (so that it fails deterministically with probability $1$). $B$ of the raw products and all $c$ consumer goods fail. Therefore, the number of failures is now $B + c$.

\noindent ($\impliedby$) Take any supplier failure assignment with at least $B + c$ supplier failures. Then, if there exists a scenario with $B+c$ failed products, then it should necessarily include all $c$ consumer goods and $B$ of the raw products, whose corresponding subsets in $S$ constitute a solution to the \textsc{3-Set-Cover} problem.
% \end{proof} 
\qed

\subsection{Distribution constraints}

For brevity, we let $N = \sum_{i \in \cK} n_i$. The definition of \cref{eq:resilience_general-app} requires defining a distribution $\nu$ over the union of the suppliers. We let $\cU = \bigcup_{i \in \cK} \cS(i)$ be the universe of suppliers. We let $\nu: 2^{\cU} \to [0, 1]$ be the distribution, where $\nu(U)$ corresponds to the probability that at subset $U \subseteq \cU$ of the suppliers is active. We require the coupling to be non-negative and normalized: 

\begin{align} \label{eq:coupling_nonnegative_normalized}
    \nu \left ( U \right ) & \ge 0, \qquad \forall U \in 2^{\cU} \\
    \sum_{U \subseteq \cU} \nu(U) & = 1,
\end{align}

and respect the corresponding marginals, i.e.

\begin{align} \label{eq:coupling_marginals}
    \sum_{T \subseteq \cU : s \notin T} \nu(T \cup \{ s \}) \le x_s.
\end{align}

Finally, we impose the budget constraint, i.e. 

\begin{align} \label{eq:budget_constraint}
    0 & \le x_{s} \le 1 \qquad \forall s \in \cU \\
    \sum_{s \in \cU} x_{s} & \le x N.
\end{align}

In matrix notation, if $\bar {x}$ is the vector of marginals,

\begin{align} \label{eq:coupling_constraint}
    \nu & \ge \zero, \bar  x \ge \zero \\
    \bar x & - \one \le \zero, \Phi \nu - \bar x \le \zero, \one^T \bar x  \le x N, \one^T \nu \le 1
\end{align}

where $\Phi$ is a matrix such that $\phi_{T, s} = 1$ if and only if $s \notin T$. The number of variables needed to define $\nu$ is $O \left ( 2^{N} \right )$.

\subsection{LP formulation}

\xpar{Upper bound on $\ev {} {F}$} We extend \cref{eq:upper_bound_lp} to account for $\nu$. Specifically, we are interested in the upper bound on the number of failures for the worst possible joint distribution  $\nu$. The primal problem corresponding to this upper bound is as follows:

\begin{align} \label{eq:correlated_failures_primal}
    p^* = \max_{\beta, q, \nu \ge \zero} \quad & \one^T \beta \\
    \text{s.t.} \quad & \beta \le 1, (I - yA^T) \beta - \Psi \nu \le \zero  \\
    & \text{\cref{eq:coupling_constraint}}.
\end{align}

We define $\Psi$ as the $K \times 2^N$ matrix with elements $\psi_{i, T} = 1$ iff $T \subseteq \cU \setminus \cS(i)$ and 0 otherwise. Taking the dual yields, 

\begin{align} \label{eq:correlated_failures_dual}
    d^* = \min_{\gamma, \theta, \zeta, \kappa, \eta, \xi \ge \zero} & \one^T \theta + \one^T \zeta + \eta + \xi x N \\
    \text{s.t.} \quad & (I - yA) \gamma + \theta \ge \one 
    \nonumber \\
    & \zeta - \kappa + \one \xi \ge \zero \nonumber \\
    & - \Psi^T \gamma + \Phi^T \kappa + \one \eta \ge \zero \nonumber.
\end{align}

\xpar{Lower bound on $R_{\cG}(\varepsilon)$} Similarly to \cref{theorem:lp_duality_resilience}, if we take a feasible dual solution $(\tilde \gamma, \tilde \theta, \tilde \zeta, \tilde \kappa, \tilde \eta, \tilde \xi)$ to \cref{eq:correlated_failures_dual} with objective value $\tilde d$, we can show from weak duality that it suffices to set $\tilde d \le 1 - \varepsilon$ to get a lower bound on $R_{\cG}(\varepsilon)$. This implies that 

\begin{align}
    R_{\cG}(\varepsilon) \ge \frac {\varepsilon - \one^T \tilde \theta - \one^T \tilde \zeta - \tilde \eta} {\tilde \xi N}.
\end{align}

Therefore, a lower bound can be devised by taking the maximum possible value of the RHS, i.e.

\begin{align} \label{eq:lower_bound_resilience_correlated_full}
    \underline R_{\cG}(\varepsilon)  = \max_{\gamma, \theta, \zeta, \kappa, \eta, \xi}  \quad & \frac {\varepsilon - \one^T \theta - \one^T \zeta - \eta} {\xi N} \\
    \text{s.t.} \quad & \text{\cref{eq:correlated_failures_dual} holds}. \nonumber
\end{align}

It is easy to show that in optimality $\eta = 0$, $\theta = \zero$ and $\zeta = \zero$, we therefore have the following. 

\begin{align} \label{eq:lower_bound_resilience_correlated_simplified}
    \underline R_{\cG}(\varepsilon)  = \max_{\gamma, \kappa, \xi}  \quad & \frac {\varepsilon} {\xi N} \\
    \text{s.t.} \quad & (I - yA) \gamma \ge \one \nonumber,\; - \kappa + \one \xi \ge \zero, \;- \Psi^T \gamma + \Phi^T \kappa \ge \zero. \nonumber
\end{align}

% \subsection{Resulting Lower bound for the Resilience}

The results yield the following theorem: 

\begin{theorem} \label{theorem:general_resilience}
    If resilience is defined as in \cref{eq:resilience_general-app}, then a lower bound on resilience can be found by solving the following optimization problem with $O(K)$ variables and $O(2^N)$ constraints: 

    \begin{align*}
    \underline R_{\cG}(\varepsilon; y)  = \max_{\gamma, \kappa, \xi \ge \zero}  \quad & \frac {\varepsilon} {\xi N} \\
    \text{s.t.} \quad & (I - yA) \gamma \ge \one \nonumber,\; - \kappa + \one \xi \ge \zero, \;- \Psi^T \gamma + \Phi^T \kappa \ge \zero, \nonumber
    \end{align*} where $\Psi, \Phi$ are matrices given in \cref{app:generalized_resilience}. 

\end{theorem}

\subsection{Proof of \texorpdfstring{\cref{prop:correlation}}{prop:correlation}}\label{app:prop:correlation}

Since $n_i = n$ and the correlations are shared and given by the correlation vector $\rho$, we can write the probability that product $i$ fails as $u_i = u(x; \rho)$ for all $i \in [K]$. The function $u(x; \rho)$ is continuous and strictly increasing with $x$, since a higher supplier failure probability yields a larger product failure probability. Therefore, the inverse function $u^{-1}( \cdot ; \rho)$ exists and has the same type of monotonicity as $u (\cdot ; \rho)$ in $x$. 

\xpar{Lower bound} Following the proof of \cref{theorem:lp_duality_resilience} (cf. \cref{app:proof:theorem:lp_duality_resilience}), we have that for the lower bound $\underline R_{\cG}(\varepsilon)$ in the resilience it holds that 

    \begin{align}
    u \left ( \underline R_{\cG}(\varepsilon); \rho \right ) = \frac {\varepsilon} {\min_{\gamma \ge \zero, \gamma \neq \zero} \one^T  \gamma}, \qquad \text{s.t.} \qquad (I - yA)\gamma \ge \one,    
    \end{align}

    analogously to \cref{eq:lower_bound_resilience_dual}. By inverting $u$ we get that 

    \begin{align}
    \underline R_{\cG}(\varepsilon) = u^{-1} \left ( \frac {\varepsilon} {\min_{\gamma \ge \zero, \gamma \neq \zero} \one^T  \gamma}; \rho \right ), \qquad \text{s.t.} \qquad (I - yA)\gamma \ge \one.     
    \end{align} 
    
    We take the dual of the denominator, which corresponds to $\max_{\beta \ge \zero} \one^T \beta$ subject to $\beta \le y A^T \beta + \one$. If $y < 1/m$, from the main result of \cite{eisenberg2001systemic} (or the KKT conditions) we get that the optimal solution is $\beta_{\cG}^\katz (y)$. Similarly, for its dual, the optimal solution is $\gamma_{\cG}^\katz (y)$.

    \xpar{Upper bound} Similarly to the lower bound, we follow the proof of \cref{theorem:resilience_graph_statistics} (cf. \cref{app:theorem:resilience_graph_statistics}). Specifically, let $F_{\cR}$ correspond to the number of failures of the raw products. Let $x_{\cR} = \sup \{ x \in (0, 1) : \Pr [F_{\cR} \le (1 - \varepsilon) K] \ge 1 - 1 / K \}$. $x_{\cR}$ is an upper bound to $R_{\cG}(\varepsilon)$ since the event $\{ F \le (1 - \varepsilon) K \}$ implies $\{ F_{\cR} \le (1 - \varepsilon) K \}$. Moreover, by the Chernoff bound, we have that for any $\varepsilon' > 0$:

\begin{align*}
    \Pr \left [ \frac {F_{\cR}} {r} \le (1 + \varepsilon') u(x; \rho) \right ] \ge 1 - e^{-2 r (\varepsilon')^2}.
\end{align*}

Letting $(1 + \varepsilon') u(x ; \rho) = (1 - \varepsilon) K / r$, $e^{-2 r (\varepsilon')^2} = 1/K$ and solving for $x$ would produce an upper bound to $x_{\cR}$ and subsequently an upper bound to $R_{\cG}(\varepsilon)$. Solving the system for $\varepsilon'$ and $u$ similarly to the proof of \cref{theorem:resilience_graph_statistics} and then using the fact that $u(\cdot ; \rho)$ is invertible yields the following result. 

\begin{align*}
    \varepsilon' = \sqrt {\frac {\log(K)} {2 r}}  \qquad & \text{and} \qquad u(x; \rho) =  \frac {(1 - \varepsilon) K} {\sqrt {2} r^{3/2} + \sqrt {r \log K}} \\
    & \implies  x = u^{-1} \left ( \frac {(1 - \varepsilon) K} {\sqrt {2} r^{3/2} + \sqrt {r \log K}} ; \rho \right ).
\end{align*}

To determine conditions where the resilience goes to zero as $K \to \infty$ we focus on the denominator of the upper bound: First, the term $\sqrt {r \log K}$ is at most $\sqrt {K \log K} < K$ and cannot grow faster than $K$. Second, the term $\sqrt 2 r^{3/2}$ grows faster than $K$ as long as $r = \omega \left ( K^{3/2} \right )$.

\newpage

\section{Experiments with World Economy Input-Output Networks} \label{app:experiments_addendum}

We use country-level input-output tables from the World Input-Output Database (WIOD) to construct directed production networks for a set of national economies. In each country network, nodes represent sectors, and a directed edge indicates an input-output dependency between two sectors, so that the graph captures how upstream production disruptions can propagate through the domestic economy. The resulting networks are relatively dense, with 55 and 56 sectors in all the countries considered. For each country, we compute basic network statistics, including average degree, density, and minimum/maximum in- and out-degrees (\cref{tab:statistics-world}), and compare their resilience profiles using the same resilience metric developed in the main text. For each country network, we estimate the resilience curve $\hat R_G(\varepsilon)$ using the same Monte Carlo procedure as in \cref{sec:numerical-estimation-resilience}. We set the number of suppliers per product to $n=1$ and summarize each country's resilience by the area under the estimated resilience curve (AUC); see \cref{fig:wolrd-io}.

\begin{table}[H]
    \footnotesize
    \centering
    \begin{tabular}{lllllll}
    \toprule
         Country  & Size ($K$) & Avg. Degree & Density & Min/Max In-degree & Min/Max Out-degree  & AUC  \\
    \midrule 
         USA & 55 & 54.00 & 1.000 & 54 & 54 & 0.052 \\
         Japan & 56 & 45.33 & 0.824 & 0 -- 50 & 0 -- 50  & 0.058 \\
         G. Britain & 56 & 52.05 & 0.946 & 0 -- 54 & 0 -- 54  & 0.052  \\
         China & 56 & 37.89 & 0.688 & 0 -- 46 & 0 -- 46  & 0.078 \\
         Indonesia & 56 & 35.75 & 0.65 & 0 -- 46 & 0 -- 46  & 0.078 \\
         India & 56 & 29.57 & 0.537 & 0 -- 41 & 0 -- 43 & 0.095 \\
    \bottomrule
    \end{tabular}
    \caption{Network Statistics and resilience AUC for the world economies. The world economy input-output networks are directed and their edge density is computed as $\frac {|\cE(\cG)|} {K^2 - K}$.}
    \label{tab:statistics-world}
\end{table}
\vspace{-30pt}
\begin{figure}[H]
    \centering
    \subfigure[Estimating $R_{\cG}(\varepsilon)$\label{subfig:world_io_tables_resilience}]{\includegraphics[width=0.48\textwidth]{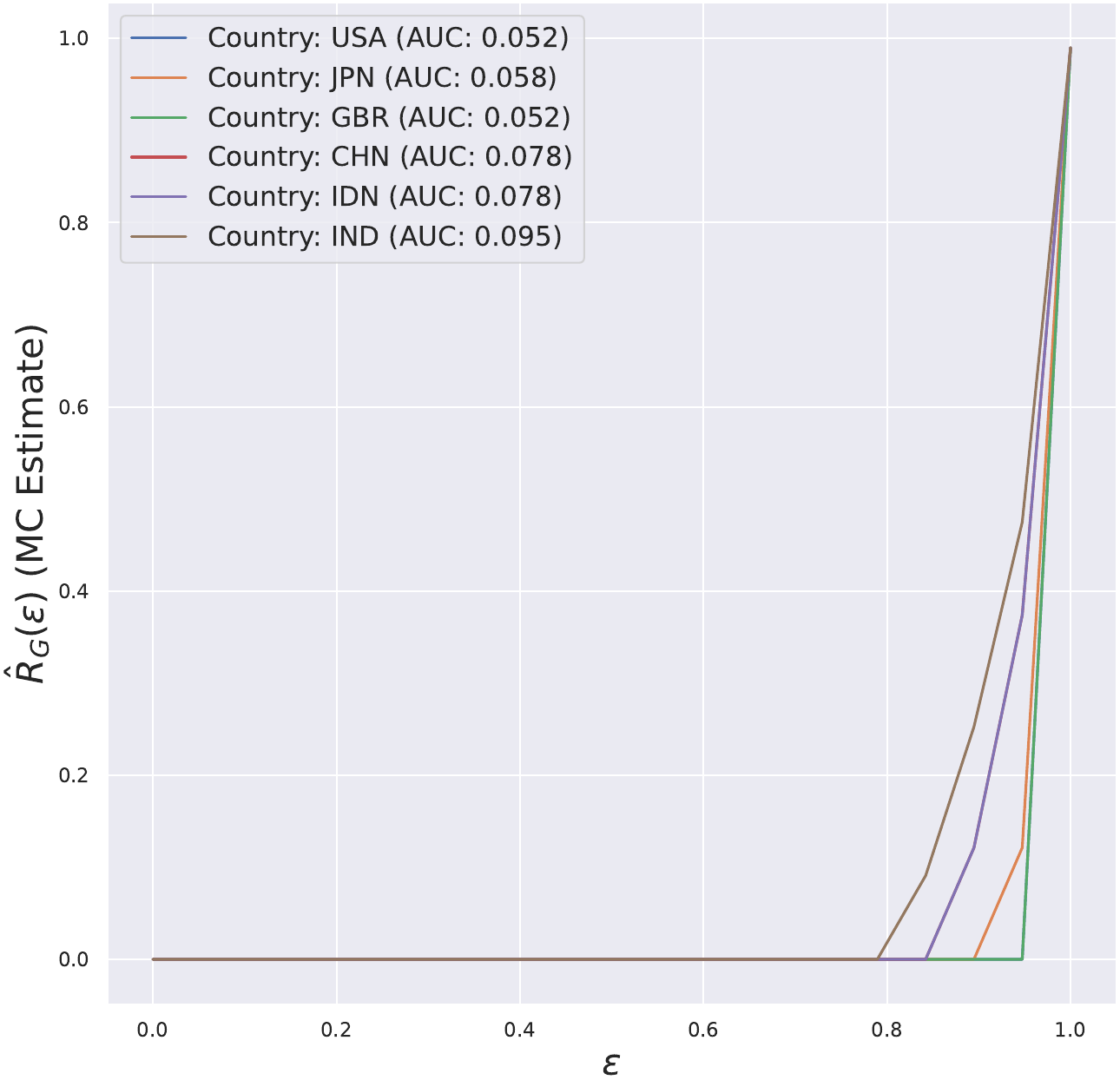}}
    % \subfigure[Optimal Interventions\label{subfig:world_io_tables_interventions}]{\includegraphics[width=0.48\textwidth]{figures/resilience_lb_vs_key_wiot.pdf}}
    \caption{World Economy Input-Output Networks. We set the number of suppliers for each product to $n = 1$.}
    \label{fig:wolrd-io}
\end{figure}

\end{document}